\documentclass[a4paper,12pt]{article}
\usepackage[a4paper,left=28mm,right=28mm,top=25mm,bottom=25mm]{geometry}

\usepackage[utf8]{inputenc}
\usepackage[english]{babel}
\usepackage{amssymb}
\usepackage{amsmath}
\usepackage{amsthm}
\usepackage{amssymb}
\usepackage{amsfonts}
\usepackage{bbm}
\usepackage{slashed}
\usepackage{mathtools}
\mathtoolsset{showonlyrefs}
\usepackage{enumitem}
\usepackage{stmaryrd}

\usepackage{hyphenat}
\usepackage{tocloft}
\newlength{\bibitemsep}
\newlength{\bibparskip}
\let\oldthebibliography\thebibliography
\renewcommand\thebibliography[1]{%
  \oldthebibliography{#1}%
  \setlength{\parskip}{\bibitemsep}%
  \setlength{\itemsep}{\bibparskip}%
}

\usepackage{color,hyperref}
\definecolor{darkblue}{rgb}{0.0,0.0,0.5}
\hypersetup{
  linkcolor  = darkblue,
  citecolor  = darkblue,
  urlcolor   = darkblue,
  colorlinks = true,
}
\makeatletter
\AtBeginDocument{
  \hypersetup{
    pdftitle = {\@title}
  }
}
\makeatother

\newcommand{\N}{\mathbb{N}}

\newcommand{\C}{\mathbb{C}}
\newcommand{\Z}{\mathbb{Z}}
\newcommand{\R}{\mathbb{R}}

\newcommand{\ri}{\mathrm{i}}
\newcommand{\rd}{\mathrm{d}}
\newcommand{\re}{\mathrm{e}}
\newcommand{\mass}{\mathrm{m}}

\newcommand{\Span}{\mathrm{span}}
\newcommand{\Ran}{\mathrm{Ran}\hspace{0.3mm}}
\newcommand{\Ker}{\mathrm{Ker}\hspace{0.3mm}}
\newcommand{\ret}{\mathrm{ret}}
\newcommand{\adv}{\mathrm{adv}}
\newcommand{\rout}{\mathrm{out}}
\newcommand{\rin}{\mathrm{in}}

\newcommand{\fin}{\mathrm{fin}}
\newcommand{\ras}{\mathrm{as}}
\newcommand{\fr}{\mathrm{fr}}
\newcommand{\rmod}{\mathrm{mod}}

\newcommand{\rD}{\mathrm{D}}
\newcommand{\rint}{\mathrm{int}}
\newcommand{\const}{\mathrm{const}}

\DeclareMathOperator{\sgn}{sgn}
\DeclareMathOperator{\Texp}{Texp}
\DeclareMathOperator{\aTexp}{\overline{T}exp}

\DeclareMathOperator{\T}{T}

\DeclareMathOperator{\aT}{\overline{T}}

\DeclareMathOperator{\sd}{sd}

\DeclareMathOperator*{\slim}{s-lim}
\DeclareMathOperator*{\wlim}{w-lim}
\DeclareMathOperator{\PV}{P}

\newcommand{\cL}{\mathcal{L}}
\newcommand{\cS}{\mathcal{S}}

\newcommand{\cD}{\mathcal{D}}
\newcommand{\cH}{\mathcal{H}}
\newcommand{\Fa}{\mathcal{F}}

\newcommand{\id}{\mathbbm{1}}

\renewcommand{\Re}{\mathrm{Re}\hspace{0.5mm}}
\renewcommand{\Im}{\mathrm{Im}\hspace{0.5mm}}
\newcommand{\supp}{\mathrm{supp}}
\newcommand{\zerop}{0}
\newcommand{\normord}[1]{:\mathrel{#1}:}
\newcommand{\mP}[1]{\frac{\rd^4 #1}{(2\pi)^4}}
\newcommand{\mH}[2]{\rd\mu_{#1}(#2)}
\newcommand{\mHm}[1]{\rd\mu_{\mass}(#1)}
\newcommand{\mHO}[1]{\rd\mu_{0}(#1)}
\newcommand{\Hm}{H_\mass}
\newcommand{\HO}{H_0}
\newcommand{\F}[1]{\tilde{#1}}

\newcommand{\ssl}{{\scriptscriptstyle <}}
\newcommand{\ssg}{{\scriptscriptstyle >}}
\newcommand{\feq}{\,\substack{\textrm{{\tiny(formally)}} \\=}\,}

\newtheorem{thm}{Theorem}
\newtheorem{dfn}[thm]{Definition}

\newtheorem{lem}[thm]{Lemma}

\newtheorem{cnj}[thm]{Conjecture}

\numberwithin{equation}{section}
\numberwithin{thm}{section}

\title{Infrared problem in perturbative quantum field theory}
\author{
Pawe{\l} Duch
\\
Max-Planck Institute for Mathematics in the Sciences\\
Inselstr. 22, 04103 Leipzig, Germany\\
Institut f\"ur Theoretische Physik, Universit\"at Leipzig\\
Br\"uderstr.\ 16, 04103 Leipzig, Germany\\
pawel.duch@mis.mpg.de
}
\date{\today}

\begin{document}

\maketitle

\begin{abstract}
We propose a mathematically rigorous construction of the scattering matrix and the interacting fields in models of relativistic perturbative quantum field theory with massless fields and long-range interactions. We consider quantum electrodynamics and a certain model of interacting scalar fields in which the standard definition of the scattering matrix is not applicable because of the infrared problem. We modify the Bogoliubov construction using the ideas of Dollard, Kulish and Faddeev. Our modified scattering matrix and modified interacting fields are constructed with the use of the adiabatic limit which is expected to exist in arbitrary order of perturbation theory. In the paper we prove this assertion in the case of the first- and the second-order corrections to the modified scattering matrix and the first-order corrections to the modified interacting fields. We study the physical properties of our construction. We conclude that the electrons and positrons are always surrounded by irremovable clouds of photons. Moreover, the physical energy-momentum operators do not coincide with the standard ones and their joint spectrum does not contain the mass hyperboloid.
\end{abstract}

\newpage
\setcounter{tocdepth}{2}
\tableofcontents
\newpage


\section{Introduction}

Quantum electrodynamics (QED) is one of the best tested theory known in physics. The success of QED is to large extent based on very accurate theoretical predictions of the anomalous magnetic moments of leptons and transition frequencies between various energy levels in light hydrogenlike atoms. In contrast to other sectors of the standard model the scattering experiments do not constitute the most stringent tests of QED and, as a matter of fact, the state of the art theoretical description of the scattering processes in QED is still far from satisfactory. One of the long-standing fundamental open problems in QED is the definition of the scattering operator. Because of the long-range character of interactions mediated by massless photons the standard method of constructing this operator is plagued by infrared (IR) divergences~\cite{bloch1937note}. In view of these problems, one usually abandons the definition of the scattering matrix and computes only the so-called inclusive cross sections~\cite{yennie1961infrared,weinberg1965infrared}. Although, this pragmatic approach is sufficient for most applications, it is not fully satisfactory. The physical interpretation of the inclusive cross sections is obscure since their construction relies on the use of charged states with sharp momenta which do not exist in the theory without the IR regulator. Moreover, the inclusive cross sections provide only partial information about the scattering process. There are processes that cannot be described without using the scattering matrix. An example of such a process is a shift of a trajectory of a charged particle moving in a weak low-frequency electromagnetic field~\cite{staruszkiewicz1981gauge,herdegen2012infrared}. It would be also challenging to study a quantum analogs of the memory effects predicted in classical electrodynamics~\cite{bieri2013electromagnetic,garfinkle2017memory,strominger2017lectures} without having a full scattering theory in QED at hand. The inclusive cross sections are completely insensitive to the IR degrees of freedom. Yet such degrees of freedom can be of physical importance. Indeed, Hawking, Perry and Strominger~\cite{hawking2016soft} have argued recently that it is possible to solve the black hole information paradox using the IR degrees of freedom of the gravitational field. The significance of the soft degrees of freedom has been also recently observed by a number of authors who studied the relation between the Weinberg soft photon theorem~\cite{weinberg1965infrared} and the so-called large gauge transformations~\cite{gabai2016large,campiglia2015asymptotic} (a more complete list of references can be found in the lecture notes \cite{strominger2017lectures}). The recent revival of interest in the IR degrees of freedom calls for a construction of the scattering matrix in QED in which these degrees of freedom are properly taken into account. 

A strategy for the construction of a IR-finite scattering matrix in QED was suggested long time ago by Kulish and Faddeev~\cite{kulish1970asymptotic}. Their construction is based on the Dollard method~\cite{dollard1964asymptotic} which was originally formulated for the Coulomb scattering in quantum mechanics. The basic idea is to define a modified scattering matrix by comparing the full dynamics of the system with some simple but nontrivial reference dynamics, called the Dollard dynamics. The Kulish-Faddeev construction was reformulated and further developed e.g. in~\cite{papanicolaou1976infrared,jauch1976theory} and more recently in \cite{gomez2016asymptotic,kapec2017infrared,hirai2019dressed}. However, the method of Kulish and Faddeev has never been put on firm mathematical grounds and tested at least in low orders of perturbation theory. In fact, up to our knowledge, a rigorous construction of the scattering matrix in relativistic QED which does not suffer from the ultraviolet (UV) or IR problem has never been given. In the present paper, we attempt to formulate such a construction using the ideas of Kulish, Faddeev, Dollard and insights obtained more recently by Morchio and Strocchi~\cite{morchio2016infrared,morchio2016dynamics} in their analysis of a toy model of QED. Our modified scattering matrix is constructed with the use of the adiabatic limit as in the method developed by Bogoliubov~\cite{bogoliubov1959introduction}. Since QED most likely does not make sense outside perturbation theory~\cite{landau1955point} we formulate our proposal in the perturbative setting. We stress that QED is a well-defined model of quantum field theory (QFT) whose perturbative construction is under full control. What is not understood is the large-time behavior of its dynamics. The main model of our interest is QED. However, we consider also a certain model of interacting scalar fields which we call the scalar model. In principle our method should be also applicable to other models with long-range interactions in the four-dimensional Minkowski space which, like QED and the scalar model, have the interaction vertices containing exactly one massless and two massive fields.

In the paper we propose a method of constructing the following objects in QED and the scalar model:
\begin{itemize}
 \item the Hilbert spaces of asymptotic outgoing and incoming states $\hat{\cH}_{\rout/\rin}(\varrho)$ parameterized by sector measures $\varrho$ (a sector measure is a signed measure on the mass hyperboloid $\Hm$ consisting of four-momenta $p$ such that $p^2=\mass^2$ and $p^0>0$),
 \item a strongly-continuous representation of the translation group $\hat U (a)=\re^{\ri\hat P\cdot a}$ acting in $\hat{\cH}_{\rin/\rout}(\varrho)$, with generators $\hat P^\mu$ satisfying the relativistic spectrum condition,
 \item the translationally-invariant scattering matrix $\hat S:\,\hat{\cH}_{\rin}(\varrho)\to \hat{\cH}_{\rout}(\varrho)$ (the inclusive cross section constructed with the use of $\hat S$ formally coincides with the one obtained using the standard procedure), 
 \item the translationally-covariant retarded and advanced interacting fields $\hat C_{\ret/\adv}(x):\,\hat{\cH}_{\rin/\rout}(\varrho)\to \hat{\cH}_{\rin/\rout}(\varrho)$, where $C$ is a (gauge invariant) polynomial in the basic fields and their derivatives (the interacting fields are local, satisfy the field equations and their Wightman and Green functions coincide with the standard ones).
\end{itemize}
We investigate various properties of the above objects. Below we summarize the results in the case of QED and comment on the case of the scalar model. The measure $\varrho$ is related to the long range tail of the electromagnetic field (the flux of the electric field) 
\begin{equation}\label{eq:long_rang_tail_intro}
 \lim_{r\to\infty} r^2\,\hat F^{\mu\nu}_{\ret/\adv}(x+rn) = -e \int_{H_\mass}\rd\varrho(p)\,
 f^{\mu\nu}(p/\mass;n),
\end{equation}
where $e$ is the elementary charge, $n$ is a normalized spatial four-vector and $f^{\mu\nu}(v;\cdot)$ is the Coulomb field of a unit point-charge moving with constant velocity $v$. Since the long range tail is a classical observable~\cite{buchholz1986gauss} the spaces $\hat{\cH}_{\rin/\rout}(\varrho)$ with different $\varrho$ correspond to different superselection sectors. By the Gauss law $\varrho(\Hm)$ is equal to the total electric charge\footnote{In the paper, by the electric charge we always mean the difference between the number of electrons and positrons.}. Thus, we demand that $\varrho(\Hm)$ is an integer (there is no such constraint in the scalar model). The spaces $\hat{\cH}_{\rin/\rout}(0)$ are called the vacuum sector. They contain a unique vector $\hat\Omega$ such that $\hat P^\mu\hat\Omega=0$, called the vacuum. The photon creation and annihilation operators are denoted by\footnote{$A^\#$ stands for the operator $A$ or its hermitian conjugate. $\hat a^\#(k)$ are (unbounded) operators when smeared with Schwartz test functions.}
\begin{equation}
 \hat a^\#(k):\,\hat{\cH}_{\rin/\rout}(\varrho)
 \to
 \hat{\cH}_{\rin/\rout}(\varrho).
\end{equation}
They satisfy the standard commutation relations, have the energy-momentum transfer on the lightcone and are obtained as the incoming/outgoing LSZ limits of the retarded/advanced interacting electromagnetic field (the massless scalar field in the case of the scalar model). They add or subtract one photon with four-momentum $k$. We also define creators and annihilators of the electrons
\begin{equation}
 \hat b^\#(\eta,\varrho_1;p):\,\hat{\cH}_{\rin/\rout}(\varrho) 
 \to
 \hat{\cH}_{\rin/\rout}(\varrho\pm\varrho_1)
\end{equation}
satisfying the standard (anti-)commutation relations. They depend on the choice a real-valued profile $\eta\in\cS(\R^4)$, $\int\rd^4 x\,\eta(x)=1$, and some measure $\varrho_1$ on $\Hm$, $\varrho_1(\Hm)=1$ (in the case of the scalar model all particles are chargeless and we can set $\varrho_1=0$). The operators $\hat b^\#(\eta,\varrho_1;p)$ do not have the energy-momentum transfer on the mass hyperboloid and do not commute with $\hat a^\#(k)$. They add or subtract one massive particle together with a coherent cloud of photons surrounding it and depending on $\eta$ and $\varrho_1$. For example, for $\hat\Omega\in\cH_{\rin/\rout}(0)$ the state $\hat b^*(\eta,\varrho_1;p)\hat\Omega \in \hat{\cH}_{\rin/\rout}(\varrho_1)$ is an (improper) eigenvector of the photon annihilation operator $\hat a(k)$ with eigenvalue 
\begin{equation}
 e\,\F{\eta}(k)\,\varepsilon(k)\cdot\left(\frac{p}{p\cdot k} - \int_{\Hm} \rd\varrho_1(p')\, \frac{p'}{p'\cdot k}\right),
\end{equation}
where $\F{\eta}$ is the Fourier transform of $\eta$ and $\varepsilon_\mu(k)$ is the polarization vector of the photon annihilated by $\hat a(k)$. There is no mass hyperboloid in the joint spectrum of the energy-momentum operators $\hat P^\mu$ regardless of the choice of the sector.

It is convenient to define the following Hilbert spaces containing sectors with all possible eigenvalues of the electric charge operator $Q$
\begin{equation}
 \hat\cH_{\rout/\rin}(\hat\varrho) := 
 \bigoplus_{q\in\Z} \hat\cH_{\rout/\rin}(\varrho_0+q\varrho_1),
\end{equation}
where $\hat\varrho=\varrho_0+Q\varrho_1$ and $\varrho_0$ and $\varrho_1$ are fixed measures on $\Hm$ such that $\varrho_0(\Hm)=0$ and $\varrho_1(\Hm)=1$ (in the scalar model $Q=0$, $\varrho_0$ need not satisfy the condition $\varrho_0(\Hm)=0$ and $\hat\cH_{\rout/\rin}(\hat\varrho)=\hat\cH_{\rout/\rin}(\varrho_0)$). Since it is difficult to discriminate experimentally sectors with the same electric charge but different long-range tail of the electromagnetic field for most practical application one can restrict attention to the spaces $\hat\cH_{\rout/\rin}(\hat\varrho)$ with one fixed $\hat\varrho$. Our standard choice in the case of the scalar model is $\hat\varrho=0$, which corresponds to the vacuum sector. In the case of QED the standard choice is $\hat\varrho = Q \varrho_{\mass \mathrm{v}}$, where $\mathrm{v}$ is some four-velocity and $\varrho_p$ is the Dirac measure at $p\in\Hm$. The corresponding spaces $\cH_{\rout/\rin}(\hat\varrho)$ contain the super-selection sectors in which the flux of the electric field is independent of the spatial direction in the reference frame of the observer moving with four-velocity~$\mathrm{v}$.

The asymptotic spaces $\hat{\cH}_{\rin/\rout}(\hat\varrho)$ are isomorphic to the standard Fock space $\cH$. However, the isomorphism is non-canonical and depends on the choice of a profile $\eta$. In particular, the physical photon creators and annihilators are not identified with the standard creation and annihilation operators of photons acting in the Fock space. We exploit this isomorphism in our construction. We first define a family of (modified) scattering operators $S_\rmod(\eta)$ and interacting fields $C_{\ret/\adv}(\eta;x)$ in the standard Fock space $\cH$. The operators $S_\rmod(\eta)$ and $C_{\ret/\adv}(\eta;x)$ depend on $\eta$ but satisfy certain compatibility condition. Using this compatibility condition we prove that the corresponding operators $\hat S$ and $\hat C_{\ret/\adv}(x)$ acting in $\hat{\cH}_{\rin/\rout}(\hat\varrho)$ are independent of choice a profile $\eta$ used to identify $\cH$ and $\hat{\cH}_{\rin/\rout}(\hat\varrho)$ (we keep the measure $\hat\varrho$ determining the superselection sector fixed).

Our modified scattering matrix corresponds roughly to the following expression
\begin{equation}\label{eq:S_mod_formal_intro}
 S_\rmod\!\feq\!\!\lim_{\substack{t_1\to-\infty\\t_2\to+\infty}}\! 
 \aTexp\left(\ri \int_0^{t_2}\! \rd t\, H_\rD(t)\right) \re^{-\ri (t_2-t_1) H} \aTexp\left(\ri \int_{t_1}^0\! \rd t\, H_\rD(t)\right),
\end{equation}
where $H$ is the full Hamiltonian of the system and $H_\rD(t)$ is an appropriately chosen time-dependent Dollard Hamiltonian. The above formula makes sense in quantum mechanics where it can be used to define the scattering matrix for systems of particles influenced by long-range potentials such as the Coulomb potential~\cite{dollard1964asymptotic,derezinski1997scattering}. In order to give meaning to expression~\eqref{eq:S_mod_formal_intro} in perturbative QFT we use the Bogoliubov method~\cite{bogoliubov1959introduction}. We first construct the modified scattering matrix with the adiabatic cutoff
\begin{equation}\label{eq:intro_S_mod_g}
 S_\rmod(\eta,g)=R(\eta,g)S^\ras_{\rout}(\eta,g) S(g) S^\ras_{\rin}(\eta,g)R(\eta,g)^{-1}.
\end{equation}
It is defined in the standard Fock space (in the case of QED we use the BRST framework) and depends on the switching function $g\in\cS(\R^4)$ and the profile $\eta$. The switching function plays the role of the IR regulator and is also called the adiabatic cutoff. The Bogoliubov operator $S(g)$ is formally defined by
\begin{equation}\label{eq:intro_S}
 S(g) := \Texp\left(\ri e \int\rd^4 x \, g(x)\mathcal{L}(x)  \right),
\end{equation}
where $\Texp$ is the time-ordered exponential and the RHS of the above equation is interpreted as a formal power series in the coupling parameter $e$. The interaction vertex $\cL$ is expressed in terms of the free local fields defined in the Fock space. These fields play an auxiliary role in the construction and do not coincide with the asymptotic LSZ fields of the interacting retarded/advanced fields that we construct. The Bogoliubov operator is the scattering operator in an unphysical theory in which the coupling constant $e$ is replaced with the function $e\,g(x)$. In order to solve the familiar UV problem in the construction of this operator we use the Epstein-Glaser method~\cite{epstein1973role,brunetti2000microlocal,hollands2002existence} (the use of the Epstein-Glaser method is not essential in our construction). The operators $S^\ras_{\rout/\rin}(g)$ are called the Dollard modifiers and up to finite self-energy renormalization are given by
\begin{equation}\label{eq:D_modifiers_Texp}
 S^\ras_{\rout/\rin}(g)=\aTexp\left( -\ri e \int \rd^4 x ~ g(x) \mathcal{L}_{\rout/\rin}(x) \right), 
\end{equation}
where $\mathcal{L}_{\rout}(x)$ and $\mathcal{L}_{\rin}(x)$ are the asymptotic outgoing and incoming vertices and $\aTexp$ is the anti-time-ordered exponential. The asymptotic vertices describe the emission or absorption of a photon by an electron or positron whose momentum is unchanged in this process. They are given by some non-local functionals of fields which can be expressed in a simple way in terms of the creation and annihilation operators. Finally, the operators $R(\eta,g)$ and $R(\eta,g)^{-1}$, which depend implicitly on the sector measure $\hat\varrho$, are used to implement a coherent transformation to sectors with desired long-range tail of the electromagnetic field (the massless scalar field in the case of the scalar model). There is no UV problem in the construction of the operators $S^\ras_{\rout/\rin}(g)$ and $R(\eta,g)$.

The physical scattering matrix $S_\rmod(\eta)$ is defined by the following adiabatic limit
\begin{equation}\label{eq:intro_strong_S_mod_limit}
 (\Psi|S_\rmod(\eta)\Psi') :=  \lim_{\epsilon\searrow 0}\, (\Psi|S_\rmod(\eta,g_\epsilon)\Psi'),
\end{equation}
where the one-parameter family of the switching functions $g_\epsilon$, $\epsilon\in(0,1)$, is defined by setting $g_\epsilon(x)=g(\epsilon x)$ with $g\in\cS(\R^4)$ such that $g(0)=1$. The states $\Psi$ and $\Psi'$ have wave functions that vanish when momenta of the electrons or positrons are sufficiently close to one another. The modified interacting fields $C_{\ret/\adv}(\eta)$ can be constructed in a similar way. It is expected that the adiabatic limit~\eqref{eq:intro_strong_S_mod_limit} exists in arbitrary order of perturbative expansion in the coupling parameter $e$ in both QED and the scalar model. In the paper we prove this claim up to the second order of perturbation theory. It is plausible that one could show the existence of the adiabatic limit~\eqref{eq:intro_strong_S_mod_limit} in all orders of perturbation theory by adapting the method of \cite{duch2018strong,epstein1976adiabatic}, which was used to construct the scattering matrix in purely massive models. Note that on the formal level the existence of the limit~\eqref{eq:intro_strong_S_mod_limit} is equivalent to the factorization of the IR divergences. Such a factorization was established at the physics level of rigor in~\cite{yennie1961infrared,weinberg1965infrared} and is usually taken for granted in the physics  literature~\cite{weinberg1995quantum,peskin1995introduction}. Thus, the existence of the adiabatic limit our modified scattering matrix is well-motivated on physical grounds. Apart from a lack of the rigorous proof of this fact our analysis of the IR structure of QED and the scalar model is complete and complies with the results obtained previously in the axiomatic framework or simplified models. 

The plan of the paper is as follows. In Sec.~\ref{sec:perturbation} we outline to the Epstein and Glaser approach to perturbative QFT. In Sec.~\ref{sec:models} we introduce the models which are investigated in the paper. In Sec.~\ref{sec:infrared} we discuss the IR problem in the perturbative construction of the scattering matrix in QED and the scalar model. In Sec.~\ref{sec:modified} we define the modified scattering matrix and the modified interacting fields in the scalar model and QED. Sec.~\ref{sec:low_orders} contains a construction of the first and second order corrections to the modified scattering matrix. The first order corrections to the modified interacting fields are presented in Sec.~\ref{sec:fields}. In Sec.~\ref{sec:energy_momentum} we investigate the translational covariance of our construction. In Sec.~\ref{sec:physical_interpretation} we define the spaces of asymptotic states, give their physical interpretation and present a method of constructing the inclusive cross section using our scattering matrix. The paper closes with a summary. Appendix~\ref{sec:value_dist} contains a useful theorem about the value of a distribution at a point defined with the use of the adiabatic limit. In Appendix~\ref{sec:BCH_Magnus} we recall the BCH formula and the Magnus expansion. In Appendix~\ref{sec:long_rang_tail} we define the long-range tail of a field. In Appendix~\ref{sec:LSZ} we discuss the LSZ limit of fields. In Appendix~\ref{sec:propagators} we give explicit expressions for various propagators used in the paper.

\noindent {\it Notation:}
\begin{itemize}\setlength\itemsep{0em}
\item The Minkowski spacetime is identified with $\R^4$. It is equipped with the inner product given by $x\cdot y = g_{\mu\nu} x^\mu x^\nu = x^0 y^0 - x^1 y^1-x^2 y^2-x^3 y^3$. If $x=(x^0,x^1,x^2,x^3)\in\R^4$ is a four-vector, then $\vec x=(x^1,x^2,x^3)\in\R^3$.
\item The future and past light cones in the Minkowski spacetime are denoted by $V^\pm=\{p\in\R^4\,:\, p^2 > 0,~\pm p^0> 0\}$. Their closures are denoted by $\overline{V}^\pm$.
\item The invariant measure on the mass shell $\Hm:=\{p\in\R^4\,:\,p^2=\mass^2,\,p^0\geq 0\}$ is denoted by $\mHm{p} := \frac{1}{(2\pi)^3} \rd^4 p\, \theta(p^0) \delta(p^2-\mass^2)$, $\mass\geq 0$. We have $\HO=\partial \overline{V}^+$. We set $E_\mass(\vec p)=\sqrt{\mass^2+|\vec p|^2}$. We say that $v\in\R^4$ is a four-velocity iff $v\in H_1$.
\item The Heaviside theta function is denoted by $\theta$.
\item The space of test functions with compact support and Schwartz functions on $\R^N$ are denoted by $C^\infty_{\textrm{c}}(\R^N)$ and $\cS(\R^N)$, respectively. 
\item  Let $t\in\cS'(\R^N)$ be a Schwartz distribution and $g\in\cS(\R^N)$. We use the notation $\int \rd^N x \, t(x) g(x)$ for the pairing between distributions and test functions.
\item The Fourier transform of the Schwartz distribution $t\in\cS'(\R^N)$ is denoted by~$\tilde{t}$. For any $g\in\cS(\R^N)$ it holds $\tilde{g}(q):=\int\rd^N x\, \exp(\ri q \cdot x) g(x)$ and $g(x)=\int\frac{\rd^N q}{(2\pi)^N}\, \exp(-\ri q \cdot x) \tilde{g}(q)$.
\item The positive-definite inner product in the (Krein-)Fock space is denoted by $\langle\cdot|\cdot\rangle$ and the corresponding norm by $\|\cdot\|$. The covariant non-degenerate inner-product, which need not be positive definite, is denoted by $(\cdot|\cdot)$. In the case of models without vector fields the above inner products coincide. The hermitian conjugate of an operator $A$ with respect to the positive-definite/covariant inner product is denoted by $A^\dagger$ and $A^*$, respectively. Moreover, $A^\#$ stands for $A$ or $A^*$. 
\item Let $\mathcal{Y}_1$ and $\mathcal{Y}_2$ be complex-linear spaces. The space of linear operators $\mathcal{Y}_1\to\mathcal{Y}_2$ is denoted by $L(\mathcal{Y}_1,\mathcal{Y}_2)$. We write $L(\mathcal{Y},\mathcal{Y})=L(\mathcal{Y})$. The space of complex-valued anti-linear functionals on $\mathcal{Y}$ is denoted by $\mathcal{Y}^*$. The space of sesquilinear forms on $\mathcal{Y}$ can be identified with $L(\mathcal{Y},\mathcal{Y}^*)$. 
\item The symbols $\slim$, $\wlim$ denote the limits in the strong and weak operator topologies.
\item The ring of formal power series in the coupling constant $e$ with coefficients in the ring $R$ will be denoted by $R\llbracket e\rrbracket$. The coefficient of order $e^n$ of a formal power series $a\in R\llbracket e\rrbracket$ is denoted by $a^{[n]}$.
\item We denote the electric charge by $-eQ$, where $e>0$ is the elementary charge. The free BRST charge is denoted by $Q_{\mathrm{BRST}}$. The free ghost charge is denoted by $Q_{\mathrm{gh}}$. The operator $\rho(p)$ is defined by Eqs.~\eqref{eq:def_charge_dist_momentum_scalar} and \eqref{eq:def_charge_dist_momentum} in the scalar model and QED, respectively. We set $\rd\rho(p):=\rho(p)\,\mHm{p}$.
\end{itemize}
\section{Outline of Epstein-Glaser approach}\label{sec:perturbation}

In this section we outline a construction of interacting models of perturbative quantum field theory in the Epstein-Glaser framework. More detailed exposition can be found in ~\cite{dutsch2018from,epstein1973role,scharf2014,scharf2016gauge,duch2017massless}. In Sec.~\ref{sec:bogoliubov} we recall the Bogoliubov method of constructing the scattering operator and the interacting fields. The precise definition of the adiabatic limit is given in Sec.~\ref{sec:adiabatic}.

\subsection{Wick products}\label{sec:Wick}

In perturbative QFT the interacting models are built with the use of the free fields. It is useful to consider first the algebra $\Fa$ of symbolic fields which is by definition a commutative algebra generated by symbols corresponding to the components $A_i$ of fields used in the definition of a given model and their derivatives $\partial^\alpha A_i$, where $\alpha$ is a multi-index. To every symbolic field $B\in\Fa$ we associate the Wick polynomial $\normord{B(x)}$ which is an operator-valued Schwartz distribution~\cite{wightman1965fields,streater2000pct} on a~suitable domain $\cD_0$ in the Fock Hilbert space $\cH$ with the positive-definite inner product $\langle\cdot|\cdot\rangle$. Given a domain $\cD\subset\cH$ let $L(\cD)$ be the space of linear maps $\cD\to\cD$. By an operator-valued Schwartz distribution on $\cD$ we mean a~map $T:\cS(\R^N)\to L(\cD)$, such that for any $\Psi,\Psi'\in\cD$ the map
\begin{equation}\label{eq:op_valued_dist}
 \cS(\R^N)\ni g \mapsto \langle\Psi|\int \rd^N\!x\, g(x)\, T(x)\Psi'\rangle \in\C
\end{equation}
is a~Schwartz distribution. The space of operator-valued Schwartz distributions on $\cD$ is denoted by $\cS'(\R^N,L(\cD))$. If $T(x)\in\cS'(\R^{N},L(\cD))$ and $T'(x')\in\cS'(\R^{N'},L(\cD))$, then $T(x)T'(x')\in\cS'(\R^{N+N'},L(\cD))$, by the nuclear theorem. The vacuum state in the Fock space $\cH$ is denoted by $\Omega$. We assume that the Fock space is equipped with a~non-degenerate sesquilinear inner product $(\cdot|\cdot)$ which in general need not be positive definite. The hermitian conjugation and the notion of unitarity are defined with respect to this product. The unitary representation of the inhomogeneous $SL(2,\C)$ group (the covering group of the Poincar{\'e} group), denoted by $U(a,\Lambda)$, where $a\in\R^4$ and $\Lambda\in SL(2,\C)$ is a~Lorentz transformation, is defined in the Fock space in the standard way. In the case of pure translations we write $U(a)\equiv U(a,\id)$. Consider the following dense, Poincar{\'e}-invariant domain in the Fock space
\begin{multline}\label{eq:dom}
 \cD_0 :=\Span_\C\bigg\{ \int\rd^4 x_1\ldots\rd^4 x_n\,f(x_1,\ldots,x_n)\,A_{i_1}(x_1)\ldots A_{i_n}(x_n)\Omega ~:
 \\
 ~n\in\N_0, ~i_1,\ldots,i_n\in\{1,\ldots,\mathrm{p}\},~f\in\cS(\R^{4n}) \bigg\}.
\end{multline}
As shown in~\cite{wightman1965fields} Wick polynomials are are well defined as operator-valued Schwartz distributions on~$\cD_0$. In fact, the following theorem holds.

\begin{thm}\label{thm:eg0}\emph{\cite{epstein1973role}} 
Let $t\in\mathcal{S}'(\R^{4n})$ be translationally invariant, i.e.
\begin{equation}
 t(x_1,\ldots,x_n) = t(x_1+a,\ldots,x_n+a)
\end{equation}
for all $a\in\R^4$. Then for any $B_1,\ldots,B_n\in\Fa$
\begin{equation}\label{eq:eg_operator}
 t(x_1,\ldots,x_n) \normord{B_1(x_1)\ldots B_n(x_n)}~\in \mathcal{S}'(\R^{4n},L(\cD_0)).
\end{equation}
\end{thm}

\subsection{Time ordered products}\label{sec:time_ordered}

In perturbative QFT the scattering operator and the interacting fields are constructed with the use of the time-ordered products. The time-ordered products are multi-linear maps
\begin{equation}
 \T:~\Fa^n \ni (B_1,\ldots,B_n) \mapsto \,\T(B_1(x_1),\ldots ,B_n(x_n)) \,\in \cS'(\R^{4n},L(\cD_0)),
\end{equation}
where $n\in\N_0$, satisfying the following axioms:
\begin{enumerate}[label=\bf{A.\arabic*}]
\item\label{axiom:init} $\T(\emptyset)=\id$, $\T(B(x))=\,\,\normord{B(x)}$,
\begin{equation}
\T(B_1(x_1),\ldots,B_n(x_n),1(x_{n+1})) = \T(B_1(x_1),\ldots,B_n(x_n)),
\end{equation}
where $1$ on the LHS of the above equality is the unity in $\Fa$.

\item\label{axiom:p} Poincar\'e covariance: 
\begin{multline}\label{eq:n_P}
 U(a,\Lambda)\T(B_1(x_1),\ldots,B_n(x_n))U(a,\Lambda)^{-1} 
 \\
 = \T((\rho(\Lambda)B_1)(\Lambda x_1+a),\ldots,(\rho(\Lambda)B_n)(\Lambda x_n+a)),
\end{multline}
where $B_1,\ldots,B_n\in\Fa$, $\rho$ is the representation of $SL(2,\C)$ acting in $\mathcal{F}$ and $U(a,\Lambda)$ is the unitary representation of the Poincar{\'e} group in $\cD_0$.

\item\label{axiom:sym} Symmetry: 
\begin{equation}\label{eq:T_graded}
 \T(B_1(x_1),\ldots,B_n(x_n)) = \T(B_{\pi(1)}(x_{\pi(1)}),\ldots,B_{\pi(n)}(x_{\pi(n)})),
\end{equation}
where $\pi$ is any permutation of the set $\{1,\ldots,n\}$.

\item\label{axiom:causality} Causality: If none of the points $x_1,\ldots,x_m$ is in the causal past of any of the points $x_{m+1},\ldots,x_n$ then
\begin{multline}\label{eq:T_causality}
 \T(B_1(x_1),\ldots,B_n(x_n)) 
 \\
 = \T(B_1(x_1),\ldots,B_m(x_m))\T(B_{m+1}(x_{m+1}),\ldots,B_n(x_n)).
\end{multline}

\item\label{axiom:wick} Wick expansion: $\T(B_1(x_1),\ldots,B_n(x_n))$ is uniquely determined by the vacuum expectations of the time-ordered products of the sub-polynomials of $B_1,\ldots,B_n$
\begin{multline}\label{eq:T_expansion}
 \T(B_1(x_1),\ldots,B_n(x_n)) 
 \\
 =
 \sum_{s_1,\ldots,s_n} 
 \,(\Omega|\T(B_1^{(s_1)}(x_1),\ldots,B_n^{(s_n)}(x_n))\Omega)
 ~\frac{\normord{A^{s_1}(x_1)\ldots A^{s_n}(x_n)}}{s_1!\ldots s_n!}.
\end{multline}
The sub-polynomials of the interaction vertex of QED are listed in~\eqref{eq:subpolynomials_qed}.

\item\label{axiom:steinmann}
Bound on the Steinmann scaling degree that controls the UV behavior of the time-ordered products:
\begin{equation}\label{eq:sd_bound}
 \sd(\,(\Omega|\T(B_1(x_1),\ldots,B_n(x_n),B_{n+1}(0)) \Omega)\,) \leq \sum_{j=1}^{n+1} \dim(B_j),
\end{equation}
where $\sd(t)$ is the Steinmann scaling degree~\cite{steinmann1971perturbative,brunetti2000microlocal} of $t\in\cS'(\R^4)$ and $\dim(B)$ is the dimension of a polynomial $B$ (the dimension $\dim(B)$ may be strictly larger than the canonical dimension; see Sec.~\ref{sec:scalar_model}).

\item\label{axiom:u} Unitarity:
\begin{equation}\label{eq:n_U}
 \T(B_1(x_1),\ldots,B_n(x_n))^* = \aT(B_n^*(x_n),\ldots,B_1^*(x_1)),
\end{equation} 
where $\aT$ is the anti-time-ordered product (cf. e.g.~\cite{epstein1973role}).

\item\label{axiom:pct} Covariance with respect to the discrete group of CPT transformations (the charge conjugation, the spatial-inversion and the time-reversal). 

\item\label{axiom:one} Field equations: If one of the arguments of a time-ordered product is a basic field, then this time-ordered product is related to some time-ordered product with one argument less (for details see e.g. the normalization condition~{\bf N.FE} introduced in Sec.~4.3 of~\cite{duch2018weak}).

\item\label{norm:ward} Ward identities in QED~\cite{dutsch1999local}: Let $B_1,\ldots,B_n$ be sub-polynomials of the QED interaction vertex $\cL=J_\mu A^\mu=\overline{\psi}\slashed{A}\psi$, i.e.
\begin{equation}\label{eq:subpolynomials_qed}
 B_1,\ldots,B_{k}\in\{1,A_\mu,\psi_a,\overline{\psi}_a,J^\mu,(\slashed{A}\psi)_a,(\overline{\psi}\slashed{A})_a,\cL\}.
\end{equation}
It holds
\begin{equation}\label{eq:n_W}
 \partial^x_{\mu}\T(J^\mu(x),B_1(x_1),\ldots,B_{k}(x_{k}))
 = \sum_{j=1}^{k} \mathbf{q}(B_j) \,\delta(x-x_j)\,\T(B_1(x_1),\ldots,B_{k}(x_{k})),
\end{equation}
where $\mathbf{q}(B)$ is the additive electric charge number defined such that $\mathbf{q}(A_\mu)=0$, $\mathbf{q}(\psi_a)=-\mathbf{q}(\overline{\psi}_a)=-1$.
\end{enumerate}

The above conditions were stated for bosonic fields and have to be slightly modified in the case of fields with Fermi statistics (see e.g.~\cite{duch2018weak}). We refer the reader to~\cite{epstein1973role,scharf2014,brunetti2000microlocal,hollands2002existence} for the construction of the time-ordered products satisfying the above axioms. Let us mention that the time-ordered products are not defined uniquely. Assume that all the time-ordered products with at most $n$ arguments are fixed. Then two possible definitions of the vacuum expectation value
\begin{equation}\label{eq:T_freedom}
 (\Omega|\T(B_1(x_1),\ldots,B_{n+1}(x_{n+1}))\Omega)\in\cS'(\R^{4(n+1)})
\end{equation}
differ by a distribution of the form
\begin{equation}\label{eq:freedom_form}
 \sum_{\substack{\gamma\\|\gamma|\leq\omega}} c_\gamma \partial^\gamma \delta(x_1-x_{n+1})\ldots \delta(x_{n}-x_{n+1})
\end{equation}
for some constants $c_\gamma\in\C$ indexed by multi-indices $\gamma$, $|\gamma|\leq\omega$, where
\begin{equation}\label{eq:omega}
 \omega := 4 - \sum_{j=1}^{n+1} (4 -\dim(B_j))
\end{equation}
is the order of singularity. The constants $c_\gamma$ have to be chosen in such a way that the new set of time-ordered products satisfies all the axioms. Note that if $\omega<0$, then~\eqref{eq:T_freedom} is determined uniquely by the time-ordered products with at most $n$ arguments. Complete characterization of the freedom in defining the time-ordered products can be found in~\cite{hollands2001local,pinter2001finite,dutsch2004causal}.

\subsection{S-matrix and interacting fields with adiabatic cutoff}\label{sec:bogoliubov}

The scattering matrix of a model of perturbative QFT with the interaction vertex $\cL\in\Fa$ is formally given by the Dyson~\cite{dyson1949radiation} formula $S = \Texp\left(\ri e\int \rd^4 x \,\mathcal{L}(x)\right)$, where $\cL\in\Fa$ is the interaction vertex. However, the above expression is ill-defined even in purely massive models because the time-ordered products $\T(\mathcal{L}(x_1),\ldots,\mathcal{L}(x_n))$ are operator-valued Schwartz distribution and generically cannot be integrated with constant functions. In order to obtain a well-defined expression for the scattering matrix in QFT we follow the idea of Bogoliubov~\cite{bogoliubov1959introduction} and multiply the interaction vertex $\mathcal{L}$ by the switching function $g$ which is a real-valued Schwartz function defined on the spacetime. The above prescription regularizes the theory in the IR regime. We call this type of the IR regularization the adiabatic cutoff. Note that in the theory with adiabatic cutoff there are only short-range interactions. The scattering matrix with adiabatic cutoff, which is also called the Bogoliubov S-matrix, is defined by~\cite{bogoliubov1959introduction}
\begin{multline}\label{eq:bogoliubov_S_op}
 S(g) = \Texp\left(\ri e\int\rd^4 x \, g(x)\cL(x)\right)
 \\
 :=\sum_{n=0}^\infty \frac{\ri^n  e^n}{n!} \!\int \rd^4 x_1\,\ldots \rd^4 x_n\, g(x_1)\ldots g(x_n) \T(\cL(x_1),\ldots,\cL(x_n)) \in L(\cD_0)\llbracket e\rrbracket.
\end{multline}
The above expression is well-defined as a formal power series of operators mapping the domain $\cD_0$ into $\cD_0$. The physical scattering operator is obtained by taking the limit $g(x)\to 1$ of the $S(g)$ in an appropriate sense to be discussed in Sec.~\ref{sec:adiabatic}.

Bogoliubov proposed also an elegant method of constructing the interacting fields. For any polynomial in the basic fields and their derivatives $C\in\Fa$ the corresponding retarded and advanced interacting fields are given by~\cite{bogoliubov1959introduction}
\begin{equation}\label{eq:bogoliubov_fields}
\begin{aligned}
 C_{\ret}(g;x) &:= (-\ri) \frac{\delta}{\delta h(x)} S(g)^{-1}S(g;h)\bigg|_{h=0}\in L(\cD_0)\llbracket e\rrbracket,
 \\
 C_{\adv}(g;x) &:= (-\ri) \frac{\delta}{\delta h(x)} S(g;h)S(g)^{-1}\bigg|_{h=0}\in L(\cD_0)\llbracket e\rrbracket,
\end{aligned}
\end{equation}
where
\begin{equation}\label{eq:bogoliubov_S_op_extended}
 S(g;h) := \Texp\left(\ri e \int\rd^4 x\, g(x) \mathcal{L}(x) + \ri\int\rd^4 x\, h(x) C(x) \right)
\end{equation}
is the so-called extended scattering matrix. If the switching function $g$ has a compact support then the retarded and advanced fields coincide with the free field $\normord{C(x)}$ outside the past and future of $\supp\,g$, respectively. If $C\in\Fa$ is a basic field, then the corresponding retarded/advanced field $C_{\ret/\adv}(g;x)$ satisfies the interacting equation of motion with the coupling constant $e$ replaced with $e g(x)$. For example, in the case of the massless $\varphi^4$ model we have
\begin{equation}
 \square_x \varphi_{\ret/\adv}(g;x)+ \frac{e}{3!} g(x)\, (\varphi^3)_{\ret/\adv}(g;x) = 0.
\end{equation}
Moreover, the interacting fields satisfy the local commutativity condition, i.e. for any $C,C'\in\Fa$ it holds \cite{epstein1973role,dutsch1999local}
\begin{equation}
 [C_{\ret}(g;x),C'_{\ret}(g;x')] = 0,
 ~~~
 [C_{\adv}(g;x),C'_{\adv}(g;x')] = 0 
\end{equation}
if $x$ and $x'$ are spatially separated. Let us mentioned that the retarded and advanced fields $C_{\ret/\adv}(g;h)$ with adiabatic cutoff $g\in\cS(\R^4)$ can be used to construct~\cite{dutsch2001algebraic,fredenhagen2015perturbative} the net of abstract algebras $\mathfrak{F}(\mathcal{O})$ of interacting fields localized in bounded spacetime regions $\mathcal{O}\subset\R^4$. The net does not depend on the switching function and satisfies the Haag-Kastler axioms~\cite{haag2012local} in the sense of formal power series.

\subsection{Adiabatic limit}\label{sec:adiabatic}

In order to define the physical scattering operator or the vacuum representation of the net $\mathfrak{F}(\mathcal{O})$ one uses the so-called strong adiabatic limit. To this end, one introduces a one-parameter family of the switching functions $g_\epsilon$, $\epsilon\in(0,1)$, defined by $g_\epsilon(x)=g(\epsilon x)$, where $g\in\cS(\R^4)$ is an arbitrary real-valued Schwartz function such that $g(0)=1$. We say that the adiabatic limit of $S(g_\epsilon)$ or $C_{\ret/\adv}(g_\epsilon;h)$ exists if the limits
\begin{equation}
 \lim_{\epsilon\searrow0}\,S(g_\epsilon),\quad
 \lim_{\epsilon\searrow0}\,C_{\ret/\adv}(g_\epsilon;h)
\end{equation}
exist in appropriate sense to be specified (see e.g. Eqs.~\eqref{eq:adiabatic_lim_massive} or Eqs.~\eqref{eq:dfn_S_op_mod_scalar}) and do not depend on the choice of $g$ used in the definition of $g_\epsilon$. The existence of the adiabatic limit of the scattering operator in massive models was proved in~\cite{epstein1976adiabatic}. In what follows, we summarize the main results of~\cite{duch2018strong}, where the above-mentioned proof is simplified and generalized. 

Consider an arbitrary model of perturbative QFT which contains only massive fields. Let us define the following domain in the Fock space
\begin{multline}\label{eq:dom_holder}
 \cD_{\mathrm{H}}:=\Span_\C\bigg\{\! \int\mHm{p_1}\ldots\mHm{p_n}\,
 h(p_1,\ldots,p_n)
 \\
 \times
 \,a^*(p_1)\ldots a^*(p_n)\Omega:
 ~n\in\N_0,~h\in \cS_{\mathrm{H}}(\Hm^{\times n}) \bigg\},
\end{multline}
where $\cS_{\mathrm{H}}(\Hm^{\times n})$ consists of H{\"o}lder-continuous functions which decay rapidly at infinity. For any $\Psi\in\cD_{\mathrm{H}}$, $C\in\Fa$ and $h\in\cS(\R^4)$ the following limits:
\begin{equation}\label{eq:adiabatic_lim_massive}
 S\Psi:=\lim_{\epsilon\searrow0}~S(g_\epsilon)\Psi\in\cD_{\mathrm{H}}\llbracket e \rrbracket,
 \quad
 C_{\ret/\adv}(h)\Psi := \lim_{\epsilon\searrow0}~C_{\ret/\adv}(g_\epsilon;h)\Psi \in\cD_{\mathrm{H}}\llbracket e \rrbracket
\end{equation}
exist in each order of perturbation theory and define the physical scattering operator $S\in L(\cD_{\mathrm{H}})\llbracket e \rrbracket$ and the physical interacting field operators $C_{\ret/\adv}(h) \in L(\cD_{\mathrm{H}})\llbracket e \rrbracket$. The above objects are covariant with respect to the Poincar{\'e} transformations, i.e.
\begin{equation}
\begin{aligned}
 U(a,\Lambda) S U(a,\Lambda)^{-1} =& S,
 \\
 U(a,\Lambda) C_{\ret/\adv}(h) U(a,\Lambda)^{-1} 
 =& (\rho(\Lambda^{-1})C)_{\ret/\adv}(h_{a,\Lambda}),
\end{aligned} 
\end{equation}
where $h_{a,\Lambda}(x) = h(\Lambda^{-1}(x-a))$, $\rho$ is the representation of the Lorentz group acting in~$\Fa$ and $U(a,\Lambda)$ is the standard representation of the Poincar{\'e} group acting in the Fock space. Moreover, the scattering operator satisfies the following conditions
\begin{equation}
 S\Omega = \Omega,
 \quad
 S|p) = |p)
\end{equation}
expressing the stability of the vacuum and one-particle states. The interacting field operators fulfill the interacting field equations and the local commutativity condition. Their LSZ limits coincide with the free fields. Thus, the physical particle interpretation of the states in the Fock space coincides with the standard one. 

The above method provides a direct and complete perturbative construction of massive models of QFT in the vacuum representation. However, because of the IR problem it is not applicable to most models with massless fields. In Sec.~\ref{sec:modified} we present a modification of the above method which can be used to construct the scattering matrix and interacting fields in QED and the scalar model introduced in Sec.~\ref{sec:models}. The construction allows us to draw a number of interesting conclusions about the IR structure of these models. The properties of the scattering matrix and interacting fields in QED and the scalar model turn out to be quite different than the one listed above. In particular, the particle interpretation of states is more involved and the energy-momentum operators do not coincide with the standard ones and have different spectrum. Moreover, there is an abundance of the superselection sectors.

\section{Models with infrared problem}\label{sec:models}

In this section we recall basic facts about the perturbative construction of QED in the BRST framework and give the definition of the scalar model. 

\subsection{Scalar model}\label{sec:scalar_model}

The scalar model is a model of QFT which we use to study the IR problem in perturbation theory. As we will see its IR structure is quite similar to that of QED. Its classical action is given by
\begin{multline}
 \mathrm{S}[\varphi,\psi] = \int\rd^4x\, \bigg( \frac{1}{2}(\partial_\mu \psi(x))(\partial^\mu \psi(x)) - \frac{1}{2} \mass^2\psi^2(x) 
 \\
 + \frac{1}{2} (\partial_\mu \varphi(x))(\partial^\mu \varphi(x)) + \frac{e}{2}\psi^2(x)\varphi(x) \bigg),
\end{multline}
where the fields $\varphi(x)$ and $\psi(x)$ are real and scalar and $e$ is the coupling constant. The field $\varphi(x)$ is massless and the field $\psi(x)$ has mass $\mass$. The Euler-Lagrange equations have the following form
\begin{equation}
 (\square + \mass^2) \psi(x) =e\psi(x)\varphi(x),
 \quad
 \square \varphi(x)=\frac{e}{2}\psi^2(x).
\end{equation}
The perturbative quantization of the model is straightforward. The Hilbert space is the tensor product of the Fock spaces describing the massless particles and the particles of mass $\mass$,
\begin{equation}
 \cH = \Gamma_s(\mathfrak{h}_0)\otimes \Gamma_s(\mathfrak{h}_\mass),
\end{equation}
where $\Gamma_{s/a}(\mathfrak{h})$ denotes the symmetric/antisymmetric Fock space built over the one-particle Hilbert space $\mathfrak{h}$. The Hilbert spaces $\mathfrak{h}_0$ and $\mathfrak{h}_\mass$ are, respectively, the completions of $C^\infty_{\mathrm{c}}(H_0)$ and $C^\infty_{\mathrm{c}}(H_\mass)$ with respect to the topologies given by the scalar products 
\begin{equation}
 (f_1|f_2)=\int \mHO{k} \,\overline{f_1(k)}f_2(k),
 \quad
 (f_1|f_2)=\int \mHm{p} \,\overline{f_1(p)}f_2(p).
\end{equation}
By analogy to QED, the massless particles are called the photons and the massive particles -- the electrons. The creation and annihilation operators of the photons are denoted by
$a^*(k)$, $a(k)$, whereas the creation and annihilation operators of the electrons are denoted by $b^*(p)$, $b(p)$. The free quantized fields are given by
\begin{equation}\label{eq:phi_def}
 \varphi(x) =  \int \mHO{k} \left(a^*(k)\re^{\ri k\cdot x}+\textrm{h.c.}\right),
 \quad
 \psi(x)= \int \mHm{p} \left(b^*(p)\re^{\ri p\cdot x}+\textrm{h.c.}\right).
\end{equation}
The interaction vertex of the model coincides with $\cL = \frac{1}{2}\psi^2\varphi =  J \varphi$, where, by analogy to QED, the scalar quantity $J = \frac{1}{2}\psi^2$ is called the current. Note that, as in QED, the interaction vertex contains two massive fields and one massless. For future reference, let us also define the operator 
\begin{equation}\label{eq:def_charge_dist_momentum_scalar}
 \rho(p):=\frac{1}{2\mass}\,b^*(p) b(p).
\end{equation}

Because the canonical dimension of the interaction vertex is equal to three the scalar model is super-renormalizable according to the standard classification. However, as explained in~\cite{duch2017massless,duch2018weak}, the normalization condition~\eqref{eq:sd_bound} is to restrictive to construct the Wightman and Green functions in the scalar model if the dimension $\dim$ is the canonical dimension, i.e. $\dim(\varphi)=\dim(\psi)=1$. For this reason, in the case of the scalar model we assign dimensions to the fields in a different way. In this paper we assume that $\dim(\varphi)=1$, $\dim(\psi)=3/2$ and each derivative increases the dimension by one. With this choice of of the dimensions the scalar model is a renormalizable (not super-renormalizable) theory. Moreover, as showed in~\cite{duch2017massless,duch2018weak}, it is possible to construct the physical Wightman and Green functions in this model using the weak adiabatic limit.

The standard unitary representation of the Poincar{\'e} group acting in $\cH$ is denoted by $U(a,\Lambda)$. The generators of the space-time translations $U(a)$ are given by
\begin{equation}
 P^\mu = \int \mHO{k}\, k^\mu a^*(k) a(k) + \int \mHm{p}\, p^\mu b^*(p)b(p).
\end{equation}

\subsection{Quantum electrodynamics}\label{sec:qed}

The standard classical action of QED has the following form
\begin{equation}\label{eq:qed_action}
 \mathrm{S}[A_\mu,\psi_a,\overline{\psi}_a] = \int\rd^4x\, \left( \overline{\psi}(x) (\ri\slashed{D}-m) \psi(x) - \frac{1}{4} F_{\mu\nu}(x)F^{\mu\nu}(x) \right),
\end{equation}
where $A_\mu(x)$ is the electromagnetic potential, $F_{\mu\nu}(x)=\partial_\mu A_\nu(x)-\partial_\nu A_\mu(x)$ is the electromagnetic field tensor, $D_\mu = \partial_\mu-\ri e A_\mu$ is the covariant derivative and $\psi_a(x)$, $\overline{\psi}_a(x)$ are Dirac spinor fields. We assign to the fields their canonical dimensions, i.e $\dim(A_\mu)=1$ and $\dim(\psi_a)=\dim(\overline{\psi}_a)=3/2$. The coupling constant, denoted by $e$, coincides with the elementary charge which is equal to minus the charge of electron. Because of the gauge freedom the Euler-Lagrange equations for the above action
do not admit a well-posed initial value formulation. In order to deal with the this problem one uses the BRST method~\cite{becchi1975renormalization,tyutin1975gauge,becchi1976renormalization} which is a generalization of the approach developed by Gupta and Bleuler~\cite{gupta1950theory,bleuler1950neue}. To this end, one considers a modified theory with the action 
\begin{multline}\label{eq:action_qed_mod}
 \mathrm{S}[A_\mu,\psi_a,\overline{\psi}_a,c,\bar c] 
 = \int\rd^4x\, \bigg(\overline{\psi}(x)(\ri\slashed{D}-m)\psi(x) 
 \\
 - \frac{1}{2} (\partial_\mu A_\nu(x)) (\partial^\mu A^\nu(x)) + \ri (\partial_\mu c(x))(\partial^\mu \bar c(x))\bigg),
\end{multline}
where $c$ and $\bar c$ are the ghost and and anti-ghost fields. The Euler-Lagrange equations for the above action are of the form
\begin{equation}
 \square A^\mu(x)=\, -e\overline{\psi}(x)\gamma^\mu\psi(x),
 \quad
 (\ri\slashed{D} -m) \psi(x) =\, 0,
 \quad
 \square c(x) =\, 0,
 \quad
 \square \bar c(x) =\, 0.
\end{equation}
The above system of partial differential equations is normally-hyperbolic and has a well-posed initial value formulation. The quantization of the above model is described in detail e.g. in~\cite{scharf2014,scharf2016gauge}. The free fields are defined in the Krein-Hilbert-Fock space which is a tensor product of the photon, electron/positron and ghost/anti-ghost Fock spaces
\begin{equation}
 \cH = \Gamma_s(\mathfrak{h}_{\textrm{ph}})\otimes \Gamma_a(\mathfrak{h}_{\textrm{el}})
 \otimes \Gamma_a(\mathfrak{h}_{\textrm{gh}}),
\end{equation}
where $\mathfrak{h}_{\textrm{ph}}$, $\mathfrak{h}_{\textrm{el}}$ and $\mathfrak{h}_{\textrm{gh}}$ are the one-photon, one-electron/positron and one-ghost/anti-ghost Hilbert spaces, respectively. The positive-definite scalar products on $\mathfrak{h}_{\textrm{ph}}$ and $\mathfrak{h}_{\textrm{gh}}$ are given by
\begin{equation}
 \langle f|g\rangle_{\textrm{ph}} := \sum_{\mu=0}^3 \int\mHO{k}\,\overline{f^\mu(k)} g^\mu(k),
 \quad
 \langle f|g\rangle_{\textrm{gh}} := \sum_{a=1,2} \int\mHO{k}\,\overline{f_a(k)} g_a(k),  
\end{equation}
respectively. The corresponding covariant Krein scalar products are defined as follows
\begin{equation}
 (f|g)_{\textrm{ph}} := -\int\mHO{k}\,g_{\mu\nu} \overline{f^\mu(k)} g^\nu(k),
 \quad
 (f|g)_{\textrm{gh}} := \int\mHO{k}\,\epsilon_{ab} \overline{f_a(k)} g_b(k),
\end{equation}
where $g_{\mu\nu}$ is the Minkowski metric with the signature $(+,-,-,-)$ and $\epsilon_{ab}$ is the antisymmetric symbol with $\epsilon_{12}=1$. The Fock space of photons contains all four polarizations. The Fock space of electrons is completely standard. In particular, the positive-definite scalar product defined in this space is covariant. The hermitian conjugation of an operator $A$ defined using the positive-definite scalar product is denoted by $A^\dagger$. On the other hand, the Krein conjugation defined using the covariant Krein scalar product is denoted by $A^*$. The creation and annihilation operators of photons are denoted by $a^\dagger_\mu(k)$, $a_\mu(k)$, $\mu=0,1,2,3$. It holds $a^\dagger_\mu(k)=a^*_\mu(k)$ for $\mu=1,2,3$ and $a^\dagger_0(k)=-a^*_0(k)$. The creation and annihilation operators of ghost are denoted by $c_a(k)$, $c_a^\dagger(k)$, $a=1,2$. It holds $c^\dagger_a(k)=\epsilon_{ab} c^*_b(k)$. The above creation and annihilation operators satisfy the following commutation relations
\begin{equation}
 [a_\mu(f^\mu),a^*_\nu(g^\nu)] = (f|g)_{\textrm{ph}},
 \quad
 [c_a(f_a),c^*_b(g_b)]_+ = (f|g)_{\textrm{gh}},
\end{equation}
where $[\cdot,\cdot]_+$ denotes the anti-commutator. The creation and annihilation operators of electrons/positrons are denoted by $b^*(\sigma,p)=b^\dagger(\sigma,p)$, $b(\sigma,p)$, $d^*(\sigma,p)=d^\dagger(\sigma,p)$, $d(\sigma,p)$, where $\sigma=1,2$ corresponds to the two spin states of electron/positron. The spin vectors are denoted by $u^a(\sigma,p)$, $v^a(\sigma,p)$. The free fields are the following operator-valued Schwartz distributions
\begin{equation}\label{eq:vector_field}
 A_\mu(x) =  \int \mHO{k} \left(a_\mu^*(k)\re^{\ri k\cdot x}+a_\mu(k) \re^{-\ri k\cdot x}\right),
\end{equation}
\begin{equation}\label{eq:psi_field}
 \psi_a(x)= \sum_{\sigma=1,2}\int \mHm{p} \left(b^*(\sigma,p)u_a(\sigma,p)\re^{\ri p\cdot x}+d(\sigma,p)v_a(\sigma,p) \re^{-\ri p\cdot x}\right),
\end{equation}
\begin{equation}\label{eq:psi_bar_field}
 \overline{\psi}_a(x)= \sum_{\sigma=1,2}\int \mHm{p} \left(d^*(\sigma,p)\overline{v}_a(\sigma,p)\re^{\ri p\cdot x}+b(\sigma,p)\overline{u}_a(\sigma,p) \re^{-\ri p\cdot x}\right),
\end{equation}
\begin{equation}\label{eq:ghost_field}
 c(x) =  \int \mHO{k} \left(c_1^*(k)\re^{\ri k\cdot x}+c_1(k) \re^{-\ri k\cdot x}\right),
\end{equation}
\begin{equation}\label{eq:anti_ghost_field}
 \bar c(x) =  \int \mHO{k} \left(c_2^*(k)\re^{\ri k\cdot x}+c_2(k) \re^{-\ri k\cdot x}\right).
\end{equation}
For future reference, let us also introduce the operator 
\begin{equation}\label{eq:def_charge_dist_momentum}
 \rho(p):=\sum_{\sigma=1,2}(b^*(\sigma,p) b(\sigma,p) - d^*(\sigma,p) d(\sigma,p)).
\end{equation}
The electric charge operator $Q = \int\mHm{p}\,\rho(p)$ coincides with the difference between the number of electrons and positrons. Thus, the physical electric charge equals $-eQ$.

The interaction vertex of QED is given by $\cL = \overline{\psi} \gamma^\mu\psi A_\mu = J^\mu A_\mu$, where $J^\mu = \overline{\psi} \gamma^\mu\psi$ is the spinor current. The definition of the scattering matrix and the interacting fields with adiabatic cutoff in the model described by the action~\eqref{eq:action_qed_mod} do not pose any difficulties. However, it is not obvious how these objects are related to the corresponding objects in QED which is supposed to be defined by the action~\eqref{eq:qed_action}. In order to address this problem one introduces the free BRST charge
\begin{equation}
 Q_{\mathrm{BRST}} 
 =\int\mHO{k}\,k^\mu (a_\mu^*(k) c_1(k) + a_\mu(k) c_1^*(k))
\end{equation}
which is well defined as an element of $L(\cD_0)$. 
The covariant inner product is semi-positive-definite on $\Ker Q_{\mathrm{BRST}}$ i.e. if $\Psi\in\Ker Q_{\mathrm{BRST}}$, then $(\Psi|\Psi)\geq 0$. Moreover, if $\Psi\in\Ker Q_{\mathrm{BRST}}\cap\Ker Q_{\mathrm{gh}}$, where $Q_{\mathrm{gh}}$ is the ghost number operator, then the condition $(\Psi|\Psi)= 0$ implies that $\Psi\in\Ran Q_{\mathrm{BRST}}$. We conclude that the covariant inner product induces a positive-definite inner product on $\cD^{\textrm{phys}}_0$, where
\begin{equation}
 \cD^{\textrm{phys}}_0 = \frac{\Ker Q_{\mathrm{BRST}}\cap\Ker Q_{\mathrm{gh}}}{\Ran Q_{\mathrm{BRST}}\cap\Ker Q_{\mathrm{gh}}}.
\end{equation}
Elements of $\cD^{\textrm{phys}}_0$ are equivalence classes denoted by $[\Psi]$ with $\Psi\in\cD_0$. Let $\cH^{\textrm{phys}}$ be the Hilbert space obtained by the completion of the pre-Hilbert space $\cD^{\textrm{phys}}_0$ with the inner product induced from the covariant inner product $(\cdot|\cdot)$. If an operator $B\in L(\cD_0)$ commutes with the free BRST and ghost charge, then it induces a unique operator $[B]\in L(\cD^{\textrm{phys}}_0)$ defined by the equality 
\begin{equation}
 [B][\Psi]:= [B\Psi].
\end{equation}
Let $\varepsilon^\mu(1,k)$ and $\varepsilon^\mu(2,k)$ be real vectors of two arbitrarily chosen physical polarizations of photons. More specifically, we demand that $\varepsilon^\mu(1,k)$, $\varepsilon^\mu(2,k)$ and $k^\mu$ are linearly independent and 
\begin{equation}
 k_\mu\varepsilon^\mu(s,k) = 0,
 \quad 
 \overline{\varepsilon_\mu(s,k)}\varepsilon^\mu(s',k) = -\delta_{ss'}. 
\end{equation}
The operators
\begin{equation}\label{eq:qed_physical_c_a}
 a^\#(s,k):=\varepsilon^\mu(s,k) [a^\#_\mu(k)],
 \quad s=1,2,
\end{equation}
are the creation and annihilation operators of the physical free photons and satisfy the standard commutation relations. Abusing the notation we identify $b^*(\sigma,p)\equiv [b^*(\sigma,p)]$ and $d^*(\sigma,p)\equiv [d^*(\sigma,p)]$. The polynomials in the operators $a^*(s,k)$, $b^*(\sigma,p)$, $d^*(\sigma,p)$ smeared with Schwartz functions generate a dense subspace in the physical Hilbert space when acting on the vacuum $[\Omega]$. 

The standard representation of the Poincar{\'e} group acting in the Krein-Hilbert-Fock space $\cH$ is denoted by $U(a,\Lambda)$. 
The transformations $U(a,\Lambda)\in L(\cD_0)$ commute with the free BRST charge. Since the representation $U(a,\Lambda)$ is Krein-unitary it induces a unitary representation $[U(a,\Lambda)]$ acting in the physical Hilbert space. The generators of the subgroup of space-time translations $[U(a)]$ have the following form
\begin{multline}
 [P^\mu] = \sum_{s=1,2} \int \mHO{k}\, k^\mu a^*(s,k) a(s,k) 
 \\
 +
 \sum_{\sigma=1,2}\int \mHm{p}\, 
 p^\mu (b^*(\sigma,p) b(\sigma,p) +d^*(\sigma,p) d(\sigma,p)).
\end{multline}

%

\section{Infrared problem in scattering theory}\label{sec:infrared}

In this section we review the origin of the IR problem in the scattering of particles in models with long-range interactions. 

\subsection{Classical and quantum mechanics}\label{sec:Coulomb_potential}

The IR problem in QED can be traced back to classical mechanics. Consider a classical particle of mass $\mass$ moving in the repulsive Coulomb potential.\footnote{We consider only the case of the repulsive Coulomb potential in order to avoid minor complications caused by the presence of the bound states.} Let $\vec x,\vec p\in\R^3$ be its position and momentum. The Hamiltonian of the system under consideration is given by the following expression 
\begin{equation}\label{eq:Coulomb_hamiltonian}
 H := H_\fr + V(\vec{x}),
 \quad
 H_\fr := \frac{|\vec p|^2}{2\mass},
 \quad
 V(\vec{x}) = \frac{e^2}{4\pi}\frac{1}{|\vec x|}.
\end{equation}
We call the position of the particle as a function of time,
\begin{equation}
 \R\ni t\mapsto \vec x(t) \in\R^3,
\end{equation}
its trajectory. If the particle is moving in the short range potential $V(\vec{x})$ such that $|V(\vec x)|< \const\,|\vec x|^{-1-\delta}$ for some $\delta>0$, then its trajectory is asymptotic in the future and in the past to the unique trajectories of the free particles whose evolution is governed by $H_\fr$, i.e. for some $\vec x_{\rout/\rin},\vec v_{\rout/\rin}\in\R^3$ it holds
\begin{equation}
 \lim_{t\to\pm\infty}\,|\vec x(t) - \vec x_{\rout/\rin} - \vec v_{\rout/\rin} t |=0.
\end{equation}
One easily verifies that the velocity of the particle moving in the Coulomb potential acquires specific values at the future and past infinity. However, the asymptotic values are approached so slowly that the distance between the particle influenced by the Coulomb potential and a freely moving particle always diverges irrespective of the choice of the initial position and velocity of the free particle -- the trajectory of the particle moving in the Coulomb potential cannot be asymptotic to a trajectory of any free particle. In fact, one shows that 
\begin{equation}
 \lim_{t\to\pm\infty}\,\left|\vec x(t) - \vec x_{\rout/\rin}(t) \right| =0,
\end{equation}
where the asymptotic trajectories are of the form
\begin{equation}\label{eq:asymp_trajectory_Coulomb}
 \vec x_{\rout/\rin}(t):=\vec x_{\rout/\rin} + \vec v_{\rout/\rin} t -  \frac{e^2}{4\pi\mass} \frac{\vec v_{\rout/\rin}}{|\vec v_{\rout/\rin}|^3}\log(|t|/\zeta)
\end{equation}
for some $\vec x_{\rout/\rin},\vec v_{\rout/\rin}\in\R^3$ and $\zeta\in\R_+$. 

A similar problem arises for the Coulomb scattering in quantum mechanics. The position and momentum of the particle, $\vec x$ and $\vec p$, are now interpreted as operators defined in the Hilbert space $\cH = L^2(\R^3)$. The free and full evolution operators are by denoted by $U_\fr(t) := \exp(-\ri t H_\fr)$, $U(t) := \exp(-\ri t H)$, respectively. The standard M{\o}ller and scattering operators in quantum mechanics, which are well-defined in the case of short-range potentials~\cite{reedsimon3}, are constructed in the following way
\begin{equation}
 \Omega_{\rout/\rin} := \slim_{t\to\pm\infty} U(-t)U_\fr(t),
 ~~~~
 S:=\Omega_{\rout}^*\Omega_{\rin}.
\end{equation}
In the case of quantum-mechanical particle moving in the Coulomb potential the IR problem manifests itself by the nonexistence of the above limits. The problem is related to the fact that the wave-function satisfying the Schr{\"o}dinger equation with the Coulomb potential approaches at large distances the free-particle wave-function only up to a logarithmically divergent phase factor~\cite{schiff1949quantum} called the Coulomb phase. The above-mentioned divergent phase factor depends only on the momentum of the particle and does not contribute to the the differential cross section for the Coulomb scattering which is finite. 

\subsection{Quantum field coupled to classical current}\label{sec:simple_model}

In the case of QED the infinite Coulomb phase is not the only manifestation of the IR problem. Another aspect of the IR problem, which we describe in this section, is the infinite photon emission. In theories without massless particles the principle of conservation of energy guarantees that the number of particles emitted during scattering is always finite. In QED or the scalar model nothing prohibits a production of infinitely many low-energetic photons. In fact, one proves that the probability of an emission of a low-energetic photon in a non-trivial scattering of charged particles is proportional to the inverse of the energy of the emitted photon. Consequently, the expected number of emitted photons is generically infinite and the probability of the emission of any finite number of photons is equal to zero~\cite{bloch1937note}. As a result, in the case of QED or the scalar model both the scattering matrix and the differential cross section are plagued by IR divergences.

In order to illustrate the above-mentioned problem we consider a simple model describing the second-quantized electromagnetic field $F^{\mu\nu}(x)$ coupled to a classical conserved current $J^\mu(x)$ of spatially compact support. The field $F^{\mu\nu}(x)$ satisfies the Maxwell equations 
\begin{equation}\label{eq:maxwell_eqs}
 \partial_\mu F^{\mu\nu}(x) = -eJ^\nu(x),
 \quad\quad
 \partial^{[\mu} F^{\alpha\beta]}(x)=0.
\end{equation}
We assume that the current $J^\mu(x)$ has future and past asymptotes $\rho_{\rin/\rout} \in C(\Hm)$ defined by
\begin{equation}\label{eq:asymptote_current}
 \lim_{\lambda\to\infty} \lambda^3 J^\mu(\pm \lambda v) =  \frac{\mass^2}{2(2\pi)^3}\,v^\mu \rho_{\rout/\rin}(\mass v),
\end{equation}
where $v\in H_1$ and the prefactor was introduced to comply with the notation used in the rest of the paper (note that the parameter $\mass$ does not appear in the model considered in this section). The above assumption with non-trivial $\rho_{\rout/\rin}$ is typically satisfied by physically relevant currents~\cite{flato1997asymptotic,herdegen1995long} (see also Sec.~\ref{sec:asymp_current_behavior}). The unique solution of~\eqref{eq:maxwell_eqs} satisfying no-incoming radiation condition is given by
\begin{equation}\label{eq:simple_model_ret}
 F^{\mu\nu}_\ret(x) = F^{\mu\nu}_\fr(x) 
 - 2e\int\rd^4 y\,D_0^\ret(x-y)\,\partial^{[\mu} J^{\nu]}(y),
\end{equation}
where $F^{\mu\nu}_\fr(x)$ is the standard free quantum field defined in the Fock space. Using the formulas from Appendix~\ref{sec:LSZ} we verify that the past LSZ limit of the field $F^{\mu\nu}_\ret(x)$ coincides with the free quantum field $F_\rin^{\mu\nu}(x)=F^{\mu\nu}_\fr(x)$ whereas the future limit gives
\begin{equation}
 F_\rout^{\mu\nu}(x) = F^{\mu\nu}_\fr(x) 
 - 2e\int\rd^4 y\,D_0(x-y)\,\partial^{[\mu} J^{\nu]}(y).
\end{equation}
One shows that the asymptotic LZS fields $F_\rin^{\mu\nu}(x)$ and $F_\rout^{\mu\nu}(x)$ are unitarily related if and only if $\rho_\rin\equiv \rho_\rout$. Since the above condition is generically violated we conclude that the scattering operator does not exist. If the initial state of the electromagnetic radiation is described by a vector in the standard Fock representation, then the final state is typically a vector is some non-Fock coherent representation which depends on both the incoming and outgoing asymptotes of the current.

In the case of the electromagnetic field $F^{\mu\nu}_\ret$ the long-range tail is informally defined by the following limit with $n$ being a unit spatial four-vector (for the precise definition of the long-range tail see Appendix~\ref{sec:long_rang_tail})
\begin{equation}\label{eq:long_rang_tail_simple}
 \lim_{r\to\infty} r^2\,F^{\mu\nu}_{\ret}(x+rn) = -e  \int_{\Hm}\rd\rho_\rin(p)\,f^{\mu\nu}(p/\mass;n),
\end{equation}
where 
\begin{equation}\label{eq:Coulomb_field}
 f^{\mu\nu}(v;x):=\frac{x^\mu v^\nu- v^\mu x^\nu}{((x\cdot v)^2-x^2)^{3/2}}
\end{equation}
is the Coulomb field of a unit point-charge moving with velocity $v$ and $\rd\rho_{\rout/\rin}(p):=\rho_{\rout/\rin}(p)\,\mHm{p}$. The above field commutes with all local fields, and as a result is a classical observable determining the super-selection sector. The above statement remains true in QED~\cite{buchholz1986gauss} and its consequences are the superselection of the velocity of the electron and the fact that the charged states cannot be eigenstates of the mass operator (the so-called infraparticle problem).

\subsection{Relativistic perturbative QFT}\label{sec:IR_problem_qft}

In the case of perturbative QED or the scalar model already the first order correction to the standard scattering operator is ill-defined. In the second order even its matrix elements between regular states are divergent.

\subsubsection{Problems in first-order of perturbation theory}

Let us investigate the first-order correction to the Bogoliubov S-matrix in the scalar model. All possible processes contributing in this order are listed in Sec.~\ref{sec:S_op_mod_first_order_corrections}. In what follows we consider only the decay of an electron into a photon and an electron. Let $p$ and $k$ be the momenta of the outgoing electron and photon, respectively, and $p'$ be the momentum of the incoming electron. All of the momenta are assumed to be on-shell. The energy-momentum conservation $p+k=p'$ implies that $p\cdot k =0$. Thus, $p=p'$ and $k=0$ which suggest that the amplitude of the decay should vanish. However, one has to bear in mind that the energy and momentum are conserved only in the adiabatic limit, if this limit exists. The amplitude for the decay is clearly non-vanishing if $\epsilon>0$. The wave function of the outgoing particles is given explicitly by 
\begin{equation}
 F_\epsilon(p,k) 
 = 
 \ri\int\mHm{p'}\,(p,k|S^{[1]}(g_\epsilon)|p')\, f(p')
 =
 \ri\int\mHm{p'}\, \F{g}_\epsilon(p+k-p')\,f(p'),
\end{equation}
where $f\in \cS(\R^4)$ is the wave function of the incoming electron. The expression for the wave function of the outgoing state can be rewritten in the form
\begin{equation}
 F_\epsilon(p,k)= \ri \int\frac{\rd^4 q}{(2\pi)^3}\, g(q)\, \theta(p^0+k^0-\epsilon q^0)\delta(2p\cdot (k-\epsilon q) + (k-\epsilon q)^2)\, f(p+k-\epsilon q).
\end{equation}
After performing the integral over $q^0$ with the use of the Dirac delta we arrive at
\begin{multline}
 F_\epsilon(p,k)=
 \frac{\ri}{\epsilon}  \int\frac{\rd^3 \vec{q}}{(2\pi)^3}\, 
 \frac{1}{2E_\mass(\vec{p}+\vec{k}-\epsilon\vec{q})}\,
 \\
 \times
 \F{g}\bigg(\frac{-E_\mass(\vec{p})-|\vec{k}|+E_\mass(\vec{p}+\vec{k}-\epsilon\vec{q})}{\epsilon},\vec{q}\bigg)f(p+k+\epsilon q),
\end{multline}
where $E_\mass(\vec p) = \sqrt{\mass^2+|\vec p|^2}$. Using the above equation we prove that the following limits exist
\begin{equation}
 \lim_{\epsilon\searrow0}  \frac{1}{\epsilon}\int\mHm{p}\mHO{k}\, \overline{h(p,k)} F_\epsilon(p,k) 
 ,
 \quad
 \lim_{\epsilon\searrow0} \int\mHm{p}\mHO{k}\,|F_\epsilon(p,k)|^2>0, 
\end{equation}
where $h\in\cS(\Hm\times\HO)$ is arbitrary.
Since $\cS(\Hm\times\HO)$ is dense in the Hilbert space $L^2(\Hm\times\HO,\rd\mu_\mass\times\rd\mu_0)$ we obtain
\begin{equation}
 \wlim_{\epsilon\searrow0}  F_\epsilon = 0,~~~\textrm{and}~~~
 \lim_{\epsilon\searrow0}  F_\epsilon ~\textrm{does not exist},
\end{equation}
where $\wlim$ and $\lim$ stand for the weak and strong limits in $L^2(\Hm\times\HO,\rd\mu_\mass\times\rd\mu_0)$, respectively. The same result holds in the case of QED. The non-existence of the of the above adiabatic limits is a manifestation of the infinite photon emission described in Sec.~\ref{sec:simple_model}.

The adiabatic limit of the first order correction to the standard scattering matrix does not exist also in the massless $\varphi^k$ theory with $k\leq 4$. Consider the wave function of the unphysical process of the decay of the vacuum into $m$ massless particles
\begin{equation}\label{eq:infrared_problem_phi_k}
 F_\epsilon(k_1,\ldots,k_k) = (k_1,\ldots,k_k|S^{[1]}(g_\epsilon)\Omega) = \F{g}_\epsilon(k_1+\ldots+k_k).
\end{equation}
One easily checks that the $L^2$ norm of the above wave-function does not converge to zero in the limit $\epsilon\searrow0$ if $k\leq 4$. Since $F_\epsilon$ converges weakly to zero the adiabatic limit $F_\epsilon$ does not exist in $L^2(\HO^{\times 4},\mu_0^{\times4})$. Note that in the case of the scalar model in higher orders of perturbation theory there appear effective vertices of the above type with any $k$. However, for $k\leq 4$ one can use the renormalization freedom to make these contributions IR-finite (e.g. the case $k=2$ is considered in Sec.~\ref{sec:vacuum_polarization}). 


\subsubsection{Problems in higher orders of perturbation theory}\label{sec:problems_higher}

To illustrate the IR problem at the second order of perturbation theory let us consider the scattering of two electrons. By analogy to QED we call this process the M{\o}ller scattering. We will show that even the matrix elements of this correction between regular states are ill-defined in the adiabatic limit. The relevant Feynman diagrams are depicted in Figure~\ref{fig:moller}: (A), (B) in Sec.~\ref{sec:moller}. For simplicity let us assume that the wave function of the two incoming electrons $f\in\cS(\Hm\times\Hm)$ is real-valued and supported outside the set of coinciding momenta. Then the wave function of the two outgoing electrons is given by
\begin{multline}
 F_\epsilon(p_1,p_2) := \int\mHm{p'_1}\mHm{p'_2}\, (p_1,p_2|S^{[2]}(g_\epsilon)|p'_1,p'_2)\,f(p'_1,p'_2)
 \\
 =-\ri\int \mHm{p'_1}\mHm{p'_2}\frac{\rd^4 k}{(2\pi)^4}\, 
 \F{g}_\epsilon(p_2-k-p'_2)\F{g}_\epsilon(p_1+k-p'_1)\,\frac{1}{k^2+\ri 0}\, f(p'_1,p'_2).
\end{multline}
Using the above expression one shows that $\Re F_\epsilon$ converges in $L^2(\Hm^{\times2},\rd\mu_\mass^{\times2})$ both strongly and weakly but the limit depends on the choice of $g$. This implies that the adiabatic limit of $\Re F_\epsilon$ does not exist (there is no distinguished switching function~$g$). Moreover, the $L^2(\Hm\times\Hm)$ norm of $\Im F_\epsilon$ diverges like $\log|\epsilon|$. Note that $\F{g}_\epsilon(q)$ converges to $(2\pi)\delta(q)$ in $\cS'(\R^4)$. However, after replacing $\F{g}_\epsilon$ by its limit in the expression for $\Im F_\epsilon$ we obtain an ill-defined integral
\begin{equation}
 \frac{1}{8}\int \frac{\rd^4 k}{(2\pi)^2}\,
 \,\frac{1}{k^2}\,
 \delta((p_1+p_2)\cdot k)
 \delta((p_1-p_2)\cdot k-k^2)\,
 f(p_1+k,p_2-k). 
\end{equation}
The above expression corresponds formally to the infinite Coulomb phase. The second-order correction to the M{\"o}ller scattering in QED has a very similar IR divergence. In both models there are other divergent contributions at the second order. 

The above results are unchanged if one uses a different IR regulator. In general, the matrix elements of the standard scattering operator between regular states are generically well-defined only if there are no contributing Feynman diagrams with internal photon propagators joining the electron lines with momenta close to the mass shell $p^2=\mass^2$~\cite{yennie1961infrared,weinberg1965infrared,weinberg1995quantum}. The modified scattering matrix introduced in Sec.~\ref{sec:modified} is free from all of the above-mentioned IR problems.

\subsubsection{Overview of various approaches}\label{sec:overview_IR}

Because of the infinite Coulomb phase and the infinite photon emission the standard scattering matrix and the standard differential cross sections are not well-defined in QED or the scalar model. In this section we give a brief overview of different approaches that have been developed to solve or circumvent this problem. For a more detailed exposition, we refer the reader to~\cite{morchio1986infrared,strocchi2013introduction,herdegen2017asymptotic,strominger2017lectures}. A thorough analysis of the IR problem in classical electrodynamics can be found in~\cite{herdegen1995long,flato1997asymptotic}.

To avoid the IR problem one usually restricts attention to the so-called infrared safe observables such as the inclusive cross section. The inclusive differential cross section with the threshold $E>0$ is usually computed with the use of the following formula~\cite{weinberg1995quantum} (or various modifications thereof formally equivalent up to terms of order $E$)
\begin{equation}\label{eq:inclusive_overview}
 \rd\sigma^{\mathrm{incl}}_{\underline{p}_1,\underline{p}_2}(E;p_1,\ldots,p_n):=\lim_{\epsilon\searrow 0}
 \sum_{m=0}^\infty\frac{1}{m!}\int_{k_1^0+\ldots+k_m^0\leq E}
 \rd\sigma^\epsilon_{\underline{p}_1,\underline{p}_2}(k_1,\ldots,k_m,p_1,\ldots,p_n),
\end{equation}
where $\underline{p}_1$, $\underline{p}_2$ and $p_1,\ldots,p_n$  are the four-momenta of the incoming and observed outgoing particles (their energies are supposed to be large compared to the threshold $E$), $k_1,\ldots,k_m$ are the four-momenta of the unobserved photons and $\rd\sigma^\epsilon$ is the standard differential cross section in a theory with some IR regulator $\epsilon$ such as: a mass of the photon~\cite{yennie1961infrared}, an artificial lower bound on the photon energies~\cite{weinberg1965infrared,weinberg1995quantum}, $d-4$ in the dimensional regularization~\cite{gastmans1973dimensional,marciano1975dimensional} or the adiabatic cutoff~\cite{dutsch1993vertex,dutsch1993infrared}. The limit $\epsilon\to0$ is supposed to exist if the threshold $E$ is fixed and positive. The construction of the inclusive cross sections requires several idealizations. As we explain in Sec.~\ref{sec:inclusive_cross_section}, their finiteness is a very delicate issue. It relies heavily on the exact energy-momentum conservation and the assumption that one can prepare an incoming state with sharp energy-momentum content. Nevertheless, it is expected that the inclusive cross sections are well defined in arbitrary order of perturbation theory. A reasoning indicating that this indeed may be the case was presented by Yiennie, Fratschi and Suura ~\cite{yennie1961infrared} and later simplified by Weinberg~\cite{weinberg1965infrared,weinberg1995quantum}. So far no rigorous proof of this statement was given. 

Since particle detectors always have a finite sensitivity soft photons with energy below some threshold may always escape undetected. Consequently, states with different content of soft photons are difficult to discriminate experimentally and have to be all taken into account as a possible final states. Thus, the inclusive differential cross sections correspond to quantities that are usually measured experimentally. Nevertheless, they do not provide a fully satisfactory description of the scattering processes. First of all, the physical interpretation of the inclusive cross section is not clear since its definition as a sum of the standard cross sections is meaningless in a theory without the IR regularization. In Sec.~\ref{sec:inclusive_cross_section} we propose a different construction which gives the same numerical predictions as the standard construction. It involves the scattering matrix constructed in the paper and makes sense in an unregulated theory. Let us also stress that the inclusive cross sections do not provide a complete information about the scattering. One can imagine~\cite{herdegen2012infrared,staruszkiewicz1981gauge} a scattering process whose only outcome is a measurable shift of the trajectories of the charged particles. Such a shift corresponds to the change of the phase of the wave function and does not give any contribution to the cross section. 

Let us briefly describe some more satisfactory attempts to solve the IR problem in perturbative QED. In the case of models with only massive particles the matrix elements of the scattering operator can be easily expressed in terms of the Green functions with the use of the LSZ reduction formulas~\cite{lehmann1955formulierung}. Note that the IR problem in the perturbative construction of the Wightman and Green functions is completely under control in most models with massless fields including QED or non-abelian Yang-Mills theories  ~\cite{blanchard1975green,lowenstein1976convergence,breitenlohner1977dimensionally1,breitenlohner1977dimensionally2,steinmann2013,duch2018weak}. Moreover, as shown in~\cite{buchholz1977collision} the standard LSZ limit of the photon field is well-defined. However, this is not the case for the Dirac field which has a non-standard asymptotic behavior. In fact, the Fourier transforms of the perturbative corrections to the interacting Feynman propagator of the electron are logarithmically divergent on the mass shell. Furthermore, it is expected that the full Feynman propagator of the electron defined in a non-perturbative way (assuming that QED can be formulated non-perturbatively) would be less singular on the mass shell than the free Feynman propagator and, in particular, it would not have a pole there~\cite{kibble1968coherent2}. For the above reasons, the application of the standard LSZ formula leads to divergent expressions in perturbation theory whereas non-perturbatively it is expected to produce vanishing scattering matrix elements. The latter statement is consistent with the fact that the scattering of charged particles is always accompanied by the infinite photon emission. The way out is to modify the LSZ reduction formulas by taking into account the non-standard asymptotic behavior of the Dirac field. One can find in the literature a number of interesting proposals for the construction of the asymptotic Dirac field which creates and annihilates the physical electrons  ~\cite{zwanziger1975scattering1,zwanziger1975scattering2,bagan1997infrared,herdegen1998semidirect,bagan2000charges1,bagan2000charges2,morchio2016infrared,collins2019new}. However, the constructions usually involve non-local functionals of interacting fields and it is not clear how to give them a proper mathematical meaning in perturbative QFT. 

Another strategy is based on the observation that the electromagnetic radiation emitted by scattered charged particles can be always accommodated in the Hilbert space of some coherent representation of the electromagnetic field. As argued by Chung and Kibble~\cite{chung1965infrared,kibble1968coherent1}, one can define a scattering matrix elements between appropriately chosen coherent states depending on the asymptotic velocities of charged particles. The drawback of this approach is an involved definition of the space of physical states and the fact that the infinite Coulomb phase, which enters the expression for the scattering matrix elements, if there are two or more charged particles in the incoming or outgoing states, has to be dropped by hand.

In this paper, we follow the approach put forward by Kulish and Faddeev~\cite{kulish1970asymptotic} and later investigated e.g. in  \cite{jauch1976theory,papanicolaou1976infrared,gomez2016asymptotic,kapec2017infrared,hirai2019dressed} (see also~\cite{murota1960radiative}). The method proposed by Kulish and Faddeev is based on the Dollard strategy~\cite{dollard1964asymptotic}. The basic object is the so-called modified scattering matrix which is constructed by comparing the full dynamics of the system with some non-trivial asymptotic dynamics. One of the advantages of this approach is the fact that it takes care of both the infinite Coulomb phase and the infinite photon emission. We describe this approach in more detail in Sec.~\ref{sec:Dollard} and give its mathematically rigorous reformulation in the case of perturbative QFT in Sec.~\ref{sec:def_modified}.

We have discussed above different strategies of solving the IR problem in perturbative QED. Many important results about the IR structure of QED were derived in the framework of axiomatic QFT. These results rely on the assumption that there exists a non-perturbative version of QED which satisfies a number of physically motivated conditions. Let us list the most important results obtained in this framework. As shown by Buchholz~\cite{buchholz1986gauss} states with non-zero electric charge cannot be eingenstates of the mass operator. Consequently, charged particles such as electrons and positrons are not elementary particles in the sense of the Wigner definition~\cite{wigner1939unitary} -- they cannot be identified with vectors in the Hilbert space of some irreducible unitary representation of the universal cover of the Poincar{\'e} group. For the above reasons charged particles are usually called in the literature infraparticles~\cite{schroer1963infrateilchen}. Another known fact is the abundance of superselection sectors. Indeed, using the assumption of locality one proves that the tail of the electromagnetic field which decays in spatial directions like the inverse distance squared is a classical observable characterizing different possible superselection sectors of the theory. Consequently, the sectors in QED are not fully characterize by their total charge. In fact, for each given physically attainable value of the total charge there are uncountably many sectors. The characterization of the space of physical states in QED was given in \cite{buchholz1982physical,buchholz2014new,herdegen1998semidirect}. A related result obtained in the axiomatic framework is the proof that the Lorentz transformation cannot be unitarily implemented in sectors with non-zero total charge~\cite{frohlich1979infrared,frohlich1979charged}.

\section{Construction of S-matrix and interacting fields}\label{sec:modified}

In this section we propose a mathematically rigorous construction of the modified scattering matrix in the scalar model and QED. We use the modified scattering theory~\cite{dollard1964asymptotic,kulish1970asymptotic} and combine it with the method of the adiabatic switching of the interaction proposed by Bogoliubov~\cite{bogoliubov1959introduction}.

\subsection{Modified scattering theory}\label{sec:Dollard}

The standard scattering theory is not applicable to most models with long range interactions. The source of the problem is the fact that even long before or after the scattering event the actual evolution of the system is not well approximated by the free evolution. The standard solution, which was originally proposed by Dollard~\cite{dollard1964asymptotic}, is to compare the evolution of the system with some non-trivial but simple asymptotic evolution. 

In this section we describe the Dollard method~\cite{dollard1964asymptotic} for the Coulomb scattering of a quantum-mechanical particle. We use the notation introduced in Sec.~\ref{sec:Coulomb_potential}. The asymptotic evolution for the model at hand is generated by the following time dependent Dollard Hamiltonian
\begin{equation}\label{eq:def_H_D}
 H_{\textrm{D}}(t) := H_\fr + \frac{e^2}{4\pi}\frac{1}{\left|t\vec p/\mass\right|}.
\end{equation}
Note that the Dollard Hamiltonian is obtained by replacing the position $\vec x$ of the particle in the expression~\eqref{eq:Coulomb_hamiltonian} for full Hamiltonian with $t \vec {v}$, where $\vec v=\vec p/\mass$ is the velocity of the particle. The corresponding unitary evolution operator is given by
\begin{equation}
 U_{\textrm{D}}(t_2,t_1) := \exp\left(-\ri\int_{t_1}^{t_2}\rd t\,H_{\textrm{D}}(t)  \right)
\end{equation}
since the Dollard Hamiltonians at different times commute. The modified M{\o}ller operators $\Omega_{\rout/\rin/\rmod}$ and the modified scattering matrix $S_\rmod$ for the Coulomb scattering are obtained by comparing the full dynamics of the system $U(t):= \exp(-\ri t H)$ with the above asymptotic dynamics
\begin{equation}
 \Omega_{\rout/\rin,\rmod}(\zeta) := \slim_{t\to\pm\infty} U(-t)U_{\textrm{D}}(t,\pm \zeta),
 ~~~~
 S_\rmod(\zeta) := \Omega_{\rout,\rmod}(\zeta)^*\Omega_{\rin,\rmod}(\zeta),
\end{equation}
where $\zeta\in\R_+$ is some fixed parameter which has to be positive in order to avoid a non-integrable singularity in time at zero of the Dollard Hamiltonian. The physical interpretation of the parameter $\zeta$ can be inferred from Eq.~\eqref{eq:asymp_trajectory_Coulomb} describing the long time behavior of the classical particle moving in the Coulomb potential. Let us mention that the modified scattering matrix $S_\rmod(\zeta)$ can be also defined by the weak limit of the expression in Eq.~\eqref{eq:S_mod_formal_intro}. The modified M{\o}ller operators and scattering matrix have the following properties
\begin{equation}\label{eq:modified_properties}
 \Omega_{\rout/\rin,\rmod}(\zeta) U_\fr(t) = U(t) \Omega_{\rout/\rin,\rmod}(\zeta),
 ~~~~~
 S_\rmod(\zeta) U_\fr(t) = U_\fr(t) S_\rmod(\zeta)
\end{equation}
and for $\zeta,\zeta'\in\R_+$
\begin{equation}\label{eq:modified_properties_V}
 \Omega_{\rout/\rin,\rmod}(\zeta') = \Omega_{\rout/\rin,\rmod}(\zeta)V_{\rout/\rin}(\zeta,\zeta'),
 \quad
 S(\zeta) = V_{\rout}(\zeta,\zeta') S(\zeta') V_{\rin}(\zeta',\zeta),
\end{equation}
where the unitary intertwining operators $V_{\rout/\rin}(\zeta,\zeta')$ are defined by
\begin{equation}
 V_{\rout/\rin}(\zeta,\zeta') := \slim_{t\to\pm\infty} U_{\textrm{D}}(\zeta;0,t)U_{\textrm{D}}(\zeta';t,0) = \exp\left(\mp \ri \frac{e^2}{4\pi}\frac{\log(\zeta'/\zeta)}{\left|\vec p/\mass\right|} \right).
\end{equation}

The $\zeta$ dependence of the modified scattering matrix may seem slightly unsatisfactory. Note that different choices of this parameter correspond to different identifications of the asymptotic states as vectors in the Hilbert space $\cH=L^2(\R^3)$. In order to make the $\zeta$ dependence of this identification explicit and construct the scattering matrix that is $\zeta$ independent we define the following spaces $\hat{\cH}_{\rin/\rout}$ of asymptotic outgoing and incoming states,
\begin{equation}
 \hat\Psi\in \hat{\cH}_{\rin/\rout}
 \quad\textrm{iff}\quad
 \hat\Psi:\R_+\to L^2(\R^3)
 \quad\textrm{and}\quad
 \hat\Psi(\zeta)=V_{\rout/\rin}(\zeta,\zeta')\hat\Psi(\zeta').
\end{equation}
The scalar product in $\hat{\cH}_{\rin/\rout}$ is given by
\begin{equation}
 (\hat\Psi|\hat\Psi'):=(\hat\Psi(\zeta)|\hat\Psi'(\zeta))
\end{equation}
and is independent of the choice of $\zeta\in\R_+$ by the unitarity of $V_{\rout/\rin}(\zeta,\zeta')$. In particular, the Hilbert spaces $\hat{\cH}_{\rin/\rout}$ are non-canonically isomorphic to $L^2(\R^3)$.  Using property~\eqref{eq:modified_properties_V} we define the unique M{\o}ller and scattering operators
\begin{equation}
\begin{gathered}
 \hat \Omega_{\rout/\rin}:\,\hat{\cH}_{\rin/\rout}\to L^2(\R^3),
 \quad
 \hat \Omega_{\rout/\rin}\hat\Psi:=\Omega_{\rmod,\rout/\rin}(\zeta)\hat\Psi(\zeta),
 \\
 \hat S:\,\hat{\cH}_{\rin}\to\hat{\cH}_{\rout},
 \quad
 (\hat S\hat\Psi)(\zeta):=S_\rmod(\zeta)\hat\Psi(\zeta).
\end{gathered} 
\end{equation}
Note that all states $\hat \Psi(\zeta)\in L^2(\R^3)$, $\zeta\in\R_+$, correspond to the same interacting state $\hat \Omega_{\rout/\rin}\hat\Psi \in L^2(\R^3)$. In Sec.~\ref{sec:state_space} we present a similar reformulation of the construction of the modified scattering matrix in the case of the scalar model and QED.

The Dollard method of constructing the modified M{\o}ller and scattering operators can be generalized (see e.g.~\cite{derezinski1997scattering}) to systems consisting of an arbitrary number of non-relativistic particles interacting via pair potentials which can be of the long-range type. The modified scattering theory for a single Dirac particle moving in the Coulomb potential was developed in~\cite{dollard1966}. The generalization to other long-range potentials was considered in~\cite{gatel2001scattering,daude2004propagation}. 

The main aim of the present paper is to formulate a mathematically rigorous generalization of the Dollard method which is applicable to models of perturbative QFT such as QED. The Dollard strategy has been applied to QED for the first time by Kulish and Faddeev~\cite{kulish1970asymptotic}. Their approach has been revisited by many authors but has not been put on firm mathematical ground so far. One of the difficulties is the fact that in interacting models of relativistic perturbative QFT it is not possible to define the interacting part of the Hamiltonian of the system as operator on a dense domain in the Fock space\footnote{This is essentially a consequence of the Haag theorem (see e.g.~\cite{streater2000pct}), which implies the non-existence of the interaction picture in QFT with no IR cutoff. Note that in spacetimes of dimension two one can define the Hamiltonian with a spatial IR cutoff. If the dimension is greater than two, then even the IR regularized Hamiltonian is ill defined. Indeed, as a consequence of worse UV properties of the free field in dimensions greater than two the Wick polynomials have to be smeared in both space and time \cite{jaffe1966wick}. Recall that the well-defined Bogoliubov S-matrix with adiabatic cutoff given by Eq.~\eqref{eq:bogoliubov_S_op} involves spacetime smearing of the time-ordered products and does not rely on the existence of the interaction Hamiltonian.}. Ignoring this problem Kulish and Faddeev suggested to construct the scattering matrix in QED with the use of the formula
\begin{multline}\label{eq:S_mod_formal}
 S_\rmod \feq \lim_{\substack{t_1\to-\infty\\ t_2\to+\infty}} \aTexp\left( \ri e \int_{0}^{t_2} \rd t ~ H_{\rD,\rint}^I(t) \right)
 \\
 \times \Texp\left( -\ri e \int_{t_1}^{t_2} \rd t ~ H_\rint^I(t) \right) 
 \,
 \aTexp\left( \ri e \int_{t_1}^{0} \rd t ~ H_{\rD,\rint}^I(t) \right),
\end{multline}
where 
\begin{equation}\label{eq:H_int_I}
 H_\rint^I(t) = U_\fr(-t)H_\rint U_\fr(t),
 \quad
 H_{\rD,\rint}^I(t) = U_\fr(-t)H_{\rD,\rint}(t) U_\fr(t)
\end{equation}
are the interaction part of the total Hamiltonian in the interaction picture and the interaction part of the Dollard Hamiltonian in the interaction picture. The formula~\eqref{eq:S_mod_formal} is formally equivalent to~\eqref{eq:S_mod_formal_intro}. Kulish and Faddeev assumed that 
\begin{equation}
 H^I_\rint(t)\feq \int\rd^3\vec x\,\mathcal{L}(t,\vec x) = 
 \int\rd^3\vec x\,\overline{\psi}(t,\vec x)\gamma^\mu\psi(t,\vec x)A_\mu(t,\vec x)
\end{equation}
and used the following Dollard Hamiltonian 
\begin{equation}\label{eq:qed_D_hamiltonian}
 H^I_{\rD,\rint}(t) \feq
 \int\mHm{p}\mHO{k}\,
 \bigg(\frac{p^\mu}{p^0}\,\rho(p)\, a_\mu^*(k)\,
 \re^{\ri \left(k^0 - \frac{\vec p\cdot \vec k}{p_0}\right) t}  
 +\,\mathrm{h.c.}\,
 \bigg),
\end{equation}
where $\rho(p)$ is defined by~\eqref{eq:def_charge_dist_momentum}. In order to motivate the form of the Dollard Hamiltonian they investigated the asymptotic behavior in time of $H^I_\rint(t)$ using heuristic reasoning based on stationary phase arguments~\cite{kulish1970asymptotic,horan2000asymptotic}. One can make this reasoning precise but only when both $H^I_\rint(t)$ and $H^I_{\rD,\rint}(t)$ are treated as forms on some domain in the Fock space, say $\cD_0$ and not as densely defined operators. This is not sufficient to make sense of the formula~\eqref{eq:S_mod_formal}. Other inconsistencies in the original Kulish and Faddeev proposal have been recently pointed out by Dybalski~\cite{dybalski2017faddeev}.

Our construction of the modified scattering operator in QED and the scalar model presented in Sec.~\ref{sec:def_modified} is based to some extent on the ideas of Kulish and Faddeev~\cite{kulish1970asymptotic} and Morchio and Strocchi~\cite{morchio2016dynamics,morchio2016infrared}. The main challenge is the implementation of the Dollard method in a relativistic model with no UV cutoff. Another difficulty is related to the perturbative nature of QED. Our modified scattering matrix is given by the adiabatic limit of the expression~\eqref{eq:intro_S_mod_g}. The factors $S^\ras_{\rout}(\eta,g) $, $S(g)$ and $S^\ras_{\rin}(\eta,g)$ in that expression play the role of the respective factors in the expression under the limit in Eq.~\eqref{eq:S_mod_formal}. The factors $R(\eta,g)$ and $R(\eta,g)^{-1}$, which have no counterpart in Eq.~\eqref{eq:S_mod_formal}, are  used to implement a coherent transformation to sectors with non-trivial long-range tail of the interacting electromagnetic field $A_\mu$ in the case of QED or the interacting massless field $\varphi$ in the case of the scalar model. In this way we are able to treat on equal footing a large class of superselection sectors that arise naturally for example in the simple model introduced in Sec.~\ref{sec:simple_model}. Moreover, in the case of QED the factors $R(\eta,g)$ and $R(\eta,g)^{-1}$ in Eq.~\eqref{eq:S_mod_formal} are essential to ensure the BRST invariance of the constructed scattering matrix. The factors $S^\ras_{\rout/\rin}(\eta,g)$ are formally given by Eq.~\eqref{eq:D_modifiers_Texp} with $\cL_{\rout/\rin}(x) := A_\mu(x) J^\mu_{\rout/\rin}(\eta;x)$, where the asymptotic currents $J^\mu_{\rout/\rin}(\eta;x)$ are regularized versions\footnote{The regularized currents are smooth and depend on a profile $\eta$ (see Def.~\ref{dfn:profile})} of 
\begin{equation}
 J_\ras^\mu(x) = \int\mHm{p}\,j_\ras^\mu(p;x)\,\rho(p),
 \quad
 j_\ras^\mu(\mass v;x):=v^\mu \int_\R \rd\tau\,\delta\left(x-\tau v\right).
\end{equation}
Note that $j_\ras^\mu(\mass v;x)$ is a current of a point-particle moving with constant velocity $v$. In order to see a relation between our Dollard modifiers $S^\ras_{\rout/\rin}(g)$ and the first and the third factor under the limit in Eq.~\eqref{eq:S_mod_formal} note that formally
\begin{equation}
 H_{\rD,\rint}(t) \feq \int\rd^3\vec x\, A_\mu(t,\vec x)\,J^\mu_\ras(x).
\end{equation}

Let us mention that the ideas of Kulish and Faddeev have been tested in several models of QED that can be defined non-perturbatively. We would like to point the construction of the one-electron states in the Nelson model~\cite{pizzo2005scattering} and non-relativistic QED~\cite{chen2010infraparticle,chen2009infraparticle} and the recent proposal for the LSZ-type formula which in principle can be used to construct many-electron states~\cite{dybalski2017faddeev}. Another interesting result has been recently obtained by Morchio and Strocchi who developed the modified scattering theory for a model of QED~\cite{morchio2016infrared} using the strategy formulated in~\cite{morchio2016dynamics}. The model describes charged particles interacting with each other via a pair potential of the same long range behavior as the Coulomb potential and coupled to the quantized electromagnetic field. The analysis of this model is simplified by the lack of the recoil of charged particles in the emission or absorption of photons. Despite this feature the model retains the key IR properties of QED. 

Let us stress that all of the above-mentioned models have some fixed UV-cutoff and are not covariant under Lorentz transformations. In contrast, the technique developed in this paper works for relativistic models and allows to solve both the IR and UV problem in these models.

\subsection{Asymptotic currents}\label{sec:asymp_currents}

In this section we define the currents $J^\mu_{\rout/\rin}(\eta;x)$, which are used in the definition of the Dollard modifiers $S^\ras_{\rout/\rin}(\eta,g)$, and the current $J^\mu_{\mathrm{sector}}(\eta,x)$, which is used in the definition of the operator $R(\eta,g)$ implementing a transformation to a sector with some non-trivial long range-tail of the massless field. 

In order to motivate our definitions of the asymptotic currents let us consider for a moment the model describing the quantized electromagnetic field coupled to an external classical current introduced in Sec.~\ref{sec:simple_model} (in this paragraph we use the notation introduced in that section). As follows from Eq.~\eqref{eq:long_rang_tail_simple} the long-range tail of the solution $F^{\mu\nu}_\ret$ of the equations of motion of this model~\eqref{eq:maxwell_eqs} with no-incoming radiation condition is correlated to the velocity of the incoming electrons (the velocity superselection problem). In order to avoid this correlation in our construction of the modified scattering matrix, we drop the no-incoming radiation condition. Instead, we choose a solution of the equations of motion~\eqref{eq:maxwell_eqs} of the form
\begin{equation}
 F^{\mu\nu}_{\ret,\rmod} = F^{\mu\nu}_\ret - \left(\begin{matrix}\textrm{free radiation field of}\\J_\rin^\mu \textrm{ depending on } \rho_\rin\end{matrix} \right),
\end{equation}
where $\rho_{\rout/\rin}$, defined by Eq.~\eqref{eq:asymptote_current}, are the future/past asymptotes of the external current $J^\mu$. Assuming that the total charge of the current $J^\mu$ vanishes we choose the asymptotic current $J_\rin^\mu$ (depending only on the past asymptote $\rho_\rin$ of the external current $J^\mu$) in such a way that the long-range tail of $F^{\mu\nu}_{\ret,\rmod}$ vanishes. The past and future LSZ limits of $F^{\mu\nu}_{\ret,\rmod}$, denoted by $F_{\rin}^{\mu\nu}$ and $F_{\rout}^{\mu\nu}$, are equivalent to free fields in some coherent, generically non-Fock, representations which depend only on the asymptotes $\rho_\rin$ and $\rho_\rout$, respectively (note that the representation class of the outgoing LSZ limit of the original field $F^{\mu\nu}_\ret$ given by Eq.~\eqref{eq:simple_model_ret} depends on both $\rho_\rin$ and $\rho_\rout$). The modified S-matrix intertwines the fields
\begin{equation}
 F_\rin^{\mu\nu} +\left(\begin{matrix}\textrm{free radiation field of}\\J_\rin^\mu \textrm{ depending on } \rho_\rin\end{matrix} \right)
 \quad\textrm{and}\quad
 F_\rout^{\mu\nu} -\left(\begin{matrix}\textrm{free radiation field of}\\J_\rout^\mu \textrm{ depending on } \rho_\rout\end{matrix} \right).
\end{equation} 
The generalization to sectors with non-trivial long-range tail of the electromagnetic field involves the use of the sector current $\underline{J}$ depending on the sector measure $\varrho$ introduced in Sec.~\ref{sec:def_modified}. In the quantum theory the asymptotes $\rho_{\rin/\rout}$ are replaced by the operator $\rho$ defined by~\eqref{eq:def_charge_dist_momentum}. We consider families of asymptotic currents $J_{\rout/\rin}^\mu$ parameterized by profiles.

\begin{dfn}\label{dfn:profile}
A profile is a real-valued Schwartz function $\eta\in\cS(\R^4)$ which satisfies the normalization condition $\int\rd^4 x\,\eta(x) \equiv \F{\eta}(0) = 1$.
\end{dfn}

There is no natural distinguished choice of a profile. A profile plays in our construction of the modified scattering matrix in perturbative QFT a similar role to the parameter $\zeta\in\R_+$ that appears in the definition of the modified M{\o}ller operators in Sec.~\ref{sec:Dollard}.

\subsubsection{QED}\label{sec:asymp_currents_qed}

We first introduce the outgoing, incoming and asymptotic numerical currents
\begin{gather}\label{eq:numerical_out_in_as_currents}
 j^\mu_{\rout/\rin}(\eta,\mass v;x):=v^\mu \int_\R \rd\tau\, \theta(\pm\tau) \eta\left(x-\tau v\right),
 \\
 j^\mu_\ras(\eta,\mass v;x):=j_\rout(\eta,\mass v;x) + j_\rin(\eta,\mass v;x) = v^\mu\, \int_\R \rd\tau\, \eta\left(x-\tau v\right),
\end{gather}
which depend on a profile $\eta$ and a four-velocity $v$ and are smooth. The Fourier transforms of the above currents have the following form
\begin{equation}
 \F{j}^\mu_{\rout/\rin}(\eta,\mass v;q)
 =
 \pm\frac{\ri\, v^\mu\,\F{\eta}(q)}{v\cdot q\pm\ri \zerop},
 \quad
 \F{j}^\mu_\ras(\eta,\mass v;q)
 =
 v^\mu\, (2\pi) \delta(v\cdot q) \F{\eta}(q).
\end{equation}
The outgoing, incoming and asymptotic operator-valued currents are defined by
\begin{equation}\label{eq:QED_J_out_in_as}
 J^\mu_{\rout/\rin/\ras}(\eta;x)
 =
 \int_{\Hm}\rd\rho(p)\,j^\mu_{\rout/\rin/\ras}(\eta,p;x),
\end{equation}
where $\rd\rho(p):=\rho(p)\,\mHm{p}$ with $\rho(p)$ defined by~\eqref{eq:def_charge_dist_momentum}. Note that
\begin{equation}\label{eq:div_out_in_currents}
\frac{\partial}{\partial x^\mu} J^\mu_{\ras}(\eta;x) = 0,
 \quad
 \frac{\partial}{\partial x^\mu} J^\mu_{\rout/\rin}(\eta;x) 
 =\pm \eta(x)Q.
\end{equation}

\subsubsection{Scalar model}\label{sec:asymp_currents_scalar}

The functions and operators we introduce in this section are Lorentz scalars. Nevertheless, by analogy to QED we call them currents. Let us first define the outgoing, incoming and asymptotic numerical currents
\begin{gather}\label{eq:j_in_out_scalar}
 j_{\rout/\rin}(\eta,\mass v;x)
 :=
 \int_\R \rd\tau\, \theta(\pm\tau) \eta\left(x-\tau v\right),
 \\
 j_\ras(\eta,\mass v;x)
 :=
 j_\rout(\eta,\mass v;x)+j_\rin(\eta,\mass v;x).
\end{gather}
The Fourier transforms of the above currents have the following form
\begin{equation}
 \F{j}_{\rout/\rin}(\eta,\mass v;q)
 =
 \pm \frac{\ri\, \F{\eta}(q)}{v\cdot q\pm\ri \zerop},
 \quad
 \F{j}_\ras(\eta,\mass v;q) 
 =
 (2\pi) \delta(v\cdot q) \F{\eta}(q).
\end{equation}
The outgoing, incoming and asymptotic operator-valued currents are given by
\begin{equation}\label{eq:asymptotic_currents_operators}
 J_{\rout/\rin/\ras}(\eta;x):=\int_{\Hm}\rd\rho(p)\, j_{\rout/\rin/\ras}(\eta,p;x),
\end{equation}
where $\rd\rho(p):=\rho(p)\,\mHm{p}$ with $\rho(p)$ defined by~\eqref{eq:def_charge_dist_momentum_scalar}.

\subsubsection{Asymptotic behavior of currents}\label{sec:asymp_current_behavior}

In this section we show that the matrix elements of the free currents and the asymptotic currents have the same timelike asymptotic behavior. We consider in detail the case of the scalar model. The obtained results generalize to QED in a straightforward manner. Let us recall that by the analogy to QED we call the operator $\normord{J(x)}\,=\,1/2\,\normord{\psi^2(x)}$ the current in the scalar model. It plays a similar role to the spinor current $\normord{J^\mu(x)}\,=\,\normord{\overline{\psi}(x)\gamma^\mu\psi(x)}$ in QED. Let
\begin{equation}
 |f) = \int \mHm{p}\,f(p)\, b^*(p)\Omega,~~~~~f\in C^\infty_{\mathrm{c}}(\Hm),
\end{equation}
be a state with one electron and let $v$ be a four-velocity. We first determine the asymptote of $\normord{J(x)}$
\begin{multline}
 \lim_{\lambda\to\infty} \lambda^3\,(f|\normord{J(\pm\lambda v)}|f) 
 =\lim_{\lambda\to\infty} \lambda^3\frac{\mass^4}{2}\,\int \rd \mu_1(v_1)\rd \mu_1(v_2)\, 
 \overline{f(\mass v_1)} f(\mass v_2) \re^{\ri \lambda \mass (v_1-v_2)\cdot v} 
 \\
 =
 \frac{\mass}{4(2\pi)^3} |f(\mass v)|^2.
\end{multline}
The last equality follows from the standard result about the asymptotic behavior of a solution of the Klein-Gordon equation (see e.g. Appendix 1 to Section XI.3 in~\cite{reedsimon3}). Next, we find the asymptotes of the asymptotic currents. To this end, we shall use the following expression for the asymptotic numerical current
\begin{multline}
 j_\ras(\eta,\mass u;\lambda v) 
 =  \int_\R \rd\tau\, \eta\left(\lambda (v - (v\cdot u) u) -\tau u\right) 
 \\
 = \int_\R \rd\tau\, \eta\left(- \tau (1+|\vec{u}|^2)^{1/2}, -\lambda (1+|\vec{u}|^2)^{1/2}\vec{u} - \tau \vec{u} \right), 
\end{multline}
which is valid for any four-velocities $v,u$ such that $v=(1,\vec{0})$. We observe that
\begin{multline}
 \lambda^3\frac{\mass}{2}\int \rd\mu_1(u)\, |f(\mass u)|^2  j_\ras(\mass u,\lambda v)
 \\
 =\frac{\mass}{4(2\pi)^3} \int\frac{\rd\tau\rd^3\vec{u}}{(1+|\vec{u}|^2/\lambda^2)^{1/2}} 
 \left|f(\mass (1+|\vec{u}|^2/\lambda^2)^{1/2},\mass \vec{u}/\lambda)\right|^2\,
 \\
 \times
 \eta\left(- \tau (1+|\vec{u}|^2/\lambda^2)^{1/2}, -(1+|\vec{u}|^2/\lambda^2)^{1/2}\vec{u} - \tau \vec{u}/\lambda \right).  
\end{multline}
Hence, by the Lebesgue dominated convergence theorem it holds
\begin{multline}
 \lim_{\lambda\to\infty} \lambda^3\, (f|J_{\rout/\rin}(\eta;\pm\lambda v)|f) 
 = \lim_{\lambda\to\infty} \lambda^3\, (f|J_\ras(\eta;\pm\lambda v)|f) 
 \\
 = \lim_{\lambda\to\infty} \lambda^3 \frac{\mass}{2}\int \rd\mu_1(u)\, |f(\mass u)|^2  j_\ras(\eta,\mass u,\lambda v)
 \\
 = \frac{\mass}{4(2\pi)^3} |f(\mass v)|^2 
 \int\rd\tau\rd^3\vec{u}\,\eta(\tau,\vec u)=
 \frac{\mass}{4(2\pi)^3} |f(\mass v)|^2.
\end{multline}
The last equality is a consequence of the normalization of the profile $\int\rd^4 x\,\eta(x)=1$. Concluding, for $\Psi,\Psi'\in\cD_0$ it holds
\begin{equation}
 \lim_{\lambda\to\infty} \lambda^3\,(\Psi|\normord{J(\pm\lambda v)}\Psi') 
 =\lim_{\lambda\to\infty} \lambda^3\,(\Psi|J_{\rout/\rin}(\eta;\pm\lambda v)\Psi') 
 =
 \frac{\mass^2}{2(2\pi)^3} (\Psi|\rho(\mass v)\Psi').
\end{equation}
In the case of QED it holds
\begin{equation}
 \lim_{\lambda\to\infty} \lambda^3\,(\Psi|\normord{J^\mu(\pm\lambda v)}\Psi') 
 =\lim_{\lambda\to\infty} \lambda^3\,(\Psi|J^\mu_{\rout/\rin}(\eta;\pm\lambda v)\Psi') 
 =
 \frac{\mass^2}{2(2\pi)^3} \,v^\mu\, (\Psi|\rho(\mass v)\Psi').
\end{equation}
We comment on the asymptotic behavior of the interacting current in Sec.~\ref{sec:fields}.

\subsection{Coulomb phase}\label{sec:coulomb_phase}

In this section we investigate the properties of the relativistic Coulomb phase which will be used in our definition of the modified scattering matrix in the scalar model and QED. In order to motivate its definition, stated below, let us consider a system of non-relativistic quantum-mechanical particles interacting via the Coulomb potential. As showed in~\cite{dollard1966adiabatic}, one can construct the modified scattering operator for this systems using the method of the adiabatic switching of the interaction\footnote{Using the adiabatic switching is not necessary~\cite{dollard1964asymptotic}.}. For simplicity, let us assume that there are only two particles of the same mass $\mass$ and denote by $\vec x_1$, $\vec x_2$ and $\vec p_1$, $\vec p_2$ their position and momenta operators. The modified scattering matrix can be obtained by the following adiabatic limit
\begin{equation}\label{eq:dollard_adiabatic}
 S_\rmod \Psi = \slim_{\epsilon\searrow0}\,
 \exp\left(\ri e^2 \Phi_{\epsilon}\right) 
 S_\epsilon 
 \exp\left(\ri e^2 \Phi_{\epsilon}\right)\Psi,
\end{equation}
where $S_\epsilon$ is the standard scattering matrix for particles interacting via the time-dependent potential $\re^{-\epsilon |t|} V(\vec{x}_1-\vec{x_2})$, where $V(\vec{x})$ is the Coulomb potential~\eqref{eq:Coulomb_hamiltonian}, and
\begin{equation}\label{eq:coulomb_non_rel}
 \Phi_{\epsilon}
 :=
 \frac{1}{4\pi}\int_{\frac{\mass}{2|\vec p_1-\vec p_2|^2}}^\infty\rd t\,
 \frac{\re^{-\epsilon|t|}}{|t|\frac{|\vec p_1-\vec p_2|}{\mass}}
\end{equation}
is the non-relativistic Coulomb phase with adiabatic cutoff. The modified scattering operator defined by~\eqref{eq:dollard_adiabatic} coincides with the one constructed using the method described in Sec.~\ref{sec:Dollard}. Note that for non-coinciding momenta of the particles the Coulomb phase~\eqref{eq:coulomb_non_rel} equals
\begin{equation}\label{eq:Coulomb_div_part}
 \frac{1}{4\pi} \frac{\mass}{|\vec p_1-\vec p_2|}\, \log(1/\epsilon)
\end{equation}
up to a term for which the limit $\epsilon\searrow0$ exists and is finite. 
\begin{dfn}
Let $g\in\cS(\R^4)$ be real-valued and $p_1,p_2\in\Hm$. The outgoing and incoming relativistic Coulomb phases with adiabatic cutoff~$g$ are given by
\begin{multline}\label{eq:def_rel_coulomb_phase}
 \Phi_{\rout/\rin}(\eta,\eta',g,p_1,p_2):=
 \\
 \int \rd^4 x \rd^4 y\,g(x)g(y)\,
 D_0^D(x-y)\,j_{\rout/\rin}(\eta,p_1;x) j_{\rout/\rin}(\eta',p_2;y),
\end{multline}
where $j_{\rout/\rin}(\eta,p_1;x)$, $j_{\rout/\rin}(\eta',p_2;y)$ are defined by Eq.~\eqref{eq:j_in_out_scalar} for arbitrary real-valued functions $\eta,\eta'\in\cS(\R^4)$ and $D^D_0$ is the Dirac propagator (cf. Appendix~\ref{sec:propagators}).
\end{dfn}
The phases $\Phi_{\rout/\rin}(\eta,\eta,g_\epsilon,p_1,p_2)$, where $\eta$ is a profile, i.e. $\int\rd^4 x\,\eta(x)=1$, will be used in the definition of the modified scattering matrix in the scalar model. In the case of QED we will use the phases $\frac{p_1\cdot p_2}{\mass^2} \Phi_{\rout/\rin}(\eta,\eta,g_\epsilon,p_1,p_2)$. By the theorem stated below, the later phases are relativistic generalizations of the non-relativistic Coulomb phase~\eqref{eq:coulomb_non_rel}. They appear in the M{\o}ller scattering in QED investigated in Sec.~\ref{sec:problems_higher} and~\ref{sec:moller}.
\begin{thm}\label{thm:coulomb_phase}
Set $g_\epsilon(x):=g(\epsilon x)$, where $g\in\cS(\R^4)$ such that $g(0)=1$.
\\
(A) Let $\eta,\eta'\in\cS(\R^4)$. If $\int\rd^4x\,\eta(x)=0$ or $\int\rd^4x\,\eta'(x)=0$, then for $p_1,p_2\in\Hm$, $p_1\neq p_2$, the pointwise limit
\begin{equation}\label{eq:phi_finite}
 \Phi_{\rout/\rin}(\eta,\eta',p_1,p_2) :=\lim_{\epsilon\searrow0}\,\Phi_{\rout/\rin}(\eta,\eta',g_\epsilon,p_1,p_2)
\end{equation}
exists. For any $\eta\in\cS(\R^4)$ it holds
\begin{equation}\label{eq:phi_finite_trans}
 \Phi_{\rout/\rin}(\eta-\eta_a,\eta+\eta_a,p_1,p_2) =0,
\end{equation}
where $\eta_a(x)=\eta(x-a)$, $a\in\R^4$. Let $h\in\cS(\Hm\times\Hm)$ with support outside the diagonal. The functions 
\begin{equation}
 \Hm\times\Hm\ni(p_1,p_2)\mapsto h(p_1,p_2)\Phi_{\rout/\rin}(\eta,\eta',p_1,p_2) \in\C
\end{equation}
belong to the Schwartz class. 
\\
(B) If $\eta$ and $\eta'$ are profiles, i.e. $\eta,\eta'\in\cS(\R^4)$ and $\int\rd^4x\,\eta(x)=\int\rd^4x\,\eta'(x)=1$, then 
\begin{equation}\label{eq:phi_div_limit}
 \Phi_{\rout/\rin}(\eta,\eta',g_\epsilon,p_1,p_2) - 
 \frac{1}{4\pi}\left(\frac{(p_1\cdot p_2)^2}{\mass^2}-1\right)^{-1/2} \log(1/\epsilon)
\end{equation} 
converges as $\epsilon\searrow0$ for $p_1,p_2\in\Hm$, $p_1\neq p_2$ (the limit depends on the choice of the switching function $g\in\cS(\R^4)$). 
\end{thm}
\begin{proof}
Let $v_1,v_2$ be four-velocities. First, we note that
\begin{multline}\label{eq:thm_coulomb_phase1}
 \Phi_{\rout/\rin}(\eta,\eta',g_\epsilon,\mass v_1,\mass v_2)
 =\frac{1}{4\pi} \int\rd\tau_1\rd\tau_2\rd^4 x\rd^4 y\,\theta(\pm\tau_1)\theta(\pm\tau_2)\,\delta((x-y)^2)
 \\
 \times\,g_\epsilon(x) g_\epsilon(y)
 \,\eta(x-v_1 \tau_1)\,\eta'(y-v_2\tau_2).
\end{multline}
In the proof of part (A) of the theorem we will use two identities:
\begin{multline}\label{eq:thm_coulomb_phase2}
 \lim_{\epsilon\searrow0}\,\frac{1}{4\pi} \int\rd\tau_1\rd\tau_2\rd^4 x\rd^4 y\,\theta(\pm\tau_1)\theta(\pm\tau_2)\,\delta((x-y)^2)
 \\
 \times [g_\epsilon(x) g_\epsilon(y)
 -g_\epsilon(\tau_1 v_1) g_\epsilon(\tau_2 v_2)]
 \, \eta(x-v_1 \tau_1)\,\eta'(y-v_2\tau_2)=0
\end{multline}
and
\begin{multline}\label{eq:thm_coulomb_phase3}
 \lim_{\epsilon\searrow0}\,\frac{1}{4\pi} \int\rd\tau_1\rd\tau_2\rd^4 x \rd^4 y\,g_\epsilon(\tau_1 v_1) g_\epsilon(\tau_2 v_2)\,\eta(x)\,\eta'(y)
 \\
 \times [\theta(\pm\tau_1)\theta(\pm\tau_2)\delta((x-y+\tau_1 v_1-\tau_2 v_2)^2)-\theta(\pm\tau_1-1)\theta(\pm\tau_2-1)\delta((\tau_1 v_1-\tau_2 v_2)^2)]
 \\
 =\frac{1}{4\pi} \int\rd^4 x \rd^4 y\,\eta(x)\,\eta(x\mp y)
 \int \rd\tau_1\rd\tau_2\,
 \big[\theta(\tau_1)\theta(\tau_2)\delta((y+\tau_1 v_1-\tau_2 v_2)^2)
 \\
 -\theta(\tau_1-1)\theta(\tau_2-1)\delta((\tau_1 v_1-\tau_2 v_2)^2)\big],
\end{multline}
which are proved below. The above limits are pointwise and exist for $v_1\neq v_2$. Note that $\Phi_{\rout/\rin}(\eta,\eta',\mass v_1,\mass v_2)$ defined by~\eqref{eq:phi_finite} coincides with the expression on the RHS of~\eqref{eq:thm_coulomb_phase3} and is a smooth function of the velocities $v_1$ and $v_2$ outside the diagonal. 

To prove~\eqref{eq:thm_coulomb_phase2} we note that
\begin{multline}
 |g_\epsilon(x) g_\epsilon(y)-g_\epsilon(\tau_1 v_1) g_\epsilon(\tau_2 v_2)| 
 \\
 \leq |g_\epsilon(x)| |g_\epsilon(y)-g_\epsilon(\tau_2 v_2)| + |g_\epsilon(x)-g_\epsilon(\tau_1 v_1)| |g_\epsilon(\tau_2 v_2)|
\end{multline}
and
\begin{equation}
 |x|^{1/2}\,|g_\epsilon(x)|\frac{|g_\epsilon(y)-g_\epsilon(\tau_2 v_2)|}{\epsilon^{1/2}|y-\tau_2 v_2|},~
 |\tau_2 v_2|^{1/2}\,|g_\epsilon(\tau_2 v_2)| \frac{|g_\epsilon(x)-g_\epsilon(\tau_1 v_1)|}{\epsilon^{1/2}|x-\tau_1 v_1|}\leq \const.
\end{equation}
Eq.~\eqref{eq:thm_coulomb_phase2} is a consequence of the following estimates:
\begin{multline}
 \epsilon^{1/2} \int\rd\tau_1\rd\tau_2\rd^4 x \rd^4 y\,\delta((x-y)^2)
 \,\frac{|x-\tau_1 v_1|}{|x|^{1/2}} |\eta(x-\tau_1v_1)|\,|\eta'(y-\tau_2v_2)|
 \\
 \leq 
 \epsilon^{1/2} \int\rd\tau_1\rd^4 x \rd^4 y\,
 \,\frac{|x|}{|\vec x-\vec y+\tau_1\vec v_1|\,|x+\tau_1 v_1|^{1/2}} |\eta(x)|\,|\eta'(y)|
 \leq \const\, \epsilon^{1/2} 
\end{multline}
and
\begin{multline}
 \epsilon^{1/2} \int\rd\tau_1\rd\tau_2\rd^4 x \rd^4 y\,\delta((x-y)^2)
 \,\frac{|y-\tau_2 v_2|}{|\tau_1v_1|^{1/2}} |\eta(x-\tau_1v_1)|\,|\eta'(y-\tau_2v_2)|
 \\
 \leq 
 \epsilon^{1/2} \int\rd\tau_2\rd^4 x \rd^4 y\,
 \,\frac{|y|}{|\vec x-\vec y-\tau_2\vec v_2|\,|\tau_2 v_2|^{1/2}} |\eta(x)|\,|\eta'(y)|
 \leq \const\, \epsilon^{1/2},
\end{multline}
where in each of the above bounds we first changed the integration variables $(x,y)\mapsto(x+\tau_1v_1,y+\tau_2v_2)$ and subsequently performed the integral over $\tau_2$ or $\tau_1$ using the Dirac delta and assuming that $v_2=(1,0,0,0)$ or $v_1=(1,0,0,0)$. 

In order to show Eq.~\eqref{eq:thm_coulomb_phase3} it is enough to note that the expression under the limit on the LHS of this equation coincides with
\begin{multline}
 \sum_{\pm}\frac{1}{8\pi} \int\rd\tau_1\rd^4 x \rd^4 y\,\eta(x)\,\eta'(y)
 \\
 \bigg[\theta(\pm\tau_1)\theta(\pm\tau^\pm_2(\tau_1,x-y))\,
 g_\epsilon(\tau_1 v_1) g_\epsilon(\tau^\pm_2 v_2(\tau_1,x-y))
 \frac{1}{|\vec{x}-\vec{y}+\tau_1\vec v_1|}
 \\
 -
 \theta(\pm\tau_1-1)\theta(\pm\tau^\pm_2(\tau_1,0)-1)\,
 g_\epsilon(\tau_1 v_1) g_\epsilon(\tau_2^\pm(\tau_1,0) v_2)
 \frac{1}{|\tau_1\vec v_1|}\bigg],
\end{multline}
where
\begin{equation}
 \tau^\pm_2(\tau_1,z)=z^0-\tau_1v^0_1 \pm |\vec{z}+\tau_1\vec v_1|.
\end{equation}
By~\eqref{eq:thm_coulomb_phase2} and~\eqref{eq:thm_coulomb_phase3} in order to prove part (B) of the theorem it is enough to study the expression 
\begin{multline}\label{eq:thm_coulomb_phase4}
 \frac{1}{4\pi} \int\rd\tau_1\rd\tau_2\,g_\epsilon(\tau_1 v_1) g_\epsilon(\tau_2 v_2)\,\theta(\pm\tau_1-1)\theta(\pm\tau_2-1)\,\delta((\tau_1 v_1-\tau_2 v_2)^2)
 \\
 =\sum_{\pm}\frac{1}{8\pi} \int\rd\tau_1\,
 \theta(\pm\tau_1-\epsilon)\theta(\pm\tau^\pm_2(\tau_1,0)-\epsilon)\,
 g(\tau_1 v_1) g(\tau_2^\pm(\tau_1,0) v_2)
 \frac{1}{|\tau_1\vec v_1|},
\end{multline}
where we assumed that $v_2=(1,0,0,0)$. Using the fact that $g(0)=1$ it is easy to see that the limit $\epsilon\searrow0$ of the difference between the above expression and $\frac{1}{4\pi|\vec v_1|} \log(1/\epsilon)$ exists (and depends on the choice of $g$). This proves part (B) of the theorem.
\end{proof}

\subsection{Domains in Fock space}\label{sec:domains}

In the perturbation theory the scattering matrix and the interacting fields are formal power series whose coefficients are densely-defined unbounded operators. In this section we introduce domains in the Fock space that we use in our construction. 

For a moment let us restrict attention to the scalar model. We begin by fixing some terminology. We use the following standard shorthand notation for the improper state with $n+m$ particles of definite momenta
\begin{equation}
 |p_1,\ldots,p_n,k_1,\ldots,k_m) := b^*(p_1)\ldots b^*(p_n)a^*(k_1)\ldots a^*(k_m)\Omega,
\end{equation}
where $\Omega$ is the vacuum state. A state $\Psi$ in the Fock space $\cH$ of the form
\begin{multline}\label{eq:state_psi}
 \Psi= \int\mHm{p_1}\ldots\mHm{p_n}\mHO{k_1}\ldots\mHO{k_m}\,
 \\
 h(p_1,\ldots,p_n,k_1,\ldots,k_m)\,
 |p_1,\ldots,p_n,k_1,\ldots,k_m),
\end{multline}
where $h\in L^2(\Hm^{\times n}\times\HO^{\times m},\rd\mu_\mass^{\times n}\times\rd\mu_0^{\times m})$, $n,m\in\N_0$, is called a state with definite particle content. The square-integrable function $h$ is called the wave function of the state $\Psi$. Without loss of generality we assume that the function $h(p_1,\ldots,p_n,k_1,\ldots,k_m)$ is symmetric under permutations of the momenta $p_1,\ldots,p_n$ and $k_1,\ldots,k_m$. The domain $\cD_\fin\subset\cH$ by definition consists of finite linear combinations of states with definite particle content. All the domains that we will consider are subsets of $\cD_\fin$. In particular this is true for the domain $\cD_0$ defined by~\eqref{eq:dom}. 

The domain $\cD_1$ consists of linear combinations of states $\Psi$ given by~\eqref{eq:state_psi} such that
\begin{equation}
 h(p_1,\ldots,p_n,k_1,\ldots,k_m) =H(\vec{p}_1,\ldots,\vec{p}_n,k_1^0,\hat{k}_1,\ldots,k_m^0,\hat{k}_m) 
\end{equation}
for some $H\in \cS(\R^{3n}\times(\R\times S^2)^m)$, where $\hat{k}=\vec{k}/|\vec{k}|$. By definition, the space $\cS(\R^{3n}\times(\R\times S^2)^m)$ consists of smooth functions $H\in C^\infty(\R^{3n}\times(\R\times S^2)^m)$ such that the Schwartz seminorms of the functions $$\R^{3n}\times\R^m\ni (\vec{p}_1,\ldots,\vec{p}_n,k^0_1,\ldots,k_m^0)\mapsto H(\vec{p}_1,\ldots,\vec{p}_n,k^0_1,\hat{k}_1,\ldots,k_m^0,\hat{k}_m)$$ are uniformly bounded for $\hat{k}_1,\ldots,\hat{k}_m\in(S^2)^m$. We have $\cD_1\supsetneq \cD_0$. An example of $\Psi\in\cD_1$, $\Psi\notin \cD_0$ is $\Psi = \int\mHO{k} \frac{\F{\eta}(k)}{p\cdot k} |k)$, where $\F{\eta}\in\cS(\R^4)$, $\F{\eta}(0)=0$ and $p\in\Hm$ is fixed. 

We will also need the domain $\cD_2$ which consists of those states belonging to $\cD_1$ whose wave functions vanish when the momenta of any pair of massive particles are sufficiently close to each other. More precisely, the domain $\cD_2$ consists of the states from $\cD_1$ whose wave functions $h$ are supported away from
\begin{equation}
 \{ (p_1,\ldots,p_n,k_1,\ldots,k_m)\,:\, \exists_{i\neq j} p_i=p_j \}.
\end{equation}
It holds $\cD_1 \supsetneq \cD_2$. The domains $\cD_2$ and $\cD_0$ are not comparable.

The above definitions of the domains $\cD_1$, $\cD_2$ can be adapted in a natural way to the case of QED. The distinction between massive and massless particles in the definitions above remains the same. However, one has to keep in mind the fact that in QED there are more types of massless particles (four polarizations of photons and two types of ghosts) as well as massive particles (two spin polarizations of electrons and positrons). Moreover, since ghosts as well as electrons and positrons are fermions the wave functions are assumed to be antisymmetric under permutations of the momenta of these particles. 

Finally, let us note that the domains $\cD_0$, $\cD_1$, $\cD_2$ are invariant under the action of the standard Fock representation of the Poincar{\'e} group $U(a,\Lambda)$ as well as the free BRST charge $Q_\mathrm{BRST}$ and the free ghost charge $Q_{\mathrm{gh}}$.

\subsection{Definition of modified S-matrix and modified interacting fields}\label{sec:def_modified}

In this section we present our proposal for the definition of the modified scattering matrix and the modified interacting fields in the scalar model and QED. We use the asymptotic currents introduced in Sec.~\ref{sec:asymp_currents} which depend on the profile $\eta$, that is a real-valued Schwartz function whose integral over spacetime is normalized to one (cf. Def.~\ref{dfn:profile}).

\begin{dfn}\label{def:D_modifiers}
The outgoing and incoming asymptotic vertices in QED are given by
\begin{equation}\label{eq:def_asymp_int_vertices}
 \cL_{\rout/\rin}(\eta;x):= J^\mu_{\rout/\rin}(\eta;x)A_\mu(x).
\end{equation}
The outgoing and incoming Dollard modifiers with adiabatic cutoff are given by
\begin{multline}\label{eq:def_modifiers}
 S^\ras_{\rout/\rin}(\eta,g)
 :=
 \exp\left(-\ri e\int \rd^4 x \, g(x)\, \cL_{\rout/\rin}(\eta;x)\right)
 \\ 
 \times\exp\left(\ri\frac{e^2}{2}\!\int\! \rd^4 x \rd^4 y\,g(x)g(y)
 g_{\mu\nu}D_0^D(x-y)\normord{J^\mu_{\rout/\rin}(\eta;x) J^\nu_{\rout/\rin}(\eta;y)}\right),
\end{multline}
where $g\in\cS(\R^4)$ is a switching function and $D_0^D$ is the Dirac propagator. In the case of the scalar model we replaced in the above formulas $A_\mu$ with $\varphi$, the vector currents with their scalar counterparts and $g_{\mu\nu}$ with $-1$. 
\end{dfn}
Let us consider for a moment the scalar model. First, note that it holds
\begin{gather}
 [\cL_{\rout/\rin}(\eta;x),\cL_{\rout/\rin}(\eta;y)] 
 = -\ri D_0(x-y)\,J_{\rout/\rin}(\eta;x) J_{\rout/\rin}(\eta;y),
 \\
 [[\cL_{\rout/\rin}(\eta;x),\cL_{\rout/\rin}(\eta;y)],\cL_{\rout/\rin}(\eta;z)] =0.
\end{gather}
where $D_0$ is the commutator function. Consequently, by the formula~\eqref{eq:magnus}, known as the Magnus expansion, the Dollard modifiers are closely related to the expression~\eqref{eq:D_modifiers_Texp}. Our Dollard modifiers differ from it only due to the normal-ordering of asymptotic currents in the second line of~\eqref{eq:def_modifiers}. The normal-ordering can be considered a renormalization of the self-energy of electron and is needed for the existence of the adiabatic limit of our modified scattering matrix (cf. Sec.~\ref{sec:2ed_se}). Let us present a more explicit expression for the Dollard modifiers 
\begin{multline}
 S^\ras_{\rout/\rin}(\eta,g_\epsilon)\,|p_1,\ldots,p_n,k_1,\ldots,k_m)
 \\
 =
 \prod_{k=1}^n W_{\rout/\rin}(\eta,g_\epsilon,p_k) 
 \prod_{\substack{k,l=1\\k\neq l}}^n 
 \exp\left(-\ri\frac{e^2}{2} \Phi_{\rout/\rin}(\eta,\eta,g_\epsilon,p_k,p_l)\right)
 |p_1,\ldots,p_n,k_1,\ldots,k_m),
\end{multline}
where
\begin{equation}
 W_{\rout/\rin}(\eta,g_\epsilon,p) 
 := 
 \exp\left(-\ri e\int \rd^4 x \, g_\epsilon(x) j_{\rout/\rin}(\eta,p;x) \varphi(x)  \right)
\end{equation}
and the functions~$\Phi_{\rout/\rin}(\eta,\eta,g_\epsilon,p_k,p_l)$ are the Coulomb phases defined by~\eqref{eq:def_rel_coulomb_phase}. The Weyl operators $W_{\rout/\rin}(\eta,g_\epsilon,p)$ are responsible for the generation of clouds of coherent photons surrounding massive particles. For any $p\in\Hm$ and $\epsilon>0$ the vector $W_{\rout/\rin}(\eta,g_\epsilon,p)\Omega$ is a coherent state in the Fock space. The expected number of photons in the state $W_{\rout/\rin}(\eta,g_\epsilon,p)\Omega$ diverges in the limit $\epsilon\searrow0$ and, in fact, this limit does not exist in the Fock space. In the adiabatic limit the operators $W_{\rout/\rin}(\eta,g_\epsilon,p)$ effectively remove the long-range tail of the modified retarded massless field $\varphi_{\ret,\rmod}(x)$ and make possible a definition of the scalar model in the vacuum sector. In order to define the scalar model in different sectors one has to use non-trivial operators $R(\eta,g)$ introduced below. Let us also note that
\begin{equation}\label{eq:dollard_no_massive}
 S^\ras_{\rout/\rin}(\eta,g)\Psi = \Psi,
 ~~~~
 (S^\ras_{\rout/\rin}(\eta,g))^{-1}\Psi = \Psi
\end{equation}
whenever the state $\Psi$ does not contain massive particles. The above identities follow immediately from the fact that $J_{\rout/\rin}(\eta;x)\Psi$ vanishes for states $\Psi$ of the above form. All of the above remarks (with suitable modifications) apply also to the case of QED.

\begin{dfn}\label{def:varrho}
A sector measure $\varrho$ is a compactly supported signed measure on the mass hyperboloid $\Hm$ with finite total variation $|\varrho|$. 
\end{dfn}

\begin{dfn}\label{def:sector}
In the case of QED we fix some sector measures $\varrho_0$, $\varrho_1$ such that $\varrho_0(\Hm)=0$ and $\varrho_1(\Hm)=1$ and set $\hat\varrho = \varrho_0 + Q\varrho_1$, where $Q$ is the electric charge. In the case of the scalar model $Q=0$ and $\hat\varrho=\varrho_0$ is not restricted by the condition $\varrho_0(\Hm)=0$. In both QED and the scalar model we define the sector currents
\begin{equation}
 \underline{J}^\mu(\eta,x) := \int_{\Hm}\rd\hat\varrho(p)\, j^\mu_\rin(\eta,p;x),
 \quad
 \underline{J}(\eta,x) := \int_{\Hm}\rd\hat\varrho(p)\, j_\rin(\eta,p;x),
\end{equation}
the sector vertices
\begin{equation}\label{eq:def_vertex_sector_qed}
 \underline{\cL}(\eta;x) := \underline{J}^\mu(\eta;x)\,A_\mu(x),
 \quad
 \underline{\cL}(\eta;x) := \underline{J}(\eta;x)\,\varphi(x)
\end{equation}
and the sector operator
\begin{equation}\label{eq:def_operator_R}
 R(\eta,g) := \exp\left(-\ri e\int \rd^4 x \, g(x)\,\underline{\cL}(\eta;x) \right),
\end{equation}
where $g\in\cS(\R^4)$ is a switching function. 
\end{dfn}
One could replace the current $j_\rin(\eta,p;x)$ used in the above definition with a different current whose divergence coincides with $-\eta(x)$. In the case of the scalar model one can set $\hat\varrho=0$ which is our standard choice. As we will show in Sec.~\ref{sec:fields}, this choice results in the construction of the modified scattering matrix and interacting fields in the vacuum sector with trivial long-range tail of the interacting massless field. In the case of QED, the constructed space of asymptotic states contains the vacuum sector if $\varrho_0=0$. The conditions $\varrho_0(\Hm)=0$ and $\varrho_1(\Hm)=1$ are necessary to respect the BRST invariance and the Gauss law. Arguably the simplest choice of $\hat\varrho$ in the case of QED is $\varrho_0=0$ and $\varrho_1=\varrho_{\mass \mathrm{v}}$, where $\mathrm{v}$ is a four-velocity and $\varrho_p$ is the Dirac measure at $p\in\Hm$. As we demonstrate in Sec.~\ref{sec:fields}, this choice of $\hat\varrho$ corresponds to the super-selection sectors with nonzero electric charge in which the flux of the electric field is independent of the spatial direction in the reference frame of some observer whose four-velocities coincides with~$\mathrm{v}$. Note that in the case of QED we are considering simultaneously a family of sectors parametrized by the total electric charge $q\in\textrm{sp}(Q)=\Z$. One could also carry out the construction in one particular sector. To this end, it it enough to replace $R(\eta,g)$ with $1_q(Q)R(\eta,g)$ where $1_q(Q)$ is the operator acting in the Fock space projecting to the eigenspace of the charge operator $Q$ with an eigenvalue $q\in\Z$.

\begin{dfn}\label{def:modified_S_fields}
Let $g$ be a switching function, $\eta$ be a profile and $\hat\varrho$ be a sector measure. The modified scattering operator with adiabatic cutoff is defined by
\begin{equation}\label{eq:mod_S_op}
 S_{\rmod}(\eta,g) :=  R(\eta,g) S^\ras_\rout(\eta,g)S(g)S^\ras_\rin(\eta,g) R(\eta,g)^{-1},
\end{equation}
where $S(g)$ is the standard Bogoliubov scattering operator given by~\eqref{eq:bogoliubov_S_op}. The modified extended scattering operator with adiabatic cutoff $S_{\rmod}(\eta,g;h)$ is obtained by the same formula with $S(g)$ replaced with $S(g;h)$ given by~\eqref{eq:bogoliubov_S_op_extended}. The modified interacting fields with adiabatic cutoff are defined by
\begin{equation}\label{eq:mod_fields}
\begin{aligned}
 C_{\ret,\rmod}(\eta,g;x) &:= (-\ri) \frac{\delta}{\delta h(x)}
 S_{\rmod}(\eta,g)^{-1} S_{\rmod}(\eta,g;h)\bigg|_{h=0},
 \\
 C_{\adv,\rmod}(\eta,g;x) &:= (-\ri) \frac{\delta}{\delta h(x)}
 S_{\rmod}(\eta,g;h) S_{\rmod}(\eta,g)^{-1}\bigg|_{h=0}.  
\end{aligned}
\end{equation}
\end{dfn}

Observe that it holds
\begin{equation}\label{eq:mod_fields_relation}
\begin{aligned}
 C_{\ret,\rmod}(\eta,g;h)&=R(\eta,g) S^\ras_\rin(\eta,g)^{-1}
 C_\ret(g;h) S^\ras_\rin(\eta,g) R(\eta,g)^{-1},
 \\
 C_{\adv,\rmod}(\eta,g;h)&= R(\eta,g) S^\ras_\rout(\eta,g)
 C_\adv(g;h) S^\ras_\rout(\eta,g)^{-1}R(\eta,g)^{-1},
\end{aligned}
\end{equation}
where $C_{\ret/\adv}(g;h)$ are the standard retarded and advanced interacting fields introduced in Sec.~\ref{sec:bogoliubov}. Note that $R(\eta,g)\Omega = \Omega$ in the case $\varrho_0=0$. Restricting attention to this case and using Eqs.~\eqref{eq:dollard_no_massive}, \eqref{eq:mod_fields_relation} we show that the vacuum expectation value of a product of the modified interacting fields coincides with the vacuum expectation value of the corresponding product of the standard interacting fields~\eqref{eq:bogoliubov_fields}. Consequently, the Wightman and Green functions of the modified retarded or advanced fields exist and coincide with the standard Wightman and Green functions. The latter were constructed in~\cite{blanchard1975green,lowenstein1976convergence,breitenlohner1977dimensionally1,breitenlohner1977dimensionally2,steinmann2013,duch2018weak} in the case of QED and in~\cite{duch2018weak} in the case of the scalar model.

The following theorem makes precise in which sense objects introduced so far in this section are defined.
\begin{thm}\label{thm:mod_definition}
(A) Expressions~\eqref{eq:def_modifiers}, \eqref{eq:def_operator_R} define the Dollard modifiers and the sector operator as unitary operators $S^\ras_{\rout/\rin}(\eta,g), R(\eta,g)\in\mathcal{B}(\cH)$ acting in the Fock space $\cH$ and as formal power series $S^\ras_{\rout/\rin}(\eta,g), R(\eta,g)\in L(\cD)\llbracket e\rrbracket$, where $\cD\in\{\cD_0,\cD_1,\cD_2\}$. 
\\
(B) Expressions~\eqref{eq:mod_S_op} and~\eqref{eq:mod_fields} define formal power series
\begin{equation}\label{eq:mod_S_op_fields}
 S_\rmod(\eta,g),\,C_{\ret/\adv,\rmod}(\eta,g;h)\in L(\cD)\llbracket e\rrbracket,
\end{equation}
where $\cD\in\{\cD_0,\cD_1\}$. 
\\
(C) Let $\cD\in\{\cD_0,\cD_1\}$. The equation
\begin{equation}\label{eq:S_mod_sym_dist}
 S^{[n]}_\rmod(\eta,g) = \int\mP{q_1}\ldots\mP{q_n}\,\F{g}(-q_1)\ldots\F{g}(-q_n)\, S^{[n]}_\rmod(\eta,q_1,\ldots,q_n) 
\end{equation}
with arbitrary $g\in\cS(\R^4)$ defines uniquely the symmetric operator-valued Schwartz distributions $S^{[n]}_\rmod(\eta,q_1,\ldots,q_n) \in \cS'(\R^{4n},L(\cD))$. One defines analogously symmetric distributions $C^{[n]}_{\ret/\adv,\rmod}(\eta,q_1,\ldots,q_n;h)\in\cS'(\R^{4n},L(\cD))$.
\end{thm}

The Dollard modifiers and the sector operator can be viewed as either unitary operators or formal power series. By abusing the notation in what follows we use the same symbols $S^\ras_{\rout/\rin}(\eta,g)$, $R(\eta,g)$ to denote different objects: the unitary operator in $\cH$ or the formal power series in $L(\cD)\llbracket e\rrbracket$. The meaning of the symbol should be clear from the context. Since the Bogoliubov S-matrix is a formal power series the modified scattering matrix and modified interacting fields can only be defined as formal power series. In general, if at least of the objects in a given formula can only be interpreted as a formal power series, than all objects in this formula should be interpreted in this sense\footnote{We work mostly with formal power series. The fact that the Dollard modifiers $S^\ras_{\rout/\rin}(\eta,g)$ and the sector operator $R(\eta,g)$ can be given a non-perturbative meaning will be only used in Sec.~\ref{sec:energy_momentum} and~\ref{sec:physical_interpretation}.}. Note that~\eqref{eq:mod_S_op_fields} does not hold for $\cD=\cD_2$.  Indeed, the scattering operator or interacting fields acting on a state described by the wave function with disjoint momenta of the massive particles need not produce a state of this type.

\begin{proof}[Proof of Thm. \ref{thm:mod_definition}]
We consider only the case of the scalar model. The Dollard modifiers $S^\ras_{\rout/\rin}(\eta,g)$ are isometric and invertible, hence, unitary. We have
\begin{multline}\label{eq:thm_vertex}
 \int \rd^4 x \, g(x) \cL_{\rout/\rin}(\eta;x)
 =\int \rd^4 x \, g(x) J_{\rout/\rin}(\eta;x) \varphi(x)  
 \\
 =\int\mHm{p}\mHO{k} \left[\int\frac{\rd^4 q}{(2\pi)^4}\, \tilde{g}(q)  
 \frac{\pm\,\ri\, \F{\eta}(k-q)}{p\cdot (k-q)\pm\ri \zerop} a^*(k) + \mathrm{h.c.} \right] 
 b^*(p)b(p).
\end{multline}
We note that there exists a Schwartz functions on $\R^4\times\R^4$ whose restriction to $\Hm\times\R^4$ coincides with the function
\begin{equation}
 \Hm\times\R^4\ni(p,k)\mapsto \int\frac{\rd^4 q}{(2\pi)^4}\, \tilde{g}(q)  
 \frac{\ri\, \F{\eta}(k-q)}{p\cdot (k-q)\pm\ri \zerop} \in \C.
\end{equation}
This shows that expression~\eqref{eq:thm_vertex} defines an operator belonging to $L(\cD)$, where $\cD\in\{\cD_0,\cD_1,\cD_2\}$. The same is true for the operator
\begin{equation}
 \int \rd^4 x \rd^4 y\,g(x)g(y)
 D_0^D(x-y)\normord{J_{\rout/\rin}(\eta;x) J_{\rout/\rin}(\eta;y)}.
\end{equation}
Indeed, the relativistic Coulomb phases $\Phi_{\rout/\rin}(\eta,\eta,g,p_1,p_2)$, given by Eq.~\eqref{eq:def_rel_coulomb_phase}, are smooth functions of $p_1$ and $p_2$ of polynomial growth. This proves part (A).

Part (B) with $\cD=\cD_0$ follows immediately from (A) and axiom~\ref{axiom:wick} of the time-ordered products and Thm.~\ref{thm:eg0}. To obtain the case $\cD=\cD_1$ we need a version of this theorem in which $\cD_0$ is replaced with $\cD_1\subset\cD_0$. The proof of this generalization requires only minor modifications in the original proof of Thm.~\ref{thm:eg0} which is contained in Appendix~1 of~\cite{epstein1973role}. As an aside, let us mention that Thm.~\ref{thm:eg0} is false if $\cD=\cD_2$ because $\cD_2$ is not a subset of $\cD_0$ and all states belonging to $\cD_0$ can be created from the vacuum by the operator~\eqref{eq:eg_operator} smeared with some Schwartz function.

Part (C) is a consequence of the property of operator-valued Schwartz distributions stated below Eq.~\eqref{eq:op_valued_dist}, Thm.~\ref{thm:eg0} and the fact that $\cL_{\rout/\rin}(\eta;x)$ and $D_0^D(x-y)\normord{J_{\rout/\rin}(\eta;x) J_{\rout/\rin}(\eta;y)}$ are operator-valued Schwartz distributions on $\cD_0$.
\end{proof}

The physical scattering matrix $S_\rmod(\eta)$ and interacting fields $C_{\ret/\adv,\rmod}(\eta;h)$ are defined with the use of the adiabatic limit and depend on the choice of a profile $\eta$. The profile characterizes the cloud of photons surrounding electrons. As we have already mentioned, there is no distinguished choice of the profile. In order to find the relation between the physical scattering matrix and interacting fields constructed with the use of different profiles we introduce a family of unitary operators which we call the intertwining operators.
\begin{dfn}\label{def:intertwiners}
Let $\eta,\eta'$ be profiles and $\hat\varrho$ be a sector measure. The intertwining operators in the QED are defined by
\begin{multline}\label{eq:dfn_intertwiner}
 V_{\rout/\rin}(\eta',\eta) 
 := 
 V(\eta',\eta)
 \\
 \times 
 \exp\left(\mp\ri\frac{e^2}{2}\int_{\Hm\times\Hm}\normord{\rd\rho(p_1)\rd\rho(p_2)}\,\frac{p_1\cdot p_2}{\mass^2}\,\Phi_{\rout/\rin}(\eta'-\eta,\eta'+\eta,p_1,p_2)\right),
\end{multline}
where $\Phi_{\rout/\rin}(\eta'-\eta,\eta'+\eta,p_1,p_2)$ and $\rd\rho(p)=\rho(p)\,\mHm{p}$ are defined by~\eqref{eq:phi_finite} and~\eqref{eq:def_charge_dist_momentum},
\begin{multline}\label{eq:dfn_intertwiner_V}
 V(\eta',\eta) 
 := 
 \exp\left(e \int_{H_0\times\Hm} (\rd\rho-\rd\hat\varrho)(p)\,\mHO{k}
 \left( v^\mu(\eta',\eta,p,k)\, a_\mu^*(k) - \mathrm{h.c.}\right)
  \right)
 \\
 \times 
 \exp\left(-\frac{e^2}{2}\int_{\Hm\times\Hm}(\rd\rho-\rd\hat\varrho)(p_1)\,(\rd\rho-\rd\hat\varrho)(p_2)\,v(\eta',\eta,p_1,p_2)\right) 
\end{multline}
and
\begin{equation}
\begin{gathered}
 v^\mu(\eta',\eta,p,k) := (\F{\eta}'(k)-\F{\eta}(k)) \frac{p^\mu}{p\cdot k},
 \\
 v(\eta',\eta,p_1,p_2) :=\int\mHO{k}
 \left(\F{\eta}'(k)\F{\eta}(-k)-\F{\eta}'(-k)\F{\eta}(k)\right)
 g_{\mu\nu}
 \frac{p_1^\mu}{p_1\cdot k} 
 \frac{p_2^\nu}{p_2\cdot k}.
\end{gathered}
\end{equation} 
In the case of the scalar model the intertwining operators are defined by the same formulas with $a^\#_\mu(k)$ replaced by $a^\#(k)$, $g_{\mu\nu}$ with $-1$, $p_1\cdot p_2/\mass^2$ with $-1$, $p^\mu/p\cdot k$ with $\mass/p\cdot k$ and $\rho(p)=1/(2\mass)\,b^*(p)b(p)\,\mHm{p}$.
\end{dfn}
\begin{thm}\label{thm:intertwiners}
The expression~\eqref{eq:dfn_intertwiner} defines the intertwining operator as a unitary operator $V_{\rout/\rin}(\eta',\eta)\in B(\cH)$ and as a formal power series $V_{\rout/\rin}(\eta',\eta) \in L(\cD_2)\llbracket e\rrbracket$. For every profiles $\eta,\eta',\eta''$ it holds
\begin{equation}\label{eq:intertwiners_properties}
\begin{gathered}
 V_{\rout/\rin}(\eta,\eta)=\id,~~~V_{\rout/\rin}(\eta,\eta')V_{\rout/\rin}(\eta',\eta'')=V_{\rout/\rin}(\eta,\eta''),
 \\
 V_{\rout/\rin}(\eta,\eta')^*=V_{\rout/\rin}(\eta,\eta')^{-1}=V_{\rout/\rin}(\eta',\eta)
\end{gathered} 
\end{equation}
and for all $\Psi\in\cD_2$ it holds
\begin{equation}\label{eq:intertwiners_limits}
\begin{aligned}
 V_{\rout}(\eta',\eta)\Psi 
 &=
 \lim_{\epsilon\searrow0}\, R(\eta',g_\epsilon) S^\ras_\rout(\eta',g_\epsilon) S^\ras_\rout(\eta,g_\epsilon)^{-1}R(\eta,g_\epsilon)^{-1}\Psi,
 \\
 V_{\rin}(\eta',\eta)\Psi 
 &= 
 \lim_{\epsilon\searrow0}\, R(\eta',g_\epsilon)^{-1}S^\ras_\rin(\eta',g_\epsilon)^{-1} S^\ras_\rin(\eta,g_\epsilon)R(\eta,g_\epsilon)\Psi.
\end{aligned}
\end{equation}
Moreover, in the case of QED $V_{\rout/\rin}(\eta',\eta)$ commutes with the electric charge and BSRT operator. The above identities hold true if the intertwining operators $V_{\rout/\rin}(\eta',\eta)$ are understood as elements of $B(\cH)$ or $L(\cD_2)\llbracket e\rrbracket$.
\end{thm}
Note that generically $V_{\rout/\rin}(\eta,\eta')\Psi\notin \cD\llbracket e\rrbracket$ with $\cD\in\{\cD_0,\cD_1\}$ even for very regular states $\Psi$. 
In the proof of the above theorem we will use the following lemma, which follows easily from the Lebesgue dominated convergence theorem. 
\begin{lem}\label{lem:intertwiners}
Let 
\begin{equation}
 v_\epsilon\,:\,\Hm\times(\R\times S^2)\to\C,
 ~~~~
 u_\epsilon\,:\,\Hm\to\ri\R,
 ~~~~
 w_\epsilon\,:\,\Hm\times\Hm\to\ri\R,
\end{equation}
be families of functions parametrized by $\epsilon\in[0,1]$. Suppose that for all $\epsilon\in[0,1]$ and all $f\in\cS(\Hm)$, $h\in\cS(\Hm\times\Hm)$ such that $h$ is supported outside the diagonal the functions
\begin{equation}\label{eq:lem_intertwiners_assumptions1}
\begin{gathered}
 \Hm\times(\R\times S^2)\ni(p,k^0,\hat k)\mapsto f(p) v_\epsilon(p,k^0,\hat k)\in\C,
 \\
 \Hm\ni p \mapsto f(p) u_\epsilon(p) \in \R,
 \\
 \Hm\times\Hm\ni(p_1,p_2)\mapsto h(p_1,p_2) w_\epsilon(p_1,p_2) \in \R
\end{gathered}
\end{equation}
belong to the Schwartz class and their absolute values are bounded by some Schwartz functions independent of $\epsilon$. Moreover, assume that
\begin{equation}\label{eq:lem_intertwiners_assumptions2}
\begin{gathered}
 \lim_{\epsilon\searrow0}\,\int\mHO{k}\,|v_\epsilon(p,k)-v_0(p,k)|^2=0,
 \\
 \lim_{\epsilon\searrow0}\, u_\epsilon(p) = u_0(p),
 \\
 \lim_{\epsilon\searrow0}\, w_\epsilon(p_1,p_2) = w_0(p_1,p_2)
\end{gathered}
\end{equation}
pointwise almost everywhere. For $\epsilon\in[0,1]$ set
\begin{multline}\label{eq:W_epsilon_lemma}
 V_\epsilon := \exp\left(\frac{e}{2} \int\mHm{p}\mHO{k}\, \left(v_\epsilon(p,k)\, a^*(k) - \mathrm{h.c.}\right)b^*(p)b(p) \right)
 \\
 \times \exp\left(\frac{e^2}{8} \int\mHm{p}\,u_\epsilon(p)\, b^*(p)b(p) \right)
 \\
 \times \exp\left(\frac{e^2}{8} \int\mHm{p_1}\mHm{p_2}\,w_\epsilon(p_1,p_2)\, b^*(p_1)b^*(p_2)b(p_1)b(p_2) \right)
\end{multline}
For each $\epsilon\in[0,1]$ the above expression defines a unitary operators in $B(\cH)$ and a formal power series in $L(\cD_2)\llbracket e\rrbracket$. It holds
\begin{equation}
 \lim_{\epsilon\searrow0}\,V_\epsilon \Psi=V_0\Psi
\end{equation}
for arbitrary $\Psi\in\cH$ when the above limit is considered in the sense of $B(\cH)$ and for arbitrary $\Psi\in\cD_2$ when the above limit is considered in the sense of $L(\cD_2)\llbracket e\rrbracket$.
\end{lem}

\begin{proof}[Proof of Thm.~\ref{thm:intertwiners}]
The algebraic properties of the intertwining operators~\eqref{eq:intertwiners_properties} can be proved with the use of the BCH formula~\eqref{eq:BCH}. Let us show the existence of the limits~\eqref{eq:intertwiners_limits}. We carry out the proof only in the case of the scalar model with $\hat\varrho\equiv 0$. The generalization to other cases is straightforward. We have to study the following expression
\begin{multline}\label{eq:W_epsilon_thm}
 S^\ras_\rout(\eta',g_\epsilon) S^\ras_\rout(\eta,g_\epsilon)^{-1}\!
 =
 \exp\left(\ri e\int \rd^4 x \, g_\epsilon(x) J_\rout(\eta-\eta',x)\, \varphi(x)\right)
 \\
 \times\exp\left(-\ri\frac{e^2}{2}\int \rd^4 x \rd^4 y\,g_\epsilon(x)g_\epsilon(y)\, D_0(x-y)  
 J_\rout(\eta',x) J_\rout(\eta;y)
 \right)
 \\
 \times\exp\left(\ri\frac{e^2}{2}\int \rd^4 x \rd^4 y\,g_\epsilon(x)g_\epsilon(y)\, D_0^D(x-y)  
 \normord{J_\rout(\eta-\eta',x) J_\rout(\eta+\eta',y)}
 \right).
\end{multline}
By Part (A) of Thm.~\ref{thm:coulomb_phase} the last factor above converges to the expression in the third line of Eq.~\eqref{eq:dfn_intertwiner}. Let us denote the product of the first and the second factor by $V_\epsilon$.
For $\epsilon>0$ we set
\begin{equation}
 j_\epsilon(\eta,p,k):=\int\frac{\rd^4 q}{(2\pi)^4} 
 \,\tilde{g}_\epsilon(q) 
 \frac{\F{\eta}(k-q)}{p\cdot (k-q)+\ri \zerop}
\end{equation}
and
\begin{equation}
\begin{gathered}
 v_\epsilon(p,k):= j_\epsilon(\eta',p,k)-j_\epsilon(\eta,p,k),
 \quad
 u_\epsilon(p):=w_\epsilon(p,p),
 \\
 w_\epsilon(p_1,p_2):=\int\mHO{k}\,\left[ 
 j_\epsilon(\eta',p_1,k)\overline{j_\epsilon(\eta,p_2,k)}
 -
 \overline{j_\epsilon(\eta',p_1,k)}j_\epsilon(\eta,p_2,k)
 \right].
\end{gathered} 
\end{equation}
We also set
\begin{equation}
\begin{gathered}
 v_0(p,k):=\frac{\F{\eta}'(k)-\F{\eta}(k)}{p\cdot k},
 \quad
 u_0(p):=w_0(p,p),
 \\
 w_0(p_1,p_2):=\int\mHO{k}\,
 \frac{\F{\eta}'(k)\F{\eta}(-k)-\F{\eta}'(-k)\F{\eta}(k)}{(p_1\cdot k)(p_2\cdot k)}. 
\end{gathered}
\end{equation}
To prove the existence of the limit $\lim_{\epsilon\searrow0}V_\epsilon$ it is enough to show that the functions $v_\epsilon$, $u_\epsilon$, $w_\epsilon$ satisfy the assumptions of Lemma~\ref{lem:intertwiners}. By direct inspection, the functions~\eqref{eq:lem_intertwiners_assumptions1} have all the required properties. Let us turn to the proof of the existence of the first limit in~\eqref{eq:lem_intertwiners_assumptions2}. Note that
\begin{multline}
 \lim_{\epsilon\searrow0}\int\mHO{k}\,\left|\,\int_{p\cdot(k-q)<\epsilon}
 \frac{\rd^4 q}{(2\pi)^4}\, \tilde{g}_\epsilon(q) 
 \frac{\F{\eta}'(k-q)-\F{\eta}(k-q)}{p\cdot (k-q)+\ri \zerop}\right|^2 
 \\
 =
 \lim_{\epsilon\searrow0}\int\mHO{k}\,\left|\,\int_{p\cdot(k-q)<1}
 \frac{\rd^4 q}{(2\pi)^4}\, \tilde{g}(q) 
 \frac{\F{\eta}'(\epsilon k-\epsilon q)-\F{\eta}(\epsilon k-\epsilon q)}{p\cdot (k-q)+\ri \zerop}\right|^2 
 =0.
\end{multline}
To prove that the limit vanishes we used the Lebesgue dominated convergence theorem. The above identity implies that
\begin{multline}
 \lim_{\epsilon\searrow0} \int \mHO{k}\,|v_\epsilon(p,k)-v_0(p,k)|^2 
 \\
 =
 \lim_{\epsilon\searrow0} \int \mHO{k}
 \left|\int_{p\cdot(k-\epsilon q)\geq\epsilon}\frac{\rd^4 q}{(2\pi)^4}\tilde{g}(q)
 \frac{\F{\eta}'(k-\epsilon q)-\F{\eta}(k-\epsilon q)}{p\cdot (k-\epsilon q)}
 -\frac{\F{\eta}'(k)-\F{\eta}(k)}{p\cdot k}\right|^2 
 \\
 = 0,
\end{multline}
where the last equality follows again from the Lebesgue theorem. To show the existence of the second and the third limit in~\eqref{eq:lem_intertwiners_assumptions2} we note that
\begin{multline}\label{eq:thm_intertwiners_w}
 w_\epsilon(p_1,p_2)=
 \int\mHO{k}
 \frac{\rd^4 q_1}{(2\pi)^4}\frac{\rd^4 q_2}{(2\pi)^4}
 \,\tilde{g}_\epsilon(q_1) \tilde{g}_\epsilon(q_2) 
 \\\times
 \frac{\F{\eta}'(k-q_1)\F{\eta}(-k+q_2)-\F{\eta}(k-q_2)\F{\eta}'(-k+q_1)}{[p_1\cdot (k-q_1)+\ri 0][p_2\cdot (k-q_2)-\ri 0]}.
\end{multline}
One can show that the limit $\epsilon\searrow0$ of the integral on the RHS of the above expression vanishes when restricted to $(k,q_1,q_2)\in\HO\times\R^4\times\R^4$ such that $p_1\cdot(k-q_1)<\epsilon$ or $p_2\cdot(k-q_2)<\epsilon$. Thus, the limit $\epsilon\searrow0$ of the expression~\eqref{eq:thm_intertwiners_w} coincides with the limit $\epsilon\searrow0$ of the expression
\begin{multline}
 \int\limits_{\substack{p_1\cdot(k-\epsilon q_1)\geq\epsilon\\p_2\cdot(k-\epsilon q_2)\geq\epsilon}}\mHO{k}
 \frac{\rd^4 q_1}{(2\pi)^4}\frac{\rd^4 q_2}{(2\pi)^4}
 \,\tilde{g}(q_1) \tilde{g}(q_2)
 \\
 \times \frac{\F{\eta}'(k-\epsilon q_1)\F{\eta}(-k+\epsilon q_2)-\F{\eta}(k-\epsilon q_2)\F{\eta}'(-k+\epsilon q_1)}{p_1\cdot (k-\epsilon q_1)~p_2\cdot (k-\epsilon q_2)}
\end{multline}
which, by the Lebesgue dominated convergence theorem, is equal to $w_0(p_1,p_2)$. This shows the existence of the second and the third limit in~\eqref{eq:lem_intertwiners_assumptions2} and finishes the proof of the first equality in~\eqref{eq:intertwiners_limits}.
\end{proof}

Our modified scattering matrix and interacting fields in the scalar model and QED are defined as stated below.
\begin{cnj}\label{cnj:S_mod}
Fix a measure $\hat\varrho$ as in Def.~\ref{def:sector}. There exists a renormalization scheme for time-ordered products such that for any profile $\eta$, all polynomials in fields and their derivatives $C\in\Fa$ (in the case of QED we assume that $C$ is either the electromagnetic field $F_{\mu\nu}$ or the spinor current $J^\mu$), $h\in\cS(\R^4)$ and $\Psi,\Psi' \in \cD_2$ the adiabatic limits
\begin{equation}\label{eq:dfn_S_op_mod_scalar}
\begin{gathered}
 (\Psi|S_\rmod(\eta)\Psi'):=\lim_{\epsilon\searrow 0}\,(\Psi|S_\rmod(\eta,g_\epsilon)\Psi'),
 \\
 (\Psi|C_{\ret/\adv,\rmod}(\eta;h) \Psi') := \lim_{\epsilon\searrow 0}\, (\Psi|C_{\ret/\adv,\rmod}(\eta,g_\epsilon;h)\Psi')
\end{gathered} 
\end{equation}
exist in $\C\llbracket e\rrbracket$ and define the physical scattering matrix and the physical interacting fields
\begin{equation}
 S_\rmod(\eta)\in L(\cD_2,\cD_2^*)\llbracket e\rrbracket,
 \quad
 C_{\ret/\adv,\rmod}(\eta;h)\in L(\cD_2,\cD_2^*)\llbracket e\rrbracket.
\end{equation}
Moreover, if $\eta,\eta'$ are two profiles, then it holds
\begin{equation}\label{eq:S_op_mod_V}
\begin{gathered}
 S_\rmod(\eta') = V_\rout(\eta',\eta) S_\rmod(\eta) V_\rin(\eta,\eta'),
 \\
 C_{\ret/\adv,\rmod}(\eta';h)=V_{\rin/\rout}(\eta',\eta)C_{\ret/\adv,\rmod}(\eta;h)V_{\rin/\rout}(\eta,\eta').
\end{gathered} 
\end{equation}
In the case of QED the forms $S_\rmod(\eta)$ and $C_{\ret/\adv,\rmod}(\eta;h)$ commute with the BRST charge $Q_{\mathrm{BRST}}$. Consequently, they induce unique forms
\begin{equation}
 [S_\rmod(\eta)],[C_{\ret/\adv,\rmod}(\eta;h)]\in L(\cD^{\mathrm{phys}}_2,{\cD_2^{\mathrm{phys}}}^*)\llbracket e\rrbracket
\end{equation}
defined in the physical Hilbert space.
\end{cnj}
\noindent Let us make some comments about the above conjecture.
\begin{enumerate}[leftmargin=*,label={(\arabic*)}]
 \item By Thm.~\ref{thm:adiabatic_first_second} the statements in the above conjecture concerning the modified scattering matrix hold true at least in the first and the second order of perturbation theory. By Thm.~\ref{thm:first_order_fields} the part of the above conjecture concerning the modified interacting fields holds at least for the first-order corrections to $\psi$ and $\varphi$ in the scalar model and $F_{\mu\nu}$ and $J^\mu$ in QED. 
 
 \item The above conjecture implies a factorization of the IR divergences of the standard scattering matrix
 \begin{equation}\label{eq:factorization}
  S(g_\epsilon) =
  S^\ras_\rout(\eta,g_\epsilon)^{-1}
  S_\rmod(\eta,g_\epsilon) 
  S^\ras_\rin(\eta,g_\epsilon)^{-1}.
 \end{equation}
 In the above formula we omitted the operators $R(\eta,g)$ and $R(\eta,g)^{-1}$ which are not essential for the existence of the adiabatic limit (they are crucial for the BRST invariance of the construction). Our modified scattering matrix $S_\rmod$ is the IR finite part of the standard scattering matrix which is obtained by factoring out the first and the third factors above and taking the adiabatic limit. The factorization of the IR divergences in the expression for the standard scattering matrix in models such as QED is a common wisdom. Strong arguments for such factorization in QED were given in \cite{yennie1961infrared} and \cite{weinberg1965infrared},
 where this property was used to argue that the differential inclusive cross sections are finite in each order of perturbation theory (cf. Sec.~\ref{sec:inclusive_cross_section}). Note that the references \cite{yennie1961infrared} and \cite{weinberg1965infrared} use a different IR regulator: positive photon mass and lower cutoff of the photon energy, respectively. The above regulators break explicitly the symmetries of the theory (the BRST or Lorentz symmetry). More importantly, the S-matrix regularized with a positive photon mass or a low-energy cutoff is not IR-finite because of IR-divergent self-energy contributions that appear even in purely massive models. The S-matrix regularized with the use of the adiabatic cutoff is free from the above-mentioned problems.
 
 \item The modified S-matrix $S_\rmod(\eta)$ and the modified interacting fields $C_{\ret/\adv,\rmod}(\eta;h)$ are defined as formal power series whose coefficients are forms on $\cD_2$. It would be preferable to define the coefficients of these series as operators on some invariant dense domain in the Fock space. However, because of the singularity of the Coulomb phase at coinciding momenta the wave function of the state $\lim_{\epsilon\searrow0}S_\rmod^{[j]}(\eta,g_\epsilon)\Psi$ is singular if the momenta of some of the outgoing massive particles coincide. High-order corrections are not square-integrable and cannot be interpreted as vectors in the Fock space. These singularities are a feature of the perturbative expansion and would be absent in the non-perturbative formulation (the problematic contribution has a form of a phase factor with a phase of order $O(e^2)$ that diverges on a set of measure zero). As an aside, let us mention that it is plausible that one can define $S_\rmod(\eta)$ and $C_{\ret/\adv,\rmod}(\eta;h)$ as operators on some invariant domain in $\cH\llbracket e\rrbracket$ which is not of the form $\cD\llbracket e\rrbracket$ with a fixed $\cD\subset\cH$.
 
 \item\label{rem:eta} The modified scattering matrix $S_\rmod(\eta)$ and the modified interacting fields $C_\rmod(\eta)$ are defined in the standard Fock space. As follows from the results of Sec.~\ref{sec:fields}, the particle interpretation of states in this space is non-standard and depends on the choice of the profile $\eta$ and the super-selection sector. In Sec.~\ref{sec:state_space} using the compatibility relations~\eqref{eq:S_op_mod_V} we present a straight-forward reformulation the construction along the lines of Sec.~\ref{sec:Dollard}. We construct a family of asymptotic spaces $\hat\cH_{\rout/\rin}(\varrho)$ and define the scattering matrix $\hat S:\,\hat\cH_\rin(\varrho)\to\hat\cH_\rin(\varrho)$ and the interacting fields $\hat C_{\ret/\adv}:\,\hat\cH_{\rin/\rout}(\varrho)\to\hat\cH_{\rin/\rout}(\varrho)$ which depend only on the choice of the superselection sector determined by a measure $\varrho$ (the operators $\hat S$ and $\hat C_{\ret/\adv}$ do not depend on $\eta$). The advantage of this reformulation is the explicit $\eta$ dependence of the states in accord with their physical properties. The asymptotic spaces $\hat\cH_{\rout/\rin}(\varrho)$ have structure similar to the Fock space. However, the creators and annihilators of photons and electrons acting in these spaces do not commute.
 
 \item Our modified interacting fields satisfy the interacting field equations. For example, in the case of scalar model it holds
 \begin{align}
  (\square + \mass^2)\psi_{\ret/\adv,\rmod}(\eta;x) 
  &= e\,(\varphi\psi)_{\ret/\adv,\rmod}(\eta;x),
  \\
  \square \varphi_{\ret/\adv,\rmod}(\eta;x)
  &=\frac{e}{2}\,(\psi^2)_{\ret/\adv,\rmod}(\eta;x).
 \end{align}
 The above statement follows immediately from Eq.~\eqref{eq:mod_fields_relation} and the fact that the standard retarded and advanced fields given by the Bogoliubov formula~\eqref{eq:bogoliubov_fields} satisfy the above equations with the coupling constant $e$ replaced by $e\,g(x)$.

 \item The renormalization scheme required for the existence of the adiabatic limit of the modified scattering matrix is (essentially) unique. Hence, the modified scattering matrix depends only on the electric charge $e$, the mass of the electron $\mass$, the choice of a superselection sector~$\hat\varrho$ and a profile $\eta$ (however, see the comment~\ref{rem:eta} above).  
 
 In order to specify the renormalization scheme we first note that the time-ordered products with odd number of massive external fields vanish both in the scalar model and QED. Thus, by the bound on the Steinmann scaling degree~\eqref{eq:sd_bound} we conclude that in order to fix a renormalization scheme it is enough to consider the time-ordered products with the following external structure: (a) no external fields (vacuum bubble), $\omega=4$ (b) $m$ massless external fields, $k\in\{1,2,3,4\}$ (vacuum polarization), $\omega=4-m$, (c) two massive external fields (self-energy), $\omega=1$, (d) two massive and one massless external field (vertex correction), $\omega=0$. The order of singularity $\omega$ is defined by Eq.~\eqref{eq:omega}. Thus, in the case of QED we are led to consider the following distributions:
\begin{equation}\label{eq:normalization_t}
 t^{[n]}_k(x_1,\ldots,x_n)
 :=
 (\Omega|
 \T(J_{\mu_1}(x_{1}),\ldots,J_{\mu_k}(x_{k}),
 \cL(x_{k+1}),\ldots,\cL(x_n))\Omega), 
\end{equation} 
with $k\in \{0,\ldots,4\}$ and
\begin{multline}\label{eq:normalization_s}
 s^{[n]}_k(p;x_1,\ldots,x_n;p')
 \\
 :=
 (b^*(\sigma,p)\Omega|
 \T(J_{\mu_1}(x_{1}),\ldots,J_{\mu_k}(x_{k}),
 \cL(x_{k+1}),\ldots,\cL(x_n))|b^*(\sigma',p'))
\end{multline}
with $k\in \{0,1\}$. The the above formulas $b^*(\sigma,p)$ stands for the creator of an electron and $J_\mu$ is the spinor current. In the case of the scalar model we consider similar distributions with $J^\mu$ replaced by $J=\frac{1}{2}\psi^2$ and $b^*(\sigma,p)$ replaced by the creator of the scalar massive particle $b^*(p)$. By Eq.~\eqref{eq:n_P} the time-ordered products are translationally covariant. Hence,
\begin{equation}\label{eq:normalization_st}
\begin{aligned}
 \F{t}^{[n]}_k(q_1,\ldots,q_n) &=: (2\pi)^4\delta(q_1+\ldots+q_n)\F{\underline{t}}^{[n]}_{k}(q_1,q_2,\ldots,q_{n-1}),
 \\
 \F{s}^{[n]}_k(p;q_1,\ldots,q_n;p') &=: (2\pi)^4\delta(p+q_1+\ldots+q_n-p')\F{\underline{s}}^{[n]}_k(p;q_1,\ldots,q_n)
\end{aligned} 
\end{equation}
for some Schwartz distributions $\F{\underline{t}}^{[n]}_k\in\cS'(\R^{4(n-1)})$ and $\F{\underline{s}}^{[n]}_k,\in\cS'(\Hm\times\R^{4n})$. 

We impose the following renormalization conditions
\begin{align}
 \label{eq:renormalizaion_conditions_t}
 \F{\underline{t}}_{k}^{[n]}(q_1,\ldots,q_{n-1}) &= O^{\textrm{dist}}(|q_1,\ldots,q_{n}|^{4-k+\delta}),\quad k\in\{0,\ldots,4\},\,n\geq 1,
 \\
 \label{eq:renormalizaion_conditions_s} 
 \F{\underline{s}}^{[n]}_{k}(p;q_1,\ldots,q_n) &= O^{\textrm{dist}}(|q_1,\ldots,q_n|^{\delta}),\quad k\in\{0,1\},\,n\geq 2
\end{align}
for some $\delta\in(0,1)$, where $O^{\textrm{dist}}(\cdot)$ is a generalization of the $O(\cdot)$ notation for distribution that was introduced in~\cite{duch2018weak}. 

Condition~\eqref{eq:renormalizaion_conditions_t} is also known as the central normalization condition~\cite{epstein1973role,scharf2014}. The proof that it can be imposed in QED and the scalar model is contained in~\cite{duch2017massless}. In order to see that this condition is necessary for the existence of the adiabatic limit it is enough to consider the process of a creation of $k$ photons from the vacuum (cf. Eq.~\eqref{eq:infrared_problem_phi_k} and the remark below that equation). In particular, the condition~~\eqref{eq:renormalizaion_conditions_t} with $k=2$ implies the on-shell normalization of the photon -- the physical mass of the photon is equal to zero and the residue of the photon propagator is equal to $1$. The renormalization condition~\eqref{eq:renormalizaion_conditions_s} with $k=0$ asserts that the parameter $\mass$ is the physical mass of the electron. The fact that such a condition can be imposed follows from H{\"o}lder continuity of the distribution $\F{\underline{s}}$ (see \cite{blanchard1975green} for the proof in the case of QED). Mass normalization is necessary for the existence of the adiabatic limit even in purely massive models~\cite{epstein1976adiabatic,duch2018strong}. The renormalization conditions~\eqref{eq:renormalizaion_conditions_s} with $k=1$ is the interaction vertex normalization. It asserts that the coupling parameter $e$ is the physical charge of the electron. 

Note that in both the scalar model and QED the time-ordered products with two massive external fields are linearly divergent ($\omega=1$). In the case of the scalar model these products are fixed uniquely by the mass normalization (i.e. condition \eqref{eq:renormalizaion_conditions_s} with $k=0$) and the Lorentz invariance of the time-ordered products~\eqref{eq:n_P}. Hence, in the case of the scalar model conditions~\eqref{eq:renormalizaion_conditions_t} and ~\eqref{eq:renormalizaion_conditions_s} fix uniquely all the time-ordered products that are used in the construction the Bogoliubov S-matrix $S(g)$. In the case of QED there is still possibility to renormalize the wave function of the electron.  A convenient choice of the wave function renormalization in QED that fixes completely the renormalization freedom is the central normalization condition
\begin{equation}
\begin{gathered}
 w^{[n]}(x_1,\ldots,x_n):=(\Omega|
 \T((\slashed{A}\psi)_a(x_1),
 \cL(x_3),\ldots,\cL(x_{n-1}),(\overline{\psi}\slashed{A})_b(x_n))\Omega),
 \\
 \F{w}^{[n]}(q_1,\ldots,q_n) =: (2\pi)^4\delta(q_1+\ldots+q_n)\,\F{\underline{w}}^{[n]}(q_1,\ldots,q_{n-1}),
 \\
 \partial_q^{\alpha}\F{\underline{w}}^{[n]}(q_1,\ldots,q_{n-1})=O^{\textrm{dist}}(|q_1,\ldots,q_n|^{\delta})
\end{gathered} 
\end{equation}
for all multi-indices $\alpha$, $|\alpha|=1$ and some $\delta\in(0,1)$. Note that it is impossible to normalize the self-energy of the electron on-shell by imposing the condition
\begin{equation}
\begin{aligned}
 \F{\underline{s}}^{[n]}_{k=0}(p;q_1,\ldots,q_n) &= O^{\textrm{dist}}(|q_1,\ldots,q_n|^{1+\delta})
\end{aligned} 
\end{equation}
with some $\delta\in(0,1)$ because of the familiar logarithmic IR singularities (cf. Eq.~\eqref{eq:self_energy_qed}). Let us mentioned that the choice of the wave function renormalization does not affect the adiabatic limit. This is the case because  we define the electric charge $e$ in terms of the distribution~\eqref{eq:normalization_s} with $m=1$ that includes the sum of all self-energy corrections. As an aside, observe also that is difficult to express the renormalization conditions using the standard vertex correction
\begin{equation}
 \T((\slashed{A}\psi)_a(x_2),J^\mu(x_1),
 \cL(x_3),\ldots,\cL(x_{n-1}),(\overline{\psi}\slashed{A})_b(x_n))\Omega) 
\end{equation}
because it is logarithmically divergent on the electron mass shell.

Finally, let us point that in the case of QED the condition~\eqref{eq:renormalizaion_conditions_t} with $k$ odd is satisfied trivially because the corresponding distributions $t^{[n]}_k$ vanish identically by the Furry theorem (see e.g.~\cite{scharf2014}). Moreover, the condition~\eqref{eq:renormalizaion_conditions_t} with $k=4$ and the condition~\eqref{eq:renormalizaion_conditions_s} with $k=1$ follow from the Ward identities~\eqref{eq:n_W}. The former condition corresponds to vanishing of the four-photon polarization tensor at zero momenta and the latter to the vanishing of the radiative corrections to the electric form factor of the electron at zero momentum transfer.  

\item In the case of QED we use the photon propagator in the so-called Feynman gauge which follows from the modified action of QED given by Eq.~\eqref{eq:action_qed_mod}. One could consider a more general class of modified actions with corresponding photon propagators in the $R_\xi$ gauge. Because of the Ward identities the modified scattering matrix and modified interacting gauge-invariant fields obtained with the use of any of these propagators coincide.
 
\item\label{rem:brst} Finally, we comment on the BRST invariance of the construction in the case of QED. The modified scattering matrix with adiabatic cutoff does not commute with the free BRST charge (the same is true for the standard Bogoliubov S-matrix). More specifically, we have
\begin{equation}\label{eq:formal_gauge_inv_S_QED}
 [Q_{\mathrm{BRST}},S_\rmod(\eta,g)] = (-\ri) \partial_\lambda S_\rmod(\eta,g,\lambda \partial_\mu g)\big|_{\lambda=0},
\end{equation}
The extended modified scattering operator $S_\rmod(\eta,g,g_\mu)$ with adiabatic cutoffs $g,g_\mu\in\cS(\R^4)$ is defined by a formula similar to~\eqref{eq:mod_S_op} with the following modifications: $g\,\cL \mapsto g\,\cL - g_\mu\, J^\mu c$ in the definition~\eqref{eq:bogoliubov_S_op} of the Bogoliubov S-matrix, $g\,\cL_{\rout/\rin} \mapsto g\,\cL_{\rout/\rin} - g_\mu\, J_{\rout/\rin}^\mu c$ in the definition~\eqref{eq:def_modifiers} of the Dollard modifiers and $g\,\underline{\cL} \mapsto g\,\underline{\cL} -  g_\mu\, \underline{J}^\mu c$ in the definition~\eqref{eq:def_operator_R} of the sector operator, where $c$ is the ghost field. In order to prove Eq.~\eqref{eq:formal_gauge_inv_S_QED} we use 
\begin{equation}
[Q_{\textrm{BRST}},A_\mu(x)]=\partial_\mu c(x),\quad
[Q_{\textrm{BRST}},J^\mu(x)]=0,\quad
[Q_{\textrm{BRST}},J^\mu_{\rout/\rin}(x)]=0,
\end{equation}
the Ward identities~\eqref{eq:n_W}, the BCH formula~\eqref{eq:BCH} and the fact that the currents $J^\mu_{\rout/\rin}\pm\underline{J}^\mu$ are conserved. 

Now, let $g_\epsilon(x):=g(\epsilon x)$ and $g_{\mu,\epsilon}(x):=g_\mu(\epsilon x)$, $g,g_\mu\in\cS(\R^4)$ and $g(0)=1$, $g_{\mu}(0)=1$. Note that if $g_\mu(x)=\partial_\mu g(x)$, then $\partial_\mu g_\epsilon(x) = \epsilon g_{\mu,\epsilon}(x)$. Consequently,
\begin{equation}\label{eq:brst_comment}
 \frac{1}{\epsilon} \partial_\lambda  S_\rmod(\eta,g_\epsilon,\lambda \partial_\mu g_\epsilon)\big|_{\lambda=0}
 =
 \partial_\lambda  S_\rmod(\eta,g_\epsilon,\lambda g_{\mu,\epsilon})\big|_{\lambda=0}.
\end{equation}
The infrared properties of the operators $S_\rmod(\eta,g,g_\mu)$ and $S_\rmod(\eta,g)$ are similar. In particular, all extra interaction vertices appearing in that operator contain two massive and one massless field. This indicates that the adiabatic limit of the RHS of Eq.~\eqref{eq:brst_comment} exists and the adiabatic limit of $S_\rmod(\eta,g_\epsilon)$ commutes with $Q_{\mathrm{BRST}}$. Actually, to reach the above conclusion it is enough to show that the RHS of Eq.~\eqref{eq:brst_comment} is of order $O(\epsilon^{-1+\delta})$ with $\delta>0$. It is unlikely that the last condition is not satisfied since the IR divergences that appear, for example, in the construction of the standard scattering matrix in QED are only logarithmic. One can use a similar argument to show the BRST-invariance of $F^{\mu\nu}_{\ret/\adv,\rmod}(\eta;h)$ and $J^\mu_{\ret/\adv,\rmod}(\eta;h)$. To generalize this reasoning to other gauge invariant fields one would have to impose the Ward identities~\eqref{eq:n_W} for polynomials $B$ which are not necessarily sub-polynomial of the interaction vertex.  
\end{enumerate}

\section{First- and second-order corrections to S-matrix}\label{sec:low_orders}

In this section we prove the existence of the adiabatic limit that defines the first and the second order corrections to the modified scattering matrix as operators from $\cD_1$ to $\cH$. This shows that one can solve (at least in low orders of perturbation theory) both the IR and UV problem in QED and similar models using the strategy presented in Sec.~\ref{sec:modified}.

\begin{thm}\label{thm:adiabatic_first_second}
Let $\Psi\in\cD_1$ (cf. Sec.~\ref{sec:domains}), $\eta$ be a profile (cf. Def.~\ref{dfn:profile}) and $\hat\varrho$ be the sector measure (cf. Def.~\ref{def:sector}).
\\
(A) It holds
\begin{equation}
 \lim_{\epsilon\searrow 0}\, \|S_\rmod^{[1]}(\eta,g_\epsilon)\Psi\|=0.
\end{equation}
(B) Assume that the time-ordered products are normalized such that
\begin{equation}\label{eq:normalization_conditions}
 \F{\underline{t}}^{[2]}_{0}(q) = O(|q|^5),
 ~~~
 \F{\underline{t}}^{[2]}_{2}(q) = O(|q|^3),
 ~~~
 \F{\underline{s}}^{[2]}_{0}(q) = 0 \textrm{~~for~~} q^2=\mass^2,
\end{equation}
where the above distributions are defined by Eqs~\eqref{eq:normalization_t}, \eqref{eq:normalization_s} and \eqref{eq:normalization_st}. There exists $S_\rmod^{[2]}(\eta)\in L(\cD_1,\cH)$ such that
\begin{equation}
 \lim_{\epsilon\searrow0}\,\|  (S_\rmod^{[2]}(\eta,g_\epsilon) -  S_\rmod^{[2]}(\eta))\Psi \|=0.
\end{equation}
Moreover, in the case of QED the following limits exist
\begin{equation}\label{eq:brst_existence}
 \lim_{\epsilon\searrow0}\,\partial_\lambda  S^{[1]}_\rmod(\eta,g_\epsilon,\lambda g_{\mu,\epsilon})\big|_{\lambda=0}\Psi,
 \quad
 \lim_{\epsilon\searrow0}\,\partial_\lambda  S^{[2]}_\rmod(\eta,g_\epsilon,\lambda g_{\mu,\epsilon})\big|_{\lambda=0}\Psi.
\end{equation}
\end{thm}
Note that by the comment~\ref{rem:brst} below Conjecture~\ref{cnj:S_mod} the existence of the limits~\eqref{eq:brst_existence} implies that $S_\rmod^{[2]}(\eta)$ commutes with the BRST charge and induces an operator in $L(\cD^{\textrm{phys}}_1,\cH^{\textrm{phys}})$. The proofs of Part (A) and (B) of the above theorem are contained in Sections~\ref{sec:S_op_mod_first_order_corrections} and~\ref{sec:S_op_mod_second_order_corrections}, respectively. We give the detailed proofs in the case of the scalar model and discuss how to adapt them to the case of QED. We prove the theorem by considering separately all possible processes which contribute to the scattering matrix at a given order of perturbation theory. The processes in the first and second-order which do not vanish for kinematical reasons are depicted in Figures~\ref{fig:first_order_non_trivial}-\ref{fig:moller}. In the Feynman diagrams presented in this paper, we use straight lines to depict electron and positron propagators and wavy lines to depict photon propagators. The lines representing the incoming particles are drawn on the right and the lines representing the outgoing particles -- on the left. To simplify the notation we always assume that the incoming states contain only particles that actually participate in the interaction process. Thus, in particular when considering the first-order corrections to the scattering matrix we consider incoming states with at most three particles. The generalization for arbitrary incoming states is trivial. For simplicity, we assume that the wave functions of the incoming states are compactly supported in momentum space. We also assume that the profile $\eta$ is of compact support in momentum space. In the proof we frequently use the lemma stated below which says that the adiabatic limit of $S^{[n]}_\rmod(\eta,g_\epsilon)$ depends only on the form of the distribution $S^{[n]}_\rmod(\eta,q_1,\ldots,q_n)$, defined by Eq.~\eqref{eq:S_mod_sym_dist}, in an arbitrarily small neighborhood of the origin. The lemma implies in particular that the adiabatic limit vanishes if there exist a neighborhood $\mathcal{O}$ of zero in $\R^4$ such that $S^{[n]}_\rmod(\eta,g)=0$ for all $g\in\cS(\R^4)$, $\supp\,\F{g}\subset\mathcal{O}$. 
\begin{lem}\label{thm:adiabatic_neighborhood}
Let $\chi\in \cS(\R^4)$ be such that $\chi\equiv 1$ in some neighborhood of the origin and let $\Psi$ be a state which belongs to any of the domains considered in this paper: $\cD_0$ or $\cD_1\supset\cD_2$. Set
\begin{equation}
 \Psi_\epsilon:=\int\mP{q_1}\ldots\mP{q_n}\,\F{g}_\epsilon(-q_1)\ldots\F{g}_\epsilon(-q_n)\,
 \\
 S^{[n]}_\rmod(\eta,q_1,\ldots,q_n)\Psi
\end{equation}
and
\begin{equation}
 \Psi^\chi_\epsilon:=\int\mP{q_1}\ldots\mP{q_n}\,\chi(q_1,\ldots,q_n)\,\F{g}_\epsilon(-q_1)\ldots\F{g}_\epsilon(-q_n)\,
 \\
 S^{[n]}_\rmod(\eta,q_1,\ldots,q_n)\Psi.
\end{equation}
It holds
\begin{equation}\label{eq:thm_neighborhood}
 \|\Psi_\epsilon - \Psi^\chi_\epsilon\| = O(\epsilon^M)
\end{equation}
for any $M\in\N_+$.
\end{lem}
\begin{proof}
Note that the expression
\begin{equation}
 (\Psi|S^{[n]}_\rmod(\eta,q'_1,\ldots,q'_n)^*S^{[n]}_\rmod(\eta,q_1,\ldots,q_n)|\Psi)
\end{equation}
is an element of $\cS'(\R^{8n})$. Eq.~\eqref{eq:thm_neighborhood} follows from the fact that the family of functions
\begin{equation}
 \epsilon^{-M}\,\F{g}_\epsilon(-q_1)\ldots\F{g}_\epsilon(-q_n)~(1-\chi(q_1,\ldots,q_n))
\end{equation}
converges to zero in the topology of $\cS(\R^{4n})$ in the limit $\epsilon\searrow0$.
\end{proof}

\subsection{First-order corrections}\label{sec:S_op_mod_first_order_corrections}

It is easy to see that in both the scalar model and QED the first-order correction to the modified scattering matrix is given by the following expression 
\begin{equation}\label{eq:first_order_general}
 S^{[1]}_\rmod(\eta,g)  =
 \ri\int\rd^4 x\, g(x)\,(\cL(x)-\cL_\ras(\eta,x)),
\end{equation}
where $\cL$ is the interaction vertex of the model and $\cL_\ras=\cL_\rin+\cL_\rout$ with $\cL_\rin$ and $\cL_\rout$ defined in Def.~\ref{def:D_modifiers}. Note that $\underline{\cL}$ (see Def.~\ref{def:sector}) does not contribute at this order. In this section we prove that the adiabatic limit of the above correction vanishes. We study all contributing processes which are depicted in Figures~\ref{fig:first_order_non_trivial} and~\ref{fig:trivial}. The processes from Figure~\ref{fig:trivial} vanish in the adiabatic limit by the conservation of the total energy and momentum. The processes from Figure~\ref{fig:first_order_non_trivial} require more care. We recall that the non-existence of the adiabatic limit of the first-order correction to the standard scattering matrix was shown in Sec.~\ref{sec:IR_problem_qft}.

\subsubsection{Non-trivial processes}

\begin{figure}[h!]
\begin{center}
\begin{tabular}{cc}
\includegraphics{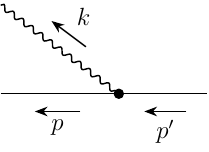}\qquad
&
\qquad
\includegraphics{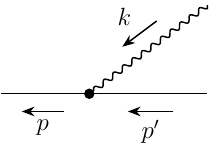}
\\
(A) & (B)
\end{tabular}
\end{center}
\caption{Non-trivial diagrams in first-order}\label{fig:first_order_non_trivial}
\end{figure}

The wave function of the outgoing particles in Figure~\ref{fig:first_order_non_trivial} (A) has the following form
\begin{multline}\label{eq:first_order_A}
 F^{(A)}_\epsilon(p,k):= \int\frac{\rd^4 q}{(2\pi)^3}\, \F{g}(q)\, \bigg(
 \delta(2p\cdot (k-\epsilon q) + (k-\epsilon q)^2)\, f(p+k-\epsilon q)
 \\
 -
 \delta(2p\cdot (k-\epsilon q))\,\F{\eta}(k-\epsilon q)\,f(p) \bigg).
\end{multline}
The momenta which are arguments of the wave functions  are always on-shell. It holds
\begin{multline}
 \epsilon F^{(A)}_\epsilon(p,\epsilon k)= \int\frac{\rd^4 q}{(2\pi)^3}\, \F{g}(q)\, \bigg(
 \delta(2p\cdot (k-q) + \epsilon (k-q)^2)\, f(p+\epsilon k-\epsilon q)
 \\
 -
 \delta(2p\cdot (k-q))\,\eta(\epsilon k-\epsilon q)f(p) \bigg).
\end{multline}
Using the above identity we get
\begin{equation}
 \lim_{\epsilon\searrow0}\int\mHm{p}\mHO{k}\,|F^{(A)}_\epsilon(p,k)|^2 
 =\lim_{\epsilon\searrow0}\int\mHm{p}\mHO{k}\,|\epsilon F^{(A)}_\epsilon(p,\epsilon k)|^2 
 =0.
\end{equation}
The wave function of the outgoing electron in Figure~\ref{fig:first_order_non_trivial} (B) is given by
\begin{multline}\label{eq:first_order_B}
 F^{(B)}_\epsilon(p):= 
 \epsilon \int\frac{\rd^4 q}{(2\pi)^3}\mHO{k'}\, \F{g}(q)\, \bigg(
 \delta(2p\cdot (k'+q) - \epsilon (k'+q)^2)\, f(p-\epsilon k'-\epsilon q,k')
 \\
 -
 \delta(2p\cdot (k'+q))\,\F{\eta}(-\epsilon k'-\epsilon q)\,f(p,k') \bigg).
\end{multline}
Consequently, it holds
\begin{equation}
 \int\mHm{p}\,|F^{(B)}_\epsilon(p)|^2 
 =o(\epsilon^2).
\end{equation}
This proves that the contributions to $S^{[1]}(g)\Psi$ coming from the processes depicted in Figure~\ref{fig:first_order_non_trivial} vanish in the adiabatic limit in the topology of the Fock-Hilbert space. 

In the case of QED the wave functions of the first-order corrections have very similar forms with the integrands in Eqs.~\eqref{eq:first_order_A} and~\eqref{eq:first_order_B} multiplied by some smooth functions of the incoming/outgoing momenta (depending on the spins of the incoming and outgoing electron and the Lorentz index of the photon). The same is true for the first of the limits in~\eqref{eq:brst_existence}. Thus, the above conclusions hold in this case as well.

\subsubsection{Trivial processes}

\begin{figure}[h!]
\begin{center}
\begin{tabular}{cccc}
\includegraphics{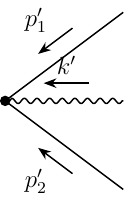}\qquad
&
\qquad
\includegraphics{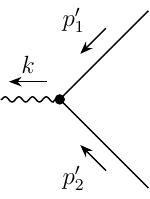}\qquad
&
\qquad
\includegraphics{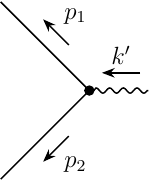}\qquad
&
\qquad
\includegraphics{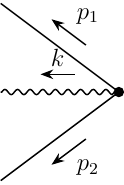}
\\
(A) & (B) & (C) & (D)
\end{tabular}
\end{center}
\caption{Trivial diagrams in first-order}\label{fig:trivial}
\end{figure}

Since the asymptotic vertices do not describe processes depicted in Figure~\ref{fig:trivial} it follows from Eq.~\eqref{eq:first_order_general} that the contributions to $S_\rmod^{[1]}(\eta,g)$ coming from these processes coincide with the corresponding contributions to the standard Bogoliubov S-matrix $S^{[1]}(g)$. The amplitudes have the following form
\begin{gather}\label{eq:first_order_trivial}
 F^{(A)}_\epsilon:=\int\mHm{p'_1}\mHm{p'_2}\mHO{k'}\,\F{g}_\epsilon(-p'_1-p'_2-k')\,f(p'_1,p'_2,k'),
 \\
 F^{(B)}_\epsilon(k):=\int\mHm{p'_1}\mHm{p'_2}\,\F{g}_\epsilon(k-p'_1-p'_2)\,f(p'_1,p'_2),
 \\
 F^{(C)}_\epsilon(p_1,p_2):=\int\mHO{k'}\,\F{g}_\epsilon(p_1+p_2-k')\, f(k'),
 \\
 F^{(D)}_\epsilon(p_1,p_2,k):=\F{g}_\epsilon(p_1+p_2+k).
\end{gather}
We remind the reader that the momenta of the incoming and outgoing particles are always on-shell and by the assumption the support of any of the functions denoted above by $f$ is compact. Using the above facts it is easy to see that in each of the cases the argument of $\F{g}_\epsilon$ is contained in the complement of some open ball centered at the origin in $\R^4$ whose radius does not depend on the momenta of outgoing particles. Since $\F{g}_\epsilon(q)=\frac{1}{\epsilon^4} \F{g}\left(\frac{q}{\epsilon}\right)$, $g\in\cS(\R^4)$, we have
\begin{equation}
 \sup_{\substack{q\in\R^4\\|q|>\lambda}}\, |\F{g}_\epsilon(q)|\leq \const(M,\lambda)\,\epsilon^M 
\end{equation}
for any $M\in\N_+$. Consequently, the $L^2$ norm of all of above wave functions behave like $O(\epsilon^M)$ for any $M\in\N_+$ and vanishes in the adiabatic limit. 

\subsubsection{Bounds}

The technical results presented in this section will be useful in the proof of Part (B) of Thm.~\ref{thm:adiabatic_first_second}, which is presented in Sec.~\ref{sec:S_op_mod_second_order_corrections}. We use the family of seminorms in the Fock space introduced in the definition below.
\begin{dfn}
Consider the scalar model. Let $M\in\N_0$, $\delta\geq 0$ and $\Psi\in \cD_1\supset \cD_2$ be such that
\begin{multline}
 \Psi = \sum_{i,j\in\N_0}\int\mHm{p_1}\ldots\mHm{p_i}\mHO{k_1}\ldots\mHO{k_j}\,
 \\
 h_{ij}(\vec{p}_1,\ldots,\vec{p}_i,\vec{k}_1,\ldots,\vec{k}_j)\,
 b^*(p_1)\ldots b^*(p_i) a^*(k_1)\ldots a^*(k_j)\Omega. 
\end{multline}
By definition
\begin{equation}
 \|\Psi\|_{\delta,M}:=
 \sum_{i,j\in\N_0}\sum_{|\alpha|<M}\sup_{\substack{p_1,\ldots,p_i\in\Hm\\k_1,\ldots,k_i\in\HO}} 
 |k^0_1|^\delta\ldots |k^0_j|^\delta \left|\partial^\alpha_{\vec{p}}
  h_{ij}(\vec{p}_1,\ldots,\vec{p}_i,\vec{k}_1,\ldots,\vec{k}_j) \right|,
\end{equation}
where $\alpha$ is a multi-index. In an analogous way we define the family of seminorms in the Fock space of QED. The corresponding expression for $\|\Psi\|_{\delta,M}$ contains an additional sum over all polarizations of photons and electrons or positrons.
\end{dfn}

\begin{lem}\label{lem:first_order_bounds}
For any $\delta\in(0,1)$ and $\epsilon>0$ it holds
\begin{equation}\label{eq:first_bounds_std}
 \|\cL(g_\epsilon)\Psi\|_{\delta,M} 
 \leq \const \,\epsilon^{\delta-1}\,\|\Psi\|_{\delta,M},
\end{equation}
\begin{equation}\label{eq:first_bounds_as}
 \|\cL_{\ras/\rout/\rin}(\eta,g_\epsilon)\Psi\|_{\delta,M} 
 \leq \const \,\epsilon^{\delta-1}\,\|\Psi\|_{\delta,M}
\end{equation}
and
\begin{equation}\label{eq:first_bounds_diff}
 \|(\cL(g_\epsilon) - \cL_\ras(\eta,g_\epsilon))\Psi\|_{\delta,M} \leq \const \,\epsilon^\delta\,\|\Psi\|_{\delta,M+1}.
\end{equation}
The above bounds are valid both in the scalar model and QED. In addition, in the case of QED it holds 
\begin{equation}\label{eq:first_bound_sector}
 \|\underline{\cL}(\eta,g_\epsilon)\Psi\|_{\delta,M} 
 \leq \const \,\epsilon^{\delta-1}\,\|\Psi\|_{\delta,M}.
\end{equation}
\end{lem}
\noindent The straightforward but tedious proof of the above lemma is omitted.

\subsection{Second-order corrections}\label{sec:S_op_mod_second_order_corrections}

The second order correction to the modified scattering matrix in the scalar model and QED is given by
\begin{multline}
 S_\rmod^{[2]}(\eta,g) = 
 \frac{\ri^2}{2}\int\rd^4 x_1\rd^4 x_2 \, g(x_1)g(x_2) \bigg(\T(\cL(x_1),\cL(x_2))
  -2\cL(x_1) \cL_\rin(x_2) 
  \\[0.5em]
  -2\cL_\rout(x_1) \cL(x_2) 
 +\aT(\cL_\rout(x_1),\cL_\rout(x_2)) 
  +2\cL_\rout(x_1)\cL_\rin(x_2)
  \\[0.5em]
  +\aT(\cL_\rin(x_1),\cL_\rin(x_2)) 
  -2[\underline{\cL}(x_1),\cL(x_2)-\cL_\rin(x_2)-\cL_\rout(x_2)]\bigg),
\end{multline}
where $\T$ and $\aT$ denote the time-ordered and anti-time-ordered products, $\cL$ is the interaction vertex of the model and other vertices were introduced in Def.~\ref{def:D_modifiers} and~\ref{def:sector}. We set $\cL_\ras=\cL_\rin+\cL_\rout$. In this section we suppress the dependence of $\cL_\ras$, $\cL_\rin$, $\cL_\rout$ on the profile $\eta$. We claim that
\begin{multline}\label{eq:second_order_general}
 \lim_{\epsilon\searrow0}\, S_\rmod^{[2]}(\eta,g_\epsilon)\Psi = 
 \lim_{\epsilon\searrow0}\, \frac{\ri^2}{2}\int\rd^4 x_1\rd^4 x_2 \, g_\epsilon(x_1)g_\epsilon(x_2)
 \\\times 
 \big(\T(\cL(x_1),\cL(x_2)) 
 -\cL(x_1)\cL(x_2)
 +\aT(\cL_\ras(x_1),\cL_\ras(x_2)) - \cL_\ras(x_1)\cL_\ras(x_2)
 \\
 -2\aT(\cL_\rin(x_1),\cL_\rout(x_2)) +2\cL_\rin(x_1)\cL_\rout(x_2)
 \big)\Psi
\end{multline}
for $\Psi\in\cD_1$ in the topology of the Fock-Hilbert space. It is easy to see that the adiabatic limit of the last line exists. In fact, for compactly supported profiles $\eta$ the support of the expression in the last line of the above equation is compact. Moreover, this expression is nonzero only in the case of the M{\o}ller/Bhabha scattering. To prove Eq.~\eqref{eq:second_order_general} it is enough to show that the contribution coming from
\begin{multline}
 \cL(x_1)\cL(x_2) -2\cL(x_1) \cL_\rin(x_2) 
 -2\cL_\rout(x_1) \cL(x_2)  
 \\
 +\cL_\ras(x_1)\cL_\ras(x_2) + 2 [\cL_\rout(x_1),\cL_\rin(x_2)]
 \\
 -2[\underline{\cL}(x_1),\cL(x_2)-\cL_\ras(x_2)]
\end{multline}
vanishes in the adiabatic limit of. By Lemma~\ref{lem:first_order_bounds} the adiabatic limit of the second and third line coincides with the adiabatic limit of
\begin{multline}
  \cL_\ras(x_1)\cL_\ras(x_2) -2\cL_\ras(x_1) \cL_\rin(x_2) 
  -2\cL_\rout(x_1) \cL_\ras(x_2)  
  \\
  + \cL_\ras(x_1)\cL_\ras(x_2) + 2 [\cL_\rout(x_1),\cL_\rin(x_2)],
\end{multline}
which is equal to zero. Using again Lemma~\ref{lem:first_order_bounds} we prove that
\begin{equation}
 [\underline{\cL}(x_1),\cL(x_2)-\cL_\ras(x_2)]
\end{equation}
is also zero in the adiabatic limit. Note that if $\Psi\in\cD_1$ does not contain massive particles, then only the first term in the second line of~\eqref{eq:second_order_general} contributes. By the results of Sec.~\ref{sec:S_op_mod_first_order_corrections}, Thm.~\ref{thm:adiabatic_first_second} holds true with $S_\rmod^{[2]}(\eta)=0$ for contributions corresponding to disconnected Feynman diagrams. The second-order  connected Feynman diagrams have zero, two or four external lines. The numbers of the external lines corresponding to electrons, positrons and photons are even. We divide all possible connected processes into the following classes:
\begin{enumerate}[leftmargin=*,label={(\arabic*)}]
 \item vacuum bubble -- no external particles,
 \item vacuum polarization -- two incoming or outgoing photons,
 \item self-energy of electron or positron -- two incoming or outgoing electrons or positrons,
 \item Compton scattering -- two incoming or outgoing photons  and one incoming and one outgoing electron or positron,
 \item pair creation -- two incoming or outgoing photons and two incoming or two outgoing electrons or positrons,
 \item M{\o}ller or Bhabha scattering -- four incoming or outgoing electrons.
\end{enumerate}
In the following sections we study all contributions to the modified scattering operator in the second-order corresponding to connected Feynman diagrams. All the Feynman diagrams of particular type corresponding to contributions which do not vanish trivially by the approximate energy-momentum conservation are depicted in each of the subsections below. We shall frequently use the formula~\eqref{eq:second_order_general}. Note that in the case of the processes (1), (2) and (5) only the first term on the RHS contributes. In the case of the processes (3) and (4) only the first two terms contribute. In the case of the process (6) all terms have to be taken into account for the limit to exist. Moreover, we observe that only processes (2), (4) and (5) contribute to the second limit in~\eqref{eq:brst_existence} with one external photon replaced by a ghost.

\subsubsection{Vacuum bubble}

Let
\begin{equation}\label{eq:distribution_vacuum_bubble}
 t(x):= (\Omega|\T(\cL(x),\cL(0))\Omega).
\end{equation}
The distribution $\F{t}(q)$ is an analytic function in some neighborhood of zero both in the scalar model and QED. It is defined uniquely up to a polynomial in $q$ of degree $4$. Thus, it is possible to normalize this distribution such that $\F{t}(q)=O(|q|^{5})$ at $q=0$. Thus, 
\begin{multline}
 (\Omega|S^{[2]}_\rmod(\eta,g_\epsilon)\Omega)
 =
 -\frac{1}{2}\,(\Omega|\T(\cL(g_\epsilon),\cL(g_\epsilon))\Omega)
 \\
 = 
 -\frac{1}{2}\int \mP{q_1}\mP{q_2}\,\F{g}(q_1)\F{g}(q_2)\,(2\pi)^4\delta(q_1+q_2)\,
 \frac{1}{\epsilon^4}\F{t}(-\epsilon q_2)
\end{multline}
vanishes in the limit $\epsilon\searrow0$. 

\begin{figure}[h!]
\begin{center}
\includegraphics{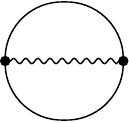}
\end{center}
\caption{Vacuum bubble diagram}
\end{figure}

Concluding, Thm.~\ref{thm:adiabatic_first_second} (B) restricted to the vacuum bubble contribution holds iff $\F{t}(q)=O(|q|^{5})$ at $q=0$, where $t$ is given by~\eqref{eq:distribution_vacuum_bubble}. We have $S^{[2]}_\rmod(\eta) = 0$ for this contribution.

\subsubsection{Vacuum polarization}\label{sec:vacuum_polarization}

Let
\begin{equation}\label{eq:distribution_vacuum_polarization}
 t(x):= \frac{1}{4}\,(\Omega|\T(\psi^2(x),\psi^2(0))\Omega).
\end{equation}
The distribution $\F{t}(q)$ is an analytic function in some neighborhood of zero. It is defined uniquely up to a Lorentz invariant polynomial in $q$ of degree $2$. Set $\F{t}(q) = \ri\Pi(q^2)$. It holds
\begin{equation}
 \Pi(q^2) = \Pi(0) + \Pi'(0)~q^2 + O(|q^2|^2).
\end{equation} 
Let
\begin{equation}
 F^{(A)}_\epsilon(k_1,k_2) 
 :=
 (k_1,k_2|S^{[2]}_\rmod(\eta,g_\epsilon)\Omega)
 =
 -\frac{1}{2}\,(k_1,k_2|\T(\cL(g_\epsilon),\cL(g_\epsilon))\Omega).
\end{equation}
It holds
\begin{multline}
 \epsilon^2 F^{(A)}_\epsilon(\epsilon k_1,\epsilon k_2)
 \\
 =
 -\ri \int\mP{q_1}\mP{q_2}\,\F{g}(q_1)\F{g}(q_2)
 \,(2\pi)^4\delta(k_1+k_2-q_1-q_2)\,\frac{1}{\epsilon^2}\Pi(\epsilon^2(k_1-q_1)^2).
\end{multline}
Consider the limit
\begin{gather}
 \lim_{\epsilon\searrow0}\int\!\mHO{k_1}\mHO{k_2}|F^{(A)}_\epsilon(k_1,k_2)|
 =
 \lim_{\epsilon\searrow0}\int\!\mHO{k_1}\mHO{k_2}
 |\epsilon^2 F^{(A)}_\epsilon(\epsilon k_1,\epsilon k_2)|.
\end{gather}
If $\Pi(0)\neq 0$, then the above limits is infinite. If $\Pi(0)=0$ but $\Pi'(0)\neq 0$, then the limit exists but it depends on $g$.  If $\Pi(0)=0$ and $\Pi'(0)=0$, then the limits exist and vanish. Set
\begin{equation}
 F^{(B)}_\epsilon:=\int \mHO{k_1}\mHO{k_2}\,(\Omega|S^{[2]}_\rmod(\eta,g_\epsilon)|k_1,k_2)\,f(k_1,k_2).
\end{equation}
It is easy to see that under the above conditions the limit
\begin{equation}
  \lim_{\epsilon\searrow0}\,|F^{(B)}_\epsilon|
\end{equation}
exists and vanishes as well. Now let
\begin{multline}
 F^{(C)}_\epsilon(k) 
 :=
 \int\mHO{k'}\,(k|S^{[2]}_\rmod(\eta,g_\epsilon)|k')\, f(k') 
 \\
 =
 -\frac{1}{2}\int\mHO{k'}\,(k|\T(\cL(g_\epsilon),\cL(g_\epsilon))|k')\, f(k').
\end{multline}
We obtain
\begin{multline}
 F^{(C)}_\epsilon(k) 
 \\
 =-\ri \int\mH{0}{k'}\mP{q_1}\mP{q_2}\F{g}_\epsilon(q_1)\F{g}_\epsilon(q_2)
 (2\pi)^4\delta(k-q_1-q_2-k')\Pi((k-q_1)^2)f(k')
 \\
 =-\ri\int\mP{q_1}\mP{q_2}\F{g}(q_1)\F{g}(q_2)\,
 \theta(k^0-\epsilon q_1^0-\epsilon q_2^0)
 \delta((k-\epsilon q_1-\epsilon q_2)^2)\,
 \\
 \times\Pi((k-\epsilon q_1)^2)\,f(k-\epsilon q_1-\epsilon q_2).
\end{multline}
Assuming that $\Pi(0)=0$, $\Pi'(0)=0$ and performing the integral with respect to $q_1^0$ using the Dirac delta we arrive at
\begin{equation}
 |F^{(C)}_\epsilon(k)|\leq \epsilon\, \frac{\const}{(1+|k|)^M}, 
\end{equation}
which implies that the limit
\begin{equation}\label{eq:vacuum_pol_limit}
 \lim_{\epsilon\searrow0}\,\int\mHO{k}\,|F^{(C)}_\epsilon(k)|
\end{equation}
vanishes.

\begin{figure}[h!]
\begin{center}
\begin{tabular}{ccc}
\includegraphics{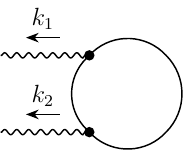}
&
\qquad
\includegraphics{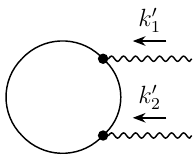}\qquad
&
\includegraphics{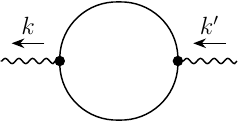}
\\
(A) & (B) & (C)
\end{tabular}
\caption{Vacuum polarization diagrams}
\end{center}
\end{figure}

In the case of QED one has to study the behavior of the distribution
\begin{equation}\label{eq:distribution_vacuum_polarization_qed}
 t_{\mu\nu}(x):= (\Omega|\T(J_\mu(x),J_\nu(0))\Omega).
\end{equation}
Using the Ward identity one shows~\cite{scharf2014} that $\F{t}_{\mu\nu}(q) = \ri (q^2 g_{\mu\nu} - q_\mu q_\nu ) \Pi(q^2)$, where $\Pi(q^2)$ is an analytic function of $q$ in some neighborhood of zero which is defined uniquely up to a constant. We renormalize it in such a way that $\Pi(0)=0$.  With this normalization of $t_{\mu\nu}(x)$ the results obtained in the scalar model generalize immediately to the case of QED.  Concluding, Thm.~\ref{thm:adiabatic_first_second} restricted to the vacuum polarization contributions holds iff $\F{t}(q)=O(|q|^{3})$ at $q=0$, where $t$ is given by Eqs.~\eqref{eq:distribution_vacuum_polarization} and~\eqref{eq:distribution_vacuum_polarization_qed}. It holds $S^{[2]}_\rmod(\eta) = 0$ for these contributions.

\subsubsection{Self-energy of electron/positron}\label{sec:2ed_se}

First, note that by Eq.~\eqref{eq:second_order_general} the adiabatic limit of
\begin{equation}
 \int\mHm{p'}\, (p|S_\rmod^{[2]}(\eta,g_\epsilon)|p')\,f(p')
\end{equation}
coincides with the adiabatic limit of
\begin{equation}
 F_\epsilon(p) :=\int\mHm{p'}\, \big[(p|\T(\cL(g_\epsilon),\cL(g_\epsilon))|p')-(p|\cL(g_\epsilon) \cL(g_\epsilon)|p')\big]\,f(p').
\end{equation}
Let
\begin{equation}
\begin{gathered}\label{eq:distribution_self_energy}
 t(x) := (\Omega|\T(\psi\varphi(x),\psi\varphi(0))\Omega),
 \\ 
 t_0(x) := \frac{1}{2}\,(\Omega|\psi\varphi(x)\,\psi\varphi(0)\Omega) +
 \frac{1}{2}\,(\Omega|\psi\varphi(0)\,\psi\varphi(x)\Omega).
\end{gathered}
\end{equation}
The distribution $\F{t}(q)$ is defined uniquely up to a Lorentz invariant polynomial in $q$ of degree $2$. We set
\begin{equation}
 \Xi(q^2-\mass^2) := -\ri(\F{t}(q) - \F{t}_0(q)).
\end{equation}
By direct computation we get
\begin{equation}
 \Xi(q^2-\mass^2) =  \frac{1}{16 \pi^2} \left(1-\frac{\mass^2}{q^2}\right) \log\left|1-\frac{q^2}{\mass^2}\right|+ c_1 + c_2 \left(q^2-\mass^2\right),
\end{equation}
where $c_1,c_2\in\R$ are some constants depending on the renormalization scheme. One verifies that
\begin{equation}\label{eq:Xi_properties}
 \Xi(-2\epsilon p\cdot q_1 +\epsilon^2 q_1^2)+\Xi(2\epsilon p\cdot q_1 - \epsilon^2 (q_1^2-2q_1\cdot q_2))
 = 2\Xi(0) + O(\epsilon^{2-\delta})
\end{equation}
for any $\delta>0$. We have
\begin{multline}\label{eq:self_energy_correction_scalar}
 F_\epsilon(p) = -\ri \int\mHm{p'}\mP{q_1}\mP{q_2}\,\F{g}_\epsilon(q_1)\F{g}_\epsilon(q_2)\,
 \\
 \times (2\pi)^4\delta(p-q_1-q_2-p')\,\Xi((p-q_1)^2-\mass^2)\,f(p')
 \\
 = -\frac{\ri}{2} \int\mP{q_1}\mP{q_2}\,\F{g}(q_1)\F{g}(q_2)\,\theta(p^0-\epsilon q_1^0-\epsilon q_2^0)\delta((p-\epsilon q_1-\epsilon q_2)^2-\mass^2)\,
 \\
 \times\left[\Xi(-2\epsilon p\cdot q_1 +\epsilon^2 q_1^2)+\Xi(2\epsilon p\cdot q_1 - \epsilon^2 (q_1^2-2q_1\cdot q_2))\right]\,f(p-\epsilon q_1-\epsilon q_2).
\end{multline}
If $\Xi(0)=c_1\neq 0$, then the limit
\begin{equation}\label{eq:self_energy_limit}
 \lim_{\epsilon\searrow0}\,\int\mHm{p}\,|F_\epsilon(p)|
\end{equation}
is infinite. If $c_1=0$, then we perform the integral with respect to $q_1^0$ using the Dirac delta and get
\begin{equation}\label{eq:self_energy_scalar_limit}
 |F_\epsilon(p)|\leq \epsilon^{1-\delta}\, \frac{\const}{(1+|p|)^M}, 
\end{equation}
which implies that the limit~\eqref{eq:self_energy_limit} vanishes. Note that this result is independent of the value of $c_2$. The self-energy contributions with two incoming or two outgoing electrons/positrons vanish trivially because of the energy-momentum conservation. 

\begin{figure}[h]
\begin{center}
\includegraphics{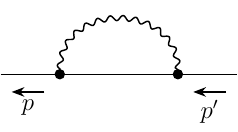}
\caption{Self-energy diagram}\label{fig:self_energy}
\end{center}
\end{figure}

The generalization of the above results to the case of QED is not obvious. Let us consider the correction to the propagator of electron in QED. We have
\begin{multline}
 F_\epsilon(\sigma,p) =\sum_{\sigma'=1,2}\int\mHm{p'}\, \big[(\sigma,p|\T(\cL(g_\epsilon),\cL(g_\epsilon))|\sigma',p')
 \\
 -(\sigma,p|\cL(g_\epsilon) \cL(g_\epsilon)|\sigma',p')\big]\,f(\sigma',p').
\end{multline}
Let
\begin{equation}\label{eq:distribution_self_energy_qed}
\begin{gathered}
 t_{ab}(x) = (\Omega|\T((\slashed{A}\psi)_a(x),(\overline{\psi}\slashed{A})_b(0))\Omega),
 \\
 t_{0,ab}(x) = \frac{1}{2}\,(\Omega|(\slashed{A}\psi)_a(x)\,(\overline{\psi}\slashed{A})_b(0)\Omega) -
 \frac{1}{2}\,(\Omega|(\overline{\psi}\slashed{A})_b(0)(\slashed{A}\psi)_a(x)\Omega)
\end{gathered} 
\end{equation}
and
\begin{equation}
 \Xi_{ab}(q) = \F{t}_{ab}(q) - \F{t}_{0,ab}(q).
\end{equation}
By direct computation (cf. e.g. \cite{scharf2014}) we get
\begin{multline}\label{eq:self_energy_qed}
 \Xi_{ab}(q) = \frac{1}{4\pi^2}\bigg[\left(1-\frac{\mass^2}{q^2}\right) 
 \left(\mass \id_{ab}+ \frac{1}{4}\left(1+\frac{\mass^2}{q^2}\right)\slashed{q}_{ab}\right) 
 \log\left|1-\frac{q^2}{\mass^2}\right|
 \\
 + \frac{\mass^2}{4q^2}\slashed{q}_{ab} - \frac{\mass}{4} \id_{ab}
 + c_1 \id_{ab}+ c_2 \left(\slashed{q}_{ab}-\mass\id_{ab}\right)\bigg],
\end{multline} 
where the constants $c_1,c_2\in\R$ depend on the renormalization scheme. We define
\begin{equation}
 \Xi(q^2-\mass^2)\delta_{\sigma,\sigma'} := \overline{u(\sigma,q)} \Xi(q)u(\sigma',q).
\end{equation}
Assume that $c_1=0$. For $p\in\Hm$ we have
\begin{equation}
 \overline{u(\sigma,p)}\Xi(p-\epsilon q_1) u(\sigma',p-\epsilon q_1-\epsilon q_2) - \Xi((p-\epsilon q_1)^2-\mass^2)\delta_{\sigma,\sigma'} = O(\epsilon^{2-\delta}),
\end{equation}
where $\delta>0$ is arbitrary. Moreover, $\Xi(q^2-\mass^2)$ satisfies the condition~\eqref{eq:Xi_properties}. Consequently, \begin{equation}
 F_\epsilon(\sigma,p) = \int\mHm{p'}\,(2\pi)^4\delta(p-q_1-q_2-p') \, \overline{u(\sigma,p)}\Xi(p-q_1) u(\sigma',p')\,f(p')
\end{equation} 
can be rewritten in the form analogous to~\eqref{eq:self_energy_correction_scalar}. Thus,~\eqref{eq:self_energy_scalar_limit} holds in QED if~$c_1=0$.

Concluding, Thm.~\ref{thm:adiabatic_first_second} restricted to the self-energy contributions holds iff $\F{t}(q)=0$ for $q^2=\mass^2$, where $t$ is given by Eqs.~\eqref{eq:distribution_self_energy} and \eqref{eq:distribution_self_energy_qed}. It holds $S^{[2]}_\rmod(\eta) = 0$ for these contributions.

\subsubsection{Compton scattering}

Let
\begin{equation}\label{eq:F_Compton}
 F^{(AB)}_\epsilon(p,k) := \int\mHm{p'}\mHO{k'}\, (p,k|S_\rmod^{[2]}(\eta,g_\epsilon)|p',k')\,f(p',k').
\end{equation}
Using~\eqref{eq:second_order_general} and neglecting terms vanishing in the adiabatic limit we obtain
\begin{multline}
 F^{(AB)}_\epsilon(p,k)
 \\
 =-\ri \int\mHm{p'}\mH{0}{k'}\frac{\rd^4 q_1}{(2\pi)^4}\frac{\rd^4 q_2}{(2\pi)^4}\, \F{g}_\epsilon(q_1)\F{g}_\epsilon(q_2)
 \,(2\pi)^4\delta(p+k-q_1-q_2-p'-k') 
 \\
 \times \left[\PV \frac{1}{(p'+k'+q_2)^2-\mass^2}+\PV\frac{1}{(p-k'-q_2)^2-\mass^2}\right]  f(p',k').
\end{multline}
If $p,p'\in\Hm$, $k,k'\in\HO$, $p+k-q_1-q_2-p'-k'=0$, $|k|+|k'|\geq \epsilon \mass$, $|q_1|,|q_2|\leq \const$ and the constant in the last bound is sufficiently small, then for arbitrary $\epsilon>0$ it holds
\begin{equation}
 (p'+k'+\epsilon q_2)^2-\mass^2,~ (p-k'-\epsilon q_2)^2-\mass^2 ~ > \const > 0,
\end{equation}
where the constant is independent of $\epsilon$. We set
\begin{equation}
 F^{(AB)}_\epsilon(p,k) = F^{\ssl}_\epsilon(p,k) + F^{\ssg}_\epsilon(p,k),
\end{equation}
where $F^{\ssg}_\epsilon(p,k)$ is given by~\eqref{eq:compton_less} and
\begin{multline}
 F^{\ssg}_\epsilon(p,k)
 \\
 :=-\ri \int\mHm{p'}\mH{0}{k'}\frac{\rd^4 q_1}{(2\pi)^4}\frac{\rd^4 q_2}{(2\pi)^4} \F{g}(q_1) \F{g}(q_2)
 \,(2\pi)^4\delta(p+k-\epsilon q_1-\epsilon q_2-p'-k') 
 \\
 \times \theta(|k|+|k'|-\epsilon \mass) 
 \left(\frac{1}{(p+k-\epsilon q_1)^2-\mass^2}
 +\frac{1}{(p-k'-\epsilon q_2)^2-\mass^2}\right) f(p',k').
\end{multline}
We note that $F^{\ssg}_\epsilon$ converges in $L^2(\Hm\times\HO,\rd\mu_\mass\times\rd\mu_0)$ to the following function
\begin{equation}
 F_0^{(AB)}(p,k)
 =\ri\int\mH{0}{k'}
 \frac{k\cdot k'}{(2p\cdot k)(2p\cdot k')}\,\delta((p+k-k')^2-\mass^2)\,f(p+k-k',k'),
\end{equation}
which belongs to $\cD_1$. Moreover, 
\begin{equation}\label{eq:compton_limit_rest}
 \int\mHm{p}\mHO{k}\,|F^{\ssl}_\epsilon(p,k)|^2 
 =
 \int\mHm{p}\mHO{k}\,|\epsilon F^{\ssl}_\epsilon(p,\epsilon k)|^2
 =
 o(\epsilon^2),
\end{equation}
which implies that $F^{(AB)}_\epsilon$ converges in $L^2(\Hm\times\HO,\rd\mu_\mass\times\rd\mu_0)$ to $F^{(AB)}_0$. In order to prove~\eqref{eq:compton_limit_rest} we first note that  
\begin{multline}\label{eq:compton_less}
 F^{\ssl}_\epsilon(p,\epsilon k)
 \\
 =-\ri
 \int\mH{0}{k'}\frac{\rd^4 q_1}{(2\pi)^4}\frac{\rd^4 q_2}{(2\pi)^3}\,
 \bigg( \theta(\mass-|k|+|k'|)  \,f(p+\epsilon k-\epsilon k'-\epsilon q_1-\epsilon q_2,\epsilon k')
 \\
 \times \F{g}(q_1) \left[
 \PV\frac{\epsilon}{(p+\epsilon k-\epsilon q_1)^2-\mass^2}
 +\PV\frac{\epsilon}{(p-\epsilon k'-\epsilon q_2)^2-\mass^2}
 \right]
 \\
 \times \F{g}(q_2)\, 
 \theta(p^0-\epsilon k^0-\epsilon k'^0-\epsilon q_1^0-\epsilon q_2^0)
 \,\epsilon\delta((p+\epsilon k-\epsilon k'-\epsilon q_1-\epsilon q_2)^2-\mass^2)\bigg).
\end{multline}
Using the above formula we prove that $|F^{\ssl}_\epsilon(p,\epsilon k)|$ is bounded by constant independent of $\epsilon>0$ and
\begin{equation}
 \lim_{\epsilon\searrow0} \,F^{\ssl}_\epsilon(p,\epsilon k)=0
\end{equation}
pointwise. Eq.~\eqref{eq:compton_limit_rest} follows now from the fact that the supports of functions $F^{\ssl}_\epsilon(p,\epsilon k)$ are contained in some compact subset of $\Hm\times\HO$ independent of $\epsilon>0$.

\begin{figure}[h]
\begin{center}
\begin{tabular}{cc}
 \includegraphics{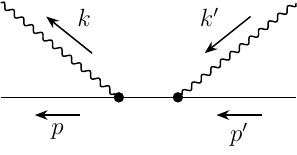}
 &
 \includegraphics{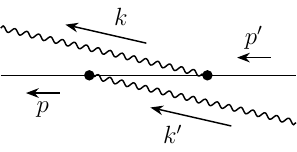}
 \\
 (A) & (B)
 \\
 \includegraphics{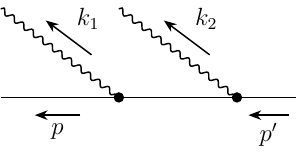}
 &
 \includegraphics{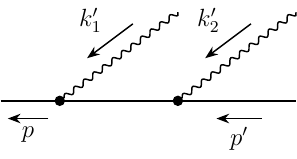}
 \\
 (C) & (D)
\end{tabular}
\caption{Compton scattering diagrams}\label{fig:compton}
\end{center}
\end{figure}

Now let
\begin{equation}\label{eq:compton_amp_2}
 F^{(C)}_\epsilon(p,k_1,k_2) := \int\mHm{p'}\, (p,k_1,k_2|S_\rmod^{[2]}(\eta,g_\epsilon)|p')\,f(p').
\end{equation}
Using~\eqref{eq:second_order_general} and neglecting terms vanishing in the adiabatic limit we get
\begin{multline}
 F^{(C)}_\epsilon(p,k_1,k_2)
 \\
 =-\frac{\ri}{2} \int\mHm{p'}\frac{\rd^4 q_1}{(2\pi)^4}\frac{\rd^4 q_2}{(2\pi)^4} \F{g}_\epsilon(q_1) \F{g}_\epsilon(q_2)
 \,(2\pi)^4\delta(p+k_1+k_2-q_1-q_2-p') 
 \\
 \times\left[\PV\frac{1}{(p+k_1-q_1)^2-\mass^2}+\PV\frac{1}{(p+k_2-q_2)^2-\mass^2}\right] f(p').
\end{multline}
It holds
\begin{multline}
 \epsilon^2 F^{(C)}_\epsilon(p,\epsilon k_1,\epsilon k_2)
 =\frac{\ri}{2}\int \frac{\rd^4 q_1}{(2\pi)^3}\frac{\rd^4 q_2}{(2\pi)^4} 
 \,\theta(p^0+\epsilon k_1^0+\epsilon k_2^0-\epsilon q_1^0-\epsilon q_2^0)
 \\
 \times\F{g}(q_1) \F{g}(q_2) \left[\PV\frac{\epsilon}{(p+\epsilon k_1-\epsilon q_1)^2-\mass^2}+\PV\frac{\epsilon}{(p+\epsilon k_2-\epsilon q_2)^2-\mass^2}\right]
 \\
 \times
 \epsilon\delta((p+\epsilon k_1+\epsilon k_2-\epsilon q_1-\epsilon q_2)^2-\mass^2)
 \,f(p+\epsilon k_1+\epsilon k_2-\epsilon q_1-\epsilon q_2).
\end{multline}
With the use of the above expression we prove that $|\epsilon^2 F^{(C)}_\epsilon(p,\epsilon k_1,\epsilon k_2)|$ is bounded by constant independent of $\epsilon>0$ and
\begin{equation}
 \lim_{\epsilon\searrow0} \,\epsilon^2 F^{(C)}_\epsilon(p,\epsilon k_1,\epsilon k_2)=0
\end{equation}
pointwise. Moreover, the supports of functions $F^{(C)}_\epsilon(p,\epsilon k)$ are contained in some compact subset of $\Hm\times\HO^{\times2}$. Thus, it holds
\begin{multline}
 \lim_{\epsilon\searrow0} \,\int\mHm{p}\mHO{k_1}\mHO{k_2}\,|F^{(C)}_\epsilon(p,k_1,k_2)|^2
 \\
 =
 \lim_{\epsilon\searrow0} \,\int\mHm{p}\mHO{k_1}\mHO{k_2}\,|\epsilon^2 F^{(C)}_\epsilon(p,\epsilon k_1,\epsilon k_2)|^2
 = 0.
\end{multline}
Finally, let 
\begin{equation}
 F^{(D)}_\epsilon(p) := \int\mHm{p'}\mHO{k'_1}\mHO{k'_2}\, (p|S_\rmod^{[2]}(\eta,g_\epsilon)|p',k'_1,k'_2)\,f(p',k'_1,k'_2).
\end{equation}
We have
\begin{multline}
 F^{(D)}_\epsilon(p)
 \\
 =-\frac{\ri}{2} \int\mHm{p'}\mHO{k'_1}\mHO{k'_2}\frac{\rd^4 q_1}{(2\pi)^4}\frac{\rd^4 q_2}{(2\pi)^4} (2\pi)^4\delta(p-q_1-q_2-p'-k'_1-k'_2) 
 \\
 \times\F{g}_\epsilon(q_1) \F{g}_\epsilon(q_2) \left[\PV\frac{1}{(p-k'_1-q_1)^2-\mass^2}+\PV\frac{1}{(p-k'_2-q_2)^2-\mass^2}\right] f(p',k'_1,k'_2)
\end{multline}
up to terms vanishing in the adiabatic limit. We easily prove that $|\epsilon^{-2} F^{(D)}_\epsilon(p)|$ is bounded by constant independent of $\epsilon>0$ and
\begin{equation}
 \lim_{\epsilon\searrow0} \,\epsilon^{-2} F^{(D)}_\epsilon(p)=0
\end{equation}
pointwise. Moreover, the supports of functions $F^{(D)}_\epsilon(p)$ are contained in some compact subset of $\Hm$. It follows that
\begin{equation}
 \int\mHm{p}\,|F^{(D)}_\epsilon(p)|^2
 = o(\epsilon^4).
\end{equation}

The integrands that appear in the expressions for the wave-functions in QED depend on the spins of the incoming/outgoing electron and the Lorentz indices of the photons and differ from the corresponding integrands in the scalar model only by some factors which are bounded (on compact sets) and smooth functions of $q_1,q_2$ and the momenta of the incoming and outgoing particles. Hence, all of the above results generalize easily to the case of QED. Concluding, Thm.~\ref{thm:adiabatic_first_second} restricted to the contributions considered in this section holds true with some nonzero $S^{[2]}_\rmod(\eta)\in L(\cD_1,\cH)$. 

\subsubsection{Pair creation/annihilation}

In this section we consider processes with two incoming or two outgoing electrons. Let
\begin{multline}
 F^{(A)}_\epsilon(p_1,p_2) := \int\mHO{k'_1}\mHO{k'_2}\, (p_1,p_2|S_\rmod^{[2]}(\eta,g_\epsilon)|k'_1,k'_2)\,f(k'_1,k'_2)
 \\
 =\int\mHO{k'_1}\mHO{k'_2}\, (p_1,p_2|S^{[2]}(g_\epsilon)|k'_1,k'_2)\,f(k'_1,k'_2).
\end{multline}
We have
\begin{multline}
 F^{(A)}_\epsilon(p_1,p_2)
 \\
 =
 -\frac{\ri}{2} \int\mHO{k'_1}\mHO{k'_2}\frac{\rd^4 q_1}{(2\pi)^4}\frac{\rd^4 q_2}{(2\pi)^4}\,
 (2\pi)^4\delta(p_1+p_2-q_1-q_2-k'_1-k'_2) 
 \\
 \times\F{g}_\epsilon(q_1) \F{g}_\epsilon(q_2)\,\PV\left[\frac{1}{(p_1-q_1-k'_1)^2-\mass^2}+\frac{1}{(p_1-q_1-k'_2)^2-\mass^2}\right] f(k'_1,k'_2).
\end{multline}
The principle value symbol in the above expression can be omitted as for $\F{g}$ with the support in sufficiently small neighborhood of the origin the denominators of the propagators never vanish. The above function converges in $L^2(\Hm^{\times 2},\rd\mu_\mass^{\times 2})$ to the following continuous function
\begin{multline}
 F_0^{(A)}(p_1,p_2)
 =
 -\frac{\ri}{2} \int\mHO{k'_1}\mHO{k'_2}\,
 \F{g}_\epsilon(q_1) \F{g}_\epsilon(q_2)
 \,(2\pi)^4\delta(p_1+p_2-k'_1-k'_2) 
 \\
 \times\left[\frac{1}{(p_1-k'_1)^2-\mass^2}+\frac{1}{(p_1-k'_2)^2-\mass^2}\right] f(k'_1,k'_2).
\end{multline}
Almost the same reasoning can be used to show convergence of
\begin{equation}
 F^{(B)}_\epsilon(k_1,k_2) := \int\mHO{k'_1}\mHO{k'_2}\, (k_1,k_2|S_\rmod^{[2]}(\eta,g_\epsilon)|p'_1,p'_2)\,f(p'_1,p'_2).
\end{equation}
Diagrams with one incoming and one outgoing photon vanish trivially in the adiabatic limit by the approximate conservation of the total energy and momentum. The same is true for diagrams with no incoming or outgoing particles.
\begin{figure}[h]
\begin{center}
\begin{tabular}{cc}
\includegraphics{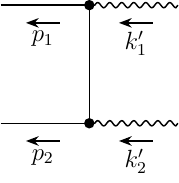}
&
\includegraphics{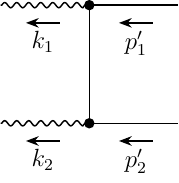}
\\
(A) & (B)
\end{tabular}
\caption{Pair creation/annihilation diagrams}
\end{center}
\end{figure}

In the above reasoning we used only general properties of the amplitudes following from the approximate energy momentum conservation in each of the two vertices. Thus, the above result generalizes immediately to the case of QED. Concluding, Thm.~\ref{thm:adiabatic_first_second} restricted to the contributions considered in this section holds true with some nonzero $S^{[2]}_\rmod(\eta)\in L(\cD_1,\cH)$.

\subsubsection{M{\o}ller/Bhabha scattering}\label{sec:moller}

Now, let us consider the scattering of two electrons/positrons. In QED, the electron-electron scattering is called the M{\o}ller scattering whereas the electron-positron scattering is called the Bhabha scattering. In the case of the scalar model there is no distinction between the two processes.

\begin{figure}[h]
\begin{center}
\begin{tabular}{ccc}
\includegraphics{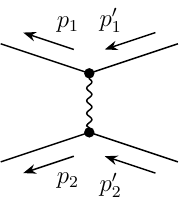}
&
\qquad
\includegraphics{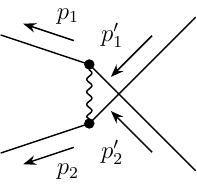}
\qquad
&
\includegraphics{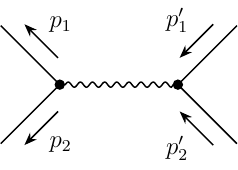}
\\
(A) $t$-channel
&
(B) $u$-channel
&
(C) $s$-channel
\end{tabular}
\caption{M{\o}ller/Bhabha scattering diagrams}\label{fig:moller}
\end{center}
\end{figure}
Let
\begin{equation}
 F_\epsilon(p_1,p_2) := \int\mHm{p'_1}\mHm{p'_2}\, (p_1,p_2|S_\rmod^{[2]}(\eta,g_\epsilon)|p'_1,p'_2)\,f(p'_1,p'_2).
\end{equation}
It is easy to see that the adiabatic limit of the s-channel contribution exist since the denominator of the propagator $(p_1+p_2-q_2)^2$ never vanishes. In what follows, we consider in detail only the t- and u-channel contribution $F^{(AB)}_\epsilon$. By Eq.~\eqref{eq:second_order_general} we have
\begin{equation}
 F^{(AB)}_\epsilon(p_1,p_2) =  F^{(I)}_\epsilon(p_1,p_2)+F^{(II)}_\epsilon(p_1,p_2),
\end{equation}
where
\begin{multline}\label{eq:moller_std}
 F^{(I)}_\epsilon(p_1,p_2) 
 \\
 =-\ri\int \mHm{p'_1}\mHm{p'_2}\frac{\rd^4 k}{(2\pi)^4}\, 
 \F{g}_\epsilon(p_2-k-p'_2)\F{g}_\epsilon(p_1+k-p'_1)\,\PV\frac{1}{k^2}\, f(p'_1,p'_2)
 \\
 =-\ri\int \frac{\rd^4 q_1}{(2\pi)^4}\frac{\rd^4 q_2}{(2\pi)^4}\frac{\rd^4 k}{(2\pi)^2}
 \F{g}(q_2)\F{g}(q_1)\,\PV\frac{1}{k^2}\,
 \theta(p^0_1+k^0-\epsilon q^0_1)\delta((p_1+k-\epsilon q_1)^2-\mass^2)
 \\
 \times \theta(p^0_2-k^0-\epsilon q^0_2)\delta((p_2-k-\epsilon q_2)^2-\mass^2)
 \,f(p_1+k-\epsilon q_1,p_2-k-\epsilon q_2)
\end{multline}
and
\begin{multline}\label{eq:moller_as}
 F^{(II)}_\epsilon(p_1,p_2) = \ri\int \frac{\rd^4 q_1}{(2\pi)^3}\frac{\rd^4 q_2}{(2\pi)^3}\frac{\rd^4 k}{(2\pi)^4}
 \,\F{g}(q_1)\F{g}(q_2)\,\PV\frac{1}{k^2}
 \\
 \delta(2p_1\cdot(k-\epsilon q_1))\,\delta(2p_2\cdot(k+\epsilon q_2))  \,\F{\eta}(k-\epsilon q_1) \F{\eta}(-k-\epsilon q_2)\,f(p_1,p_2).
\end{multline}
Let us choose the reference frame such that
\begin{equation}
 p_1+p_2 = (E,0,0,0),~~~~~~p_1-p_2=(0,0,0,2Q).
\end{equation}
Since $f\in\cD_2$ it holds $f(p_1,p_2)=0$ if $\frac{p_1\cdot p_2}{\mass^2}-1 \geq \const>0$. Thus, for $\F{g}$ with the support in sufficiently small neighborhood of the origin $|Q|\geq \const>0$ in the region where the integrands in the expressions~\eqref{eq:moller_std} and~\eqref{eq:moller_as} are nonzero. In the case of the first of the above expressions it is the consequence of the approximate energy momentum conservation. Using one of the Dirac deltas in the expressions~\eqref{eq:moller_std} and~\eqref{eq:moller_as} we evaluate first the integrals over $k^0$. Next using another one we evaluate the integral over $k^3$. One shows the following pointwise convergence
\begin{multline}\label{eq:moller_limit}
 F^{(AB)}_0(p_1,p_2):= \lim_{\epsilon\searrow0}\,F^{(AB)}_\epsilon(p_1,p_2) 
 =
 \lim_{\epsilon\searrow0}\,(F^{(I)}_\epsilon(p_1,p_2)+F^{(II)}_\epsilon(p_1,p_2))
 \\
 =-\frac{\ri}{8 E}\int \frac{\rd k^1\rd k^2}{(2\pi)^2} \, \bigg\{ \sum_\pm 
 \left[
 \frac{\theta(Q^2-(k^1)^2-(k^2)^2)}{(Q^2-(k^1)^2-(k^2)^2)^{1/2}} 
 \frac{1}{\,K^2_\pm} f(p_1+K_\pm,p_2-K_\pm)
 \right] 
 \\
 -\frac{1}{|Q|} \frac{|\eta(0,k^1,k^2,0)|^2}{(-(k^1)^2-(k^2)^2)} f(p_1,p_2)  \bigg\},
\end{multline}
where
\begin{equation}
 K_\pm:=(0,k^1,k^2,Q\pm(Q^2-(k^1)^2-(k^2)^2)^{1/2}).
\end{equation}
The integrand in the expression~\eqref{eq:moller_limit} is a sum of three terms. The term with $+$ is integrable, whereas the term with $-$ and the term in the last line are not separately integrable but their sum is integrable. Since the supports of the functions $F^{(I)/(II)}_\epsilon$ are contained in a compact subset of $\Hm^{\times 2}$ independent of $\epsilon>0$, by the Lebesgue dominated convergence theorem, the pointwise convergence~\eqref{eq:moller_limit} implies the convergence in the sense of $L^2(\Hm^{\times 2})$. Observe that the function $F_0^{(AB)}$ can be also expressed in explicitly covariant form
\begin{multline}
 F^{(AB)}_0(p_1,p_2)=-\frac{\ri}{4}\int \frac{\rd^4 k}{(2\pi)^2}\frac{1}{k^2}
 \delta((p_1+p_2)\cdot k)
 \\
 \times 
 \left(
 \delta((p_1-p_2)\cdot k-k^2)
 f(p_1+k,p_2-k) 
 -\delta((p_1-p_2)\cdot k)\, |\F{\eta}(k)|^2\,f(p_1,p_2)\right).
\end{multline}
Diagrams with one incoming and three outgoing or three incoming and one outgoing electron/positron vanish trivially in the adiabatic limit by the approximate energy-momentum conservation. The same is true for diagrams with no incoming or outgoing particles.

The above reasoning can be easily adapted to the case of QED because the singular behavior of the integrands that appear in the expression for the wave functions in QED is the same as in the scalar model. Concluding,  Thm.~\ref{thm:adiabatic_first_second} restricted to the contributions considered in this section holds true with some nonzero $S^{[2]}_\rmod(\eta)\in L(\cD_1,\cH)$.

\section{First-order corrections to interacting fields}\label{sec:fields}

In this section we give explicit expressions for the first-order corrections to some of the modified interacting fields in the scalar model and QED. We investigate the LSZ limits of the modified interacting fields, determine their long range tails and draw some conclusions about the physical interpretation of states in the Fock space. We prove the theorem stated below.

\begin{thm}\label{thm:first_order_fields}
Let $\Psi\in\cD_1$, $h\in C^\infty_{\mathrm{c}}(\R^4)$ and let $C\in\{\psi,\varphi\}$ in the case of the scalar model and $C\in\{J^\mu,A_\mu,F_{\mu\nu}\}$ in the case of QED. For every profile $\eta$ (cf. Def.~\ref{dfn:profile}) and sector measure $\hat\varrho$ (cf. Def.~\ref{def:sector}) there exists $C^{[1]}_{\ret/\adv,\rmod}(\eta;h)\in L(\cD_1,\cH)$ such that
\begin{equation}
 \lim_{\epsilon\searrow0}\, \|(C^{[1]}_{\ret/\adv,\rmod}(\eta,g_\epsilon;h)-C^{[1]}_{\ret/\adv,\rmod}(\eta;h))\Psi\|=0.
\end{equation}
\end{thm}
 
In the proof of the above theorem we use the following formula for the modified interacting fields with the adiabatic cutoff
\begin{multline}
 C_{\ret/\adv,\rmod}(\eta,g;x) =\, \normord{C(x)} 
 \\
 + \ri e \int\rd^4 y\, g(y)\,\left( T(C(x),\cL(y)) -\frac{1}{2}[\normord{C(x)},\normord{\cL(y)}]_+ \pm \frac{1}{2}[\normord{C(x)},\normord{\cL(y)}]\right)
 \\
 +\ri e \int\rd^4 y\, g(y) \,[\normord{C(x)},\mp\cL_{\rin/\rout}(\eta;y)+\underline{\cL}(\eta;y)]
 +O(e^2),
\end{multline}
where $[\cdot,\cdot]_+$ is the anti-commutator, $\cL(x)$ is the standard interaction vertex in the scalar model or QED, $\cL_{\rout/\rin}(\eta;y)$ are the asymptotic vertices introduced in Def.~\ref{def:D_modifiers} and $\underline{\cL}(\eta;y)$ is the sector vertex introduced in Def.~\ref{def:sector}. If the polynomial $C$ coincides with one of the basic fields of the scalar model or QED, then the time-ordered product $T(C(x),\cL(y))$ does not require renormalization and it holds
\begin{multline}
 C_{\ret/\adv,\rmod}(\eta,g;x) =\, \normord{C(x)} + \ri e \int\rd^4 y\, g(y) \theta(\pm x^0\mp y^0)\, [\normord{C(x)},\normord{\cL(y)}]
 \\
 +\ri e \int\rd^4 y\, g(y) \left[\normord{C(x)},\mp\cL_{\rin/\rout}(\eta;y)+\underline{\cL}(\eta;y)\right]
 +O(e^2).
\end{multline}
In what follows we restrict attention to the retarded fields.

\subsection{Scalar model}\label{sec:fields_mod_first_order_corrections_scalar}

We recall that $\cL=J\varphi = \frac{1}{2}\psi^2\varphi$. Consider first the field describing the electron. The interacting field with adiabatic cutoff is given by
\begin{multline}
 \psi_{\ret,\rmod}(\eta,g;h) = \psi(h) + e \int \rd^4 x\rd^4 y \,h(x)\, D_\mass^\ret(x-y)\,g(y)\, \psi(y) \varphi(y) 
 \\
 + 
 \ri \frac{e}{2\mass} \int \rd^4 y\mHm{p}\, g(y)
 \left( b^*(p) j_\rin(\eta,p;y) \F{h}(p) -\mathrm{h.c.} \right)\varphi(y)
 + O(e^2).
\end{multline}
The standard first-order correction can be rewritten in the following form
\begin{multline}
 e \int\rd^4 x\rd^4 y\, h(x)\, D_\mass^\ret(x-y)\, g(y)\,\psi(y) \varphi(y) 
 \\
 = -e\int \rd \mu_0(k) \mHm{p}\frac{\rd^4 q}{(2\pi)^4}\,  \frac{\F{g}(q) \F{h}(k+p-q)~a^*(k) b^*(p)}{(p+k-q)^2-\mass^2-\ri \zerop \sgn (k^0+p^0-q^0)}
+\ldots,
\end{multline}
where on the RHS of the above equation we wrote explicitly only the term proportional to $a^*(k) b^*(p)$. The terms proportional to $a(k) b^*(p)$, $a^*(k) b(p)$ and $a(k) b(p)$ have a very similar form. Now let us consider the part of the first-order correction coming from the incoming Dollard modifiers
\begin{multline}
 \ri \frac{e}{2\mass} \int \rd^4 y\mHm{p}\, g(y)
 \left( b^*(p) j_\rin(\eta,p;y) \F{h}(p) -\mathrm{h.c.} \right)\varphi(y)
 \\
 =e\int \rd \mu_0(k) \mHm{p}\frac{\rd^4 q}{(2\pi)^4} \frac{\F{g}(q)\F{\eta}(k-q)\F{h}(p)~a^*(k) b^*(p)}{2p\cdot (k-q) - \ri \zerop}  + \ldots.
\end{multline}
Consequently, in the adiabatic limit we obtain
\begin{multline}
 \psi_{\ret,\rmod}(\eta;h)\Psi = \lim_{\epsilon\searrow0}\, \psi_{\ret,\rmod}(\eta,g_\epsilon;h)\Psi 
 \\
 = \psi(h)\Psi
 - e \int \mHm{p}\mHO{k}
 \frac{\F{h}(k+p)-\F{h}(p)\F{\eta}(k)}{2p\cdot k}a^*(k)b^*(p)\Psi + \ldots + O(e^2),
\end{multline}
where $\Psi\in\cD_1$. Thus, the first order correction $\psi^{[1]}_{\ret,\rmod}(\eta;h)$ is well defined as an element of $L(\cD_1)$. 

Now let us consider the modified retarded massless field. The field with adiabatic cutoff is given by
\begin{multline}
 \varphi_{\ret,\rmod}(\eta,g;x)
 =
 \varphi(x)
 +e\int \rd^4 y \,D_0^\ret(x-y) \,g(y)\normord{J(y)}
 \\
 -e \int \rd^4 y\, D_0(x-y) \,g(y)\,J_\rin(\eta;y) 
 + e\, \underline{\varphi}(\eta,g;x) 
 +O(e^2),
\end{multline}
where 
\begin{equation}
 \underline{\varphi}(\eta,g;x) := \int \rd^4 y\, D_0(x-y) \,g(y)\,\underline{J}(\eta;x).
\end{equation}
It is easy to see that the adiabatic limit of a function $\underline{\varphi}(\eta,g_\epsilon;x)$ exists and is a continuous function which we denote by $\underline{\varphi}(\eta;x)$. For $\Psi,\Psi'\in\cD_1$ it holds
\begin{multline}\label{eq:field_varphi_adiabaitc}
 (\Psi|\varphi_{\ret,\rmod}(\eta;x)\Psi')
 =
 \lim_{\epsilon\searrow0}\,(\Psi|\varphi_{\ret,\rmod}(\eta,g_\epsilon;x)\Psi')
 =
 (\Psi|\varphi(x)\Psi')
 \\
 +
 e\int \rd^4 y \, D_0^\ret(x-y) \, (\Psi|\normord{J(y)}\Psi')
 -
 e\int \rd^4 y \, D_0^\ret(x-y) \, (\Psi|J_\rin(\eta;y)\Psi') 
 \\
 +
 e \int \rd^4 y\,  D_0^\adv(x-y) \, (\Psi|J_\rin(\eta;y)\Psi')
 +
 e\,\underline{\varphi}(\eta;x)\,(\Psi|\Psi')
 +
 O(e^2),
\end{multline}
where the terms on the RHS of the above equation are well-defined continuous functions. To see this, it is enough to use the bounds 
\begin{equation}
 t^3\,|(\Psi|\normord{J(t,\vec{x})}\Psi')|,
 ~
 t^3\,|(\Psi|J_\rin(\eta;t,\vec{x})\Psi')|
 \leq \const,
\end{equation}
which follow from the estimates proved in Sec.~\ref{sec:asymp_current_behavior}, and note that for $f\in C^\infty(\R^4)$ such that $t^3\,|f(t,\vec x)|\leq \const$ the expression
\begin{equation}
 \int\rd^4 y\,D^{\ret/\adv}_0(x-y)f(y) 
 = \frac{1}{4\pi}\int\rd^3\vec y\,\frac{1}{|\vec x-\vec y|} 
 f(x^0\mp|\vec x-\vec y|,\vec y).
\end{equation}
is a continuous function of $x$. Using~\eqref{eq:field_varphi_adiabaitc} one can define $\varphi^{[1]}_{\ret,\rmod}(\eta;h)$ as a form in $\cD_1$. Moreover, it holds
\begin{multline}\label{eq:varphi_mod_correction}
 \varphi_{\ret,\rmod}(\eta;h)\Psi 
 = \lim_{\epsilon\searrow0}\, \varphi_{\ret,\rmod}(\eta,g_\epsilon;h)\Psi  
 \\
 = \varphi(h)\Psi 
 -\frac{e}{2} \int\mHm{p_1}\mHm{p_2}
 \bigg(\frac{\F{h}(p_1+p_2)}{2p_1\cdot p_2+2\mass^2}b^*(p_1)b^*(p_2) 
 \\
 -\frac{\F{h}(p_1-p_2)}{p_1\cdot p_2-\mass^2}b^*(p_1)b(p_2)
 + \frac{\F{h}(-p_1-p_2)}{2p_1\cdot p_2+2\mass^2}b(p_1)b(p_2) \bigg)\Psi
 \\
 -\frac{e}{2} \int\mHm{p}\frac{\rd^4 k}{(2\pi)^3}\,\sgn(k^0)\delta(k^2)\, \frac{\F{h}(-k)\F{\eta}(k)}{p\cdot k} b^*(p)b(p)\Psi
 +
 e\,\underline{\varphi}(\eta;x) \Psi
 +
 O(e^2),
\end{multline}
where $\Psi\in\cD_1$. Note that for every Schwartz functions $f$ and $h$ the function
\begin{equation}
 F(p_1):=\int\mHm{p_2}\,\frac{\F{h}(p_1-p_2)}{p_1\cdot p_2-\mass^2} f(p_2)
\end{equation}
is also a Schwartz function. Thus, the first-order correction $\varphi^{[1]}_{\ret,\rmod}(\eta;h)$ is well defined as an element of $L(\cD_1)$. 

Let us determine the long-range tail of $\varphi_{\ret,\rmod}(\eta;x)$. The long range tail of a field is defined in Appendix~\ref{sec:long_rang_tail}. One shows that the long-range tail of the free massless field vanishes, the long range tails of the second and third term on the RHS of Eq.~\eqref{eq:field_varphi_adiabaitc} cancel each other and the long-range tail of the fourth term on the RHS of Eq.~\eqref{eq:field_varphi_adiabaitc} vanishes. Thus, the long-range tail of $\varphi_{\ret,\rmod}(\eta;x)$ coincides with the long-range tail of $\underline{\varphi}(\eta,x)$ and is equal to
\begin{equation}
 \lim_{R\to\infty}\,R\,(\Psi|\varphi_{\ret,\rmod}(\eta,x+Rn)\Psi') = 
 e \int_{\Hm} \rd\hat\varrho(p) \,\frac{1}{((n\cdot \frac{p}{m})^2-n^2)^{1/2}}
 \,(\Psi|\Psi')
 +O(e^2),
\end{equation}
where $n$ is a unit spatial four-vector and $\Psi\in\cD_0$. One shows that the corrections of order $e^2$ vanish. In particular, in the case $\hat\varrho=0$ the long-range tail of $\varphi_{\ret,\rmod}(\eta,x)$ is trivial. 

Let us study the LSZ limits of the modified retarded fields. The outgoing and incoming LSZ limits are defined in Appendix~\ref{sec:LSZ} and are used to construct the operators that create or annihilate a single particle when acting on a state in the asymptotic Hilbert spaces. The past LSZ limits of the retarded fields and the future LSZ limits of the advanced fields have simple forms. For example, in the case of purely massive models these limits coincide with the standard creation and annihilation operators defined in the Fock space. This is not true in the scalar model. The past LSZ limit of the massive modified retarded field $\psi_{\ret,\rmod}(\eta;h)$ exists only as a form on $\cD_1$. Thus, it cannot be used in the construction of asymptotic states. The massive interacting field in the scalar model has a non-standard asymptotic behavior. In order to define the asymptotic massive field one would have to modify the LSZ limit in an appropriate way which is outside the scope of this paper. Using~\eqref{eq:field_varphi_adiabaitc} and the estimates from Sec.~\ref{sec:asymp_current_behavior} we determine the past LSZ limits of the massless retarded field
\begin{align}
 \lim_{t\to-\infty}\,(-\ri) \int\rd^3\vec x~ f_t(\vec x)\,
 \overset{\leftrightarrow}{\partial_t}\,(\Psi|\varphi_{\ret,\rmod}(\eta;t,\vec{x})\Psi') 
 &= (\Psi|a^*(\eta,f)\Psi')+ O(e^2),
 \\
 \lim_{t\to-\infty}\,(-\ri) \int\rd^3\vec x~ \overline{f_t(\vec x)}\,
 \overset{\leftrightarrow}{\partial_t}\,(\Psi|\varphi_{\ret,\rmod}(\eta;t,\vec{x})\Psi') 
 &= (\Psi|a(\eta,f)\Psi')+ O(e^2), 
\end{align}
where
\begin{equation}\label{eq:asymp_a_scalar}
 a(\eta,k) := a(k) - e\,\F{\eta}(k) \int_{\Hm} (\rd\rho-\rd\hat\varrho)(p)\, \frac{\mass}{p\cdot k}.
\end{equation}
The above operators $a^*(\eta,k)$ and $a(\eta,k)$ are the creation and annihilation operators of the physical incoming photons. Note that in the case $\hat\varrho=0$ the ground state $\Omega$ in the Fock space is annihilated by $a(\eta,k)$ and, thus, does not contain any incoming photons. In contrast, all states with at least one electron contain infinite number of photons. Eq.~\ref{eq:asymp_a_scalar} implies that the cloud of photons is correlated with the momentum of the electron and depends on the profile $\eta$. In particular, the momenta of photons in the cloud belong to the support of the Fourier transform of $\eta$.

\subsection{QED}\label{sec:fields_mod_first_order_corrections_qed}

We recall that $\cL=A_\mu J^\mu=\overline{\psi}\slashed{A}\psi$. Using the method presented in the case of the scalar model we show that the adiabatic limit of the interacting vector potential is given by
\begin{multline}
 (\Psi|A_{\ret/\adv,\rmod}^{\mu}(\eta;x)\Psi') = (\Psi|A^\mu(x)\Psi')
 - e\int \rd^4 y \, D_0^{\ret/\adv}(x-y) ~(\Psi|\normord{J^\mu(y)}\Psi')
 \\
 +e \int \rd^4 y\, D_0(x-y)~(\Psi|\left(\pm J^\mu_{\rin/\rout}(\eta;y) - \underline{J}^\mu(\eta;y)\right)\Psi') + O(e^2).
\end{multline}
The above expression defines the first order correction $A_{\ret/\adv,\rmod}^{[1]\mu}(\eta;x)$ as a form on $\cD_1$. One can also define $A_{\ret/\adv,\rmod}^{[1]\mu}(\eta;x)$ as an operator in $L(\cD_1)$ by the expression analogous to the expression~\eqref{eq:varphi_mod_correction} used to define the first order correction to the interacting massless field in the scalar model. The interacting potential is obviously not BRST invariant. However, one easily verifies that the interacting electromagnetic field $F_{\ret/\adv,\rmod}^{\mu\nu}(\eta;x)$ is BRST invariant. Hence, it induces an operator in the physical Hilbert space. 

Now, let us consider the modified retarded spinor current. We assume that the time-ordered products are normalized such that $\F{\underline{t}}_{2}^{[2]}(q) = O(|q|^3)$, where $\F{\underline{t}}_{2}^{[2]}$ is defined by Eq.~\ref{eq:normalization_st}. For $\Psi\in\cD_1$ we obtain
\begin{multline}
 J^\mu_{\ret,\rmod}(\eta;h_\mu)\Psi = 
 \lim_{\epsilon\searrow0}
 J^\mu_{\ret,\rmod}(\eta,g_\epsilon;h_\mu)\Psi
 \\
 =\,\normord{J^{\mu}(h_\mu)}\Psi
 -e \sum_{\sigma,\sigma'=1,2}\int\mHO{k}\mHm{p}\mHm{p'} 
 \\
 \bigg[
 (\F{h}_\mu(-p+k+p')-\F{h}_\mu(-p+p')\F{\eta}(k))\, \overline{u}(\sigma,p)\gamma^\mu u(\sigma',p')
 \left(\frac{p'^\nu}{p'\cdot k} - \frac{p^\nu}{p\cdot k}\right)
 \\
 +\F{h}_\mu(-p+k+p')\,\overline{u}(\sigma,p)
 \left(\frac{\gamma^\mu \slashed{k}\gamma^\nu}{2p'\cdot k}
 -\frac{\gamma^\nu \slashed{k}\gamma^\mu}{2p\cdot k}\right)
 u(\sigma',p')\bigg]
 \\[1mm]
 \times \, a_\nu^*(k)b^*(\sigma',p') b(\sigma,p) \Psi
 +\ldots + O(e^2),
\end{multline}
where we wrote explicitly only the term proportional to $a_\nu^*(k)b^*(\sigma',p') b(\sigma,p)$. All terms of order $O(e)$ are proportional to $a_\nu^\#(k)b^\#(\sigma',p') b^\#(\sigma,p)$ or $a_\nu^\#(k)d^\#(\sigma',p') d^\#(\sigma,p)$ and have a form similar to the form of the term displayed above. The above expression defines $J_{\ret,\rmod}^{[1]\mu}(\eta;h_\mu)$ as an operator in $L(\cD_1)$. One easily shows that it commutes with the BRST charge. Hence, it induces an operator in the physical Hilbert space.

Using estimates from Appendix 1 to Section XI.3 in~\cite{reedsimon3} we determine the past asymptote of the interacting retarded current 
\begin{equation}
 \lim_{\lambda\to\infty}\,(\Psi|J_{\ret,\rmod}(\eta;-\lambda v)\Psi) = \lim_{\lambda\to\infty}\,(\Psi|\normord{J(-\lambda v)}\Psi) = \frac{\mass^2}{2(2\pi)^3} \,v^\mu\, (\Psi|\rho(\mass v)\Psi'),
\end{equation}
where $\Psi,\Psi'\in \cD_0$, $v$ is a four-velocity, $\rho(p)$ is defined by Eq.~\eqref{eq:def_charge_dist_momentum} and the last equality follows from the results of Sec.~\ref{sec:asymp_current_behavior}. The future asymptote of the advanced interacting current has the same form. We conclude that the radiative corrections do not change the asymptotes of the current.

The long-range tail of the interacting electromagnetic field has the following form 
\begin{equation}
 \lim_{R\to\infty}\,R^2\,(\Psi|F^{\mu\nu}_{\ret/\adv,\rmod}(\eta;x+Rn)\Psi) = -e \int_{\Hm} (\Psi|\rd\hat\varrho(p) \Psi') \,f^{\mu\nu}(p/\mass;n),
\end{equation}
where $f^{\mu\nu}$ is the Coulomb field defined by Eq.~\eqref{eq:Coulomb_field} and the measure $\hat\varrho = \varrho_0 + Q\varrho_1$ (cf. Def.~\ref{def:sector}). If $\varrho_0=0$ and $\varrho_1=\varrho_{m\mathrm{v}}$ is the Dirac measure at $m\mathrm{v}$, then
\begin{equation}
 \lim_{R\to\infty}\,R^2\,(\Psi|F^{\mu\nu}_{\ret/\adv,\rmod}(\eta;x+Rn)\Psi) = -e\, (\Psi|Q \Psi') \,f^{\mu\nu}(\mathrm{v};n).
\end{equation}
We see that the asymptotic flux of the electric field is nontrivial only in sectors with non-zero electric charge $Q$. The asymptotic flux of the electric field integrated over the sphere $S^2$ of spatial directions gives $-eQ$ in compliance with the Gauss law. The past/future LSZ limits $a_\mu(\eta;k)$ and $a^*_\mu(\eta;k)$ of the interacting retarded/advanced vector potential are given by
\begin{equation}\label{eq:asymp_a_qed}
 a_\mu(\eta,k) = a_\mu(k) +e\,\F{\eta}(k) \int_{\Hm}(\rd\rho-\rd\hat\varrho)(p)\frac{p_\mu}{p\cdot k}.
\end{equation}
Let $\varepsilon^\mu(s,k)$, $s=1,2$, be the polarization vectors of photons introduced in Sec.~\ref{sec:qed}. The operators
\begin{equation}\label{eq:asymp_a_qed_phys}
 a^\#(\eta;s,k):=\varepsilon^\mu(s,k)\,[a_\mu^\#(\eta,k)],
\end{equation}
defined in the physical Hilbert space $\cH^{\textrm{phys}}$, are the creation and annihilation operators of the physical incoming photons. As in the case of the scalar model the incoming LSZ field are obtained from the free field by a coherent transformation. If $\varrho_0=0$, then the ground state $\Omega$ does not contain any incoming photons. However, states with electrons or positrons are never annihilated by $a(\eta;s,k)$. In fact, these states contain infinite number of incoming photons.

Finally, let us consider the first-order corrections to the modified retarded spinor field. For simplicity we assume that $\hat\varrho = Q\varrho_{m\mathrm{v}}$. We find that for $\Psi,\Psi'\in\cD_1$ it holds
\begin{multline}
 (\Psi|\psi_{\ret,\rmod}(\eta;f)\Psi') = \lim_{\epsilon\searrow0} (\Psi|\psi_{\ret,\rmod}(\eta,g_\epsilon;f)\Psi')  
 =(\Psi| \psi(f)\Psi')
 \\
 -e\sum_{\sigma=1,2}\int \mHm{p}\mHO{k} 
 \bigg[
 \frac{\F{f}(k+p)\, (\slashed{p}+\slashed{k}+m)\gamma^\mu 
 -2\F{f}(p)\,p^\mu\F{\eta}(k)}{2p\cdot k}u(\sigma,p)
 \\
 -\frac{\F{f}(p) \, \mathrm{v}^\mu\F{\eta}(k)\,u(\sigma,p)}{\mathrm{v}\cdot k}\bigg](\Psi|a_\mu^*(k)b^*(\sigma,p)\Psi') 
 +\ldots + O(e^2),
\end{multline}
where we wrote explicitly only the term proportional to $a_\mu^*(k)b^*(\sigma,p)$. Because the term depending on the four-velocity $\mathrm{v}$ behaves like $O(|k|^{-1})$ the correction $\psi^{[1]}_{\ret,\rmod}(\eta;f)$ cannot be interpreted as an operator. It makes sense as a form on $\cD_1$. However, it is not BRST invariant, $[Q_{\mathrm{BRST}},\psi^{[1]}_{\ret,\rmod}(\eta;f)]\neq 0$. Moreover, we do not expect that higher order corrections to $\psi_{\ret,\rmod}(\eta;f)$ exist even as forms. As specified in Conjecture~\ref{cnj:S_mod}, our proposal for the modified interacting fields works only for observable fields. In order to define the interacting Dirac field one has to introduce further modifications in the definition of $\psi_{\ret,\rmod}(\eta,g;x)$ such that the resulting expression is formally gauge invariant in a sense explained in Sec.~\ref{sec:qed}. Because the interacting Dirac field interpolates between sectors of different electric charge it cannot be a point-local field. In fact, its localization region cannot be bounded as otherwise the Gauss law would be violated. It is plausible that one can construct a modified retarded spinor field localized in a semi-infinite string using ideas of~\cite{cardoso2018string,mandelstam1968feynman}. However, in order to define a charged field localized in some spatial cone one has to replace the current $j_\rin^\mu$ in Def.~\ref{def:sector} with some current supported in such a cone. We leave this issue for future investigation.

\section{Energy-momentum operators}\label{sec:energy_momentum}

In this section we study the covariance of the modified scattering matrix and the modified interacting fields under space-time translations. In the case of quantum-mechanical models the modified scattering matrix constructed with the use of the Dollard procedure typically commutes with the free dynamics. This is usually not the case in field theory models~\cite{morchio2016dynamics}. Assuming that Conjecture~\ref{cnj:S_mod} holds true we show that the modified scattering matrix and the modified interacting fields in the scalar model and QED are covariant under the action of some unitary representation of the translation group which is not equivalent to the standard Fock representation. We give the explicit expressions for the energy-momentum operators in the scalar model and QED and investigate their joint spectrum. 

\subsection{Scalar model}\label{sec:energy_momentum_scalar}

Let us indicate how different objects transform under spacetime translations. The standard Bogoliubov scattering matrix, given by~\eqref{eq:bogoliubov_S_op}, and the standard interacting fields, given by~\eqref{eq:bogoliubov_fields}, have the following properties
\begin{equation}
\begin{aligned}
 U(a) S(g) U(a)^{-1} 
 =&\,
 S(g_a),
 \\
 U(a) C_{\ret/\adv}(g;h) U(a)^{-1} 
 =&\,C_{\ret/\adv}(g_a;h_a),
\end{aligned} 
\end{equation}
where $U(a)$, $a\in\R^4$, is the standard unitary representation of the group of translation in the Fock space,
$g$ is the adiabatic cutoff, $g_a(x):=g(x-a)$, $h\in\cS(\R^4)$ and $h_a(x):=h(x-a)$. In case of the Dollard modifiers, given by~\eqref{eq:def_modifiers}, the modified scattering matrix with adiabatic cutoff, given by~\eqref{eq:mod_S_op}, and the modified interacting fields, given by~\eqref{eq:mod_fields}, we have
\begin{equation}
\begin{aligned}
 U(a) S^\ras_{\rout/\rin}(\eta,g) U(a)^{-1} 
 =&\,
 S^\ras_{\rout/\rin}(\eta_a,g_a),
 \\
 U(a) S_\rmod(\eta,g) U(a)^{-1} 
 =&\,
 S_\rmod(\eta_a,g_a),
 \\
 U(a)C_{\ret/\adv,\rmod}(\eta,g;h)U(a)^{-1} 
 =&\, C_{\ret/\adv,\rmod}(\eta_a,g_a;h_a),
\end{aligned} 
\end{equation}
where $\eta$ and $\eta_a(x):=\eta(x-a)$ are profiles, that is real-valued Schwartz functions with integral over spacetime normalized to one (cf. Def.~\ref{dfn:profile}). From now on let us assume that the adiabatic limit of the modified scattering matrix and modified interacting fields exists in the sense of Conjecture~\ref{cnj:S_mod}. The fact that $U(a)\Psi\in\cD_2$ if $\Psi\in\cD_2$ implies that for any $a\in\R^4$ and $\Psi,\Psi'\in\cD_2$ the limit
\begin{multline}\label{eq:S_mod_translation1}
 \lim_{\epsilon\searrow0}\, (\Psi|S_\rmod(\eta_a,(g_{\epsilon a})_\epsilon)\Psi')
 =
 \lim_{\epsilon\searrow0}\, (\Psi|U(a) S_\rmod(\eta,g_\epsilon) U(a)^{-1}\Psi')
 \\
 =(\Psi|U(a) S_\rmod(\eta) U(a)^{-1}\Psi')
\end{multline}
exists, where 
\begin{equation}
 (g_a)_\epsilon(x) = g_a(\epsilon x)= g(\epsilon x-a),
 \quad
 (g_{\epsilon a})_\epsilon(x) = g(\epsilon (x-a)) = g_\epsilon(x-a).
\end{equation}
It holds
\begin{multline}\label{eq:S_mod_translation2}
 \lim_{\epsilon\searrow0}\,(\Psi|S^{[n]}_\rmod(\eta_a,(g_{\epsilon a})_\epsilon)\Psi') 
 =
 \lim_{\epsilon\searrow0}\,\int\mP{q_1}\ldots\mP{q_n}\,
 \F{g}_\epsilon(-q_1)\ldots\F{g}_\epsilon(-q_n)\,
 \\
 \times
 \re^{-\ri (q_1+\ldots+q_n)\cdot a}\, 
 S^{[n]}_\rmod(\eta_a;q_1,\ldots,q_n)
 \\
 = 
 \lim_{\epsilon\searrow0}\,\int\mP{q_1}\ldots\mP{q_n}\,
 \F{g}_\epsilon(-q_1)\ldots\F{g}_\epsilon(-q_n)\,
 S^{[n]}_\rmod(\eta_a;q_1,\ldots,q_n)
 \\
 =
 \lim_{\epsilon\searrow0}\,(\Psi|S^{[n]}_\rmod(\eta_a,g_\epsilon)\Psi')
 =
 (\Psi|S^{[n]}_\rmod(\eta_a)\Psi')
\end{multline}
where $S^{[n]}_\rmod(\eta,q_1,\ldots,q_n)\in\cS'(\R^{4n})$ is given by~\eqref{eq:S_mod_sym_dist}. The second of the above identities is a consequence of Thm.~\ref{thm:dist_value} stated in Appendix~\ref{sec:value_dist}. Combining Eqs.~\eqref{eq:S_mod_translation1} and~\eqref{eq:S_mod_translation2} we get the following equality
\begin{equation}\label{eq:scalar_covariance_S_op}
 U(a) S_\rmod(\eta) U(a)^{-1}=S_\rmod(\eta_a).
\end{equation}
In an analogous way one proves that
\begin{equation}\label{eq:scalar_covariance_fields}
 U(a)C_{\ret/\adv,\rmod}(\eta;h)U(a)^{-1} 
 =\, C_{\ret/\adv,\rmod}(\eta_a;h_a).
\end{equation}
Using Def.~\ref{def:intertwiners} and Eqs.~\eqref{eq:phi_finite_trans} we note that $V_\rout(\eta,\eta_a)=V_\rin(\eta,\eta_a)=:V(\eta,\eta_a)$. Next, by Eqs. \eqref{eq:intertwiners_properties} and~\eqref{eq:S_op_mod_V} we conclude that
\begin{equation}\label{eq:covariance_s_op_scalar_eta}
\begin{aligned}
 U(a)S_\rmod(\eta)U(a)^{-1} 
 =&\,
 V(\eta,\eta_a)^{-1}S_\rmod(\eta)V(\eta,\eta_a),
 \\
 U(a)C_{\ret/\adv,\rmod}(\eta;h)U(a)^{-1}
 =&\,
 V(\eta,\eta_a)^{-1} C_{\ret/\adv,\rmod}(\eta;h_a)V(\eta,\eta_a).
\end{aligned}
\end{equation}

\begin{dfn}\label{def:translation_scalar}
Let $\eta$ be a profile (cf. Def.~\ref{dfn:profile}) and $\hat\varrho$ be a sector measure (cf. Def.~\ref{def:sector}). Set
\begin{equation}\label{eq:scalar_translation_mod}
 \R^4 \ni a \mapsto  U_\rmod(\eta;a):=V(\eta,\eta_a) U(a) \in B(\cH).
\end{equation}
\end{dfn}
By Thm.~\ref{thm:mod_translation} the map~\eqref{eq:scalar_translation_mod} is a unitary representation of the translation group. The following theorem is a trivial consequence of Eqs.~\eqref{eq:scalar_covariance_S_op}, \eqref{eq:scalar_covariance_fields} and \eqref{eq:covariance_s_op_scalar_eta}. It states that the modified scattering operator and the modified interacting fields are covariant with respect to the representation of the translation group defined above.
\begin{thm}
Assume that Conjecture~\ref{cnj:S_mod} holds true in the case of the scalar model. For arbitrary $a\in \R^4$
\begin{equation}
\begin{aligned}
 U_\rmod(\eta;a) S_\rmod(\eta)U_\rmod(\eta;a)^{-1} 
 =&\, S_\rmod(\eta),
 \\
 U_\rmod(\eta;a) C_{\ret/\adv,\rmod}(\eta;h)U_\rmod(\eta;a)^{-1}
 =&\, C_{\ret/\adv,\rmod}(\eta;h_a).
\end{aligned} 
\end{equation}
\end{thm}
\begin{thm}\label{thm:mod_translation}
(A) For any profile $\eta$ and a sector measure $\hat\varrho$ the map~\eqref{eq:scalar_translation_mod} is a strongly continuous unitary representation of the group of spacetime translations which is not unitarily equivalent to the standard representation $U(a)$. 
\\
(B) The generators $P_\rmod(\eta)$ of $U_\rmod(\eta,a)=\exp(\ri P_\rmod(\eta)\cdot a)$ are explicitly given by
\begin{equation}\label{eq:energy_momentum_scalar}
P_\rmod^\mu(\eta) =\int \mHO{k}\, k^\mu a^*(\eta,k) a(\eta,k) + \int \mHm{p}\, p^\mu b^*(p) b(p), 
\end{equation}
where $a^\#(\eta;k)$ are defined by~\eqref{eq:asymp_a_scalar} and $b^\#(p)$ are the standard electron annihilation and creation operators.
\\
(C) The joint spectrum of $P_\rmod^\mu(\eta)$ coincides with the closed forward lightcone. If $\hat\varrho=0$, then $P_\rmod^\mu(\eta)\Omega=0$ and $P_\rmod^\mu(\eta)|k) = k^\mu |k)$. Otherwise, there is no vacuum and massless one-particle states in the joint spectrum of $P_\rmod^\mu(\eta)$. Regardless of the choice of $\hat\varrho$ there are no massive one-particle states in the joint spectrum of $P_\rmod^\mu(\eta)$.
\end{thm}
\begin{proof}
We first note that by Eqs.~\eqref{eq:dfn_intertwiner} and~\eqref{eq:phi_finite_trans} it holds 
\begin{multline}
 V(\eta,\eta_a)
 =
 \exp\left(e\int_{\Hm\times H_0} (\rd\rho-\rd\varrho)(p)\,\mHO{k}
 \left[\frac{\F{\eta}(k)(1-\re^{\ri k\cdot a})}{2p\cdot k} a^*(k) - \mathrm{h.c.}\right]\right)
 \\
 \exp\left(\-\ri e^2 \int_{\Hm\times\Hm\times H_0} 
 (\rd\rho-\rd\varrho)(p_1)\,(\rd\rho-\rd\varrho)(p_2)\mHO{k}\,
 \frac{|\F{\eta}(k)|^2\sin(k\cdot a)}{4(p_1\cdot k)(p_2\cdot k)}\right).
\end{multline}
The above operator is isometric and invertible, hence unitary. Thus, the operators $U_\rmod(\eta,a)$ are unitary. One easily verifies that  
\begin{equation}
 U(a_1) V(\eta,\eta_{a_2})  U(a_1)^{-1} = V(\eta_{a_1},\eta_{a_1+a_2}).
\end{equation}
Using the above identity and Eq.~\eqref{eq:intertwiners_properties} we show that $V(\eta,\eta_a)$ satisfies the following cocyle condition
\begin{equation}
 V(\eta,\eta_{(a_1+a_2)}) = V(\eta,\eta_{a_1}) U(a_1) V(\eta,\eta_{a_2})  U(a_1)^{-1}.
\end{equation}
This proves that $U_\rmod(\eta,a)$ is a unitary representation of the group of translations. It is easy to check that this representation is strongly continuous. At least formally, its generators $P_\rmod^\mu(\eta)$ are given by~\eqref{eq:energy_momentum_scalar}. In order to show that $P_\rmod^\mu(\eta)$, $\mu=0,1,2,3$, is a family of commuting self-adjoint operators and find its joint spectrum we make the following decomposition of the Fock-Hilbert space
\begin{equation}
 \cH=\oplus_{n=1}^\infty \cH_n,
 \quad
 \cH_n \simeq L_s^2(\Hm^{\times n},\rd\mu_\mass^{\times n})\otimes \Gamma_s(\mathfrak{h}_0),
\end{equation}
where $L_s^2(\Hm^{\times n},\rd\mu_\mass^{\times n})$ is the subspace of $L^2(\Hm^{\times n},\rd\mu_\mass^{\times n})$ consisting of functions symmetric under permutations of their arguments and $\Gamma_s(\mathfrak{h}_0)$ is the Fock space of photons in the scalar model. It turns out that $P_\rmod^\mu(\eta)$ is a decomposable operator in the sense of the definition given in Sec.~XIII.16 of~\cite{reedsimon4}. We have
\begin{equation}\label{eq:energy_momentum_decomposition}
 P_\rmod^\mu(\eta) = \oplus_{n=0}^\infty K_\rmod^{(n)\mu}(\eta),
\end{equation}
where
\begin{equation}
 K_\rmod^{(0)\mu}(\eta) \equiv K^{\mu} = \int \mHO{k}\, k^\mu a^*(k) a(k)
\end{equation}
and
\begin{equation}
 (K_\rmod^{(n)\mu}(\eta) \Psi)(p_1,\ldots,p_n) 
 = 
 \left(\sum_{j=1}^n p^\mu_j  + K^{(n)\mu}_\rmod(\eta,p_1,\ldots,p_n) \right) \Psi(p_1,\ldots,p_n).
\end{equation}
For $n\in\N_+$ and $p_1,\ldots,p_n\in\Hm$ the fiber operators
\begin{equation}
 K^{(n)\mu}_\rmod(\eta,p_1,\ldots,p_n)
 = \int \mHO{k}\, k^\mu \left(a^*(k) - \sum_{j=1}^n\frac{e\,\overline{\eta(k)}}{2p_j\cdot k} \right) \left(a(k)-\sum_{j=1}^n\frac{e\,\eta(k)}{2p_j\cdot k}\right)
\end{equation}
are essentially self-adjoint on $\textrm{Dom}\,K^0$ by the Nelson commutator theorem applied with the comparison operator $\id + K^0$. Using the method of Proposition 3.13 of~\cite{derezinski2003van} we show that for any $n\in\N_+$, $p_1,\ldots,p_n\in\Hm$ and $\mu=0,1,2,3$ the spectrum of $K^{(n)\mu}_\rmod(\eta,p_1,\ldots,p_n)$ is absolutely continuous. It is well-known that the joint spectrum of operators $K^\mu$, $\mu=0,1,2,3$, coincides with $\overline{V^+}$. The bound  
\begin{equation}
 g_{\mu\nu}\, (\Psi|K^{(n)\mu}_\rmod(\eta,p_1,\ldots,p_n) K^{(n)\nu}_\rmod(\eta,p_1,\ldots,p_n)\Psi) \geq 0
\end{equation}
implies that for any $n\in\N_+$ and $p_1,\ldots,p_n\in\Hm$ the joint spectrum of operators $K^{(n)\mu}_\rmod(\eta,p_1,\ldots,p_n)$, $\mu=0,1,2,3$, is contained in $\overline{V}_+$. Actually, using the method of the proof of Proposition 3.10 of~\cite{derezinski2003van} one can show that the joint spectrum coincides with $\overline{V}_+$. Thm. XIII.85 of~\cite{reedsimon4} implies the operators $P_\rmod^\mu(\eta)$ are self-adjoint on the domain specified in Sec.~XIII.16 of that reference. Using the decomposition~\eqref{eq:energy_momentum_decomposition} and the properties of the fiber operators it is straightforward to prove the claims about the joint spectrum of $P_\rmod^\mu(\eta)$.
\end{proof}

\subsection{QED}\label{sec:energy_momentum_qed}

\begin{dfn}\label{def:translation_qed}
Let $\eta$ be a profile (cf. Def.~\ref{dfn:profile}) and $\hat\varrho$ be a sector measure (cf. Def.~\ref{def:sector}). Set
\begin{equation}\label{eq:qed_translation_mod}
 \R^4 \ni a \mapsto  U_\rmod(\eta;a):=V(\eta,\eta_a) U(a) \in B(\cH),
\end{equation}
where $\eta_a(x) = \eta(x-a)$, $V(\eta,\eta_a)$ is the intertwining operator given by~\eqref{eq:dfn_intertwiner} and $U(a)$ is the standard representation of the translation group in the Fock space~$\cH$. The map~\eqref{eq:qed_translation_mod} induces a unique map
\begin{equation}\label{eq:qed_translation_mod_phys}
 \R^4 \ni a \mapsto  [U_\rmod(\eta;a)]\in B(\cH^{\mathrm{phys}}).
\end{equation}
\end{dfn}
By the theorem stated below the physical modified scattering matrix and interacting fields are covariant with respect to the representation of the spacetime translations defined above. In Thm.~\ref{thm:mod_translation_qed} we state some properties of this representation. The proofs of these theorems are similar to the proofs of the analogous theorems stated in the previous section and are omitted.
\begin{thm}
Assume that Conjecture~\ref{cnj:S_mod} holds true in the case of QED. For arbitrary $a\in \R^4$
\begin{equation}
\begin{aligned}
 [U_\rmod(\eta;a)] [S_\rmod(\eta)][U_\rmod(\eta;a)]^{-1} 
 =&\, [S_\rmod(\eta)]
 \\
 [U_\rmod(\eta;a)] [C_{\ret/\adv,\rmod}(\eta;h)][U_\rmod(\eta;a)]^{-1}
 =&\, [C_{\ret/\adv,\rmod}(\eta;h_a)].
\end{aligned} 
\end{equation}
\end{thm}
\begin{thm}\label{thm:mod_translation_qed}
(A) For any profile $\eta$ and a sector measure $\hat\varrho$ the map~\eqref{eq:qed_translation_mod_phys} is a strongly continuous unitary representation of the group of spacetime translations which is not unitarily equivalent to the standard representation $[U(a)]$. 
\\
(B) The generators of $[U_\rmod(\eta;a)]$ are given by
\begin{multline}
 [P_\rmod^\mu(\eta)]
 =
 \sum_{s=1,2} \int \mHO{k}\, k^\mu a^*(\eta;s,k) a(\eta;s,k) 
 \\
 +
 \sum_{\sigma=1,2}\int \mHm{p}\, 
 p^\mu (b^*(\sigma,p) b(\sigma,p) +d^*(\sigma,p) d(\sigma,p)), 
\end{multline}
where $a^\#(\eta;s,k)$ are defined by~\eqref{eq:asymp_a_qed_phys} and $b^\#(\sigma,p)$, $d^\#(\sigma,p)$ are defined in Sec.~\ref{sec:qed}.
\\
(C) The joint spectrum of $[P_\rmod^\mu(\eta)]$ coincides with the closed forward lightcone. If $\varrho_0=0$, then $[P_\rmod^\mu(\eta)]\Omega=0$ and $[P_\rmod^\mu(\eta)]|k,s) = k^\mu |k,s)$, where $|k,s)=a^*(k,s)\Omega$. Otherwise, there is no vacuum and massless one-particle states in the joint spectrum of $[P_\rmod^\mu(\eta)]$. Regardless of the choice of $\hat\varrho$ there are no massive one-particle states in the joint spectrum of $[P_\rmod^\mu(\eta)]$.
\end{thm}

\section{Physical interpretation of construction}\label{sec:physical_interpretation}

In this section we give the physical interpretation of the modified scattering matrix and the modified interacting fields. We also present a construction of the inclusive cross section that uses our scattering matrix.

\subsection{Space of asymptotic states}\label{sec:state_space}

The reformulation of the construction of the modifies S-matrix and interacting fields presented in this section makes transparent the physical interpretation of the asymptotic states. A similar reformulation of the modified scattering theory in the context of quantum mechanics was outlined in Sec.~\ref{sec:Dollard}.

Let $\mathcal{P}$ be the set of profiles $\eta$ (cf. Def.~\ref{dfn:profile}) and $\mathcal{Q}$ be the set of sector measures $\varrho$ (cf. Def.~\ref{def:varrho}). Moreover, in the case of the scalar model let $\cH$ be the standard Fock-Hilbert space. In the QED case let $\cH$ be equal to $1_q(Q)\cH^{\mathrm{phys}}$ for some $q\in\mathrm{sp}(Q)=\Z$, where $1_q(Q)$ is the projection onto the eigenspace of the charge operator $Q$ with eigenvalue $q$ and $\cH^{\mathrm{phys}}$ is the physical Fock space 
with a positive-definite inner product defined in Sec.~\ref{sec:qed}. In what follows we identify BRST invariant operators defined in the standard Fock space containing unphysical degrees of freedom  with operators they induced in $\cH^{\mathrm{phys}}$. 

For any $\varrho\in\mathcal{Q}$ we define the asymptotic Hilbert spaces $\hat\cH_{\rout/\rin}(\varrho)$ which, by definition, consist of the following vectors
\begin{equation}
 \hat\Psi\in \hat{\cH}_{\rin/\rout}(\varrho)
 \quad\textrm{iff}\quad
 \hat\Psi:\mathcal{P}\to \cH
 \quad\textrm{and}\quad
 \hat\Psi(\eta)=V_{\rout/\rin}(\eta,\eta')\hat\Psi(\eta'),
\end{equation}
where $V_{\rout/\rin}(\eta,\eta',\varrho)\equiv V_{\rout/\rin}(\eta,\eta')$ is the intertwining operator defined by Eq.~\eqref{eq:dfn_intertwiner} and in the case of QED $\cH=1_q(Q)\cH^{\mathrm{phys}}$ with $q=\varrho(\Hm)\in\Z$. The scalar product in $\hat{\cH}_{\rin/\rout}(\varrho)$ is given by
\begin{equation}
 (\hat\Psi|\hat\Psi'):=(\hat\Psi(\eta)|\hat\Psi'(\eta))
\end{equation}
and is independent of the choice of $\eta\in\mathcal{P}$ by the unitarity of $V_{\rout/\rin}(\eta,\eta')$. Let us define the vector $\hat\Omega\in\hat\cH_{\rout/\rin}(\varrho)$ by
\begin{equation}
 \hat\Omega\,:\mathcal{P}\ni\eta\mapsto \Omega\in\cH,
\end{equation}
where $\Omega$ is the ground state in the Fock space $\cH$. We also define a dense domain  $\hat\cD_{\rout/\rin}(\varrho)\subset \hat\cH_{\rout/\rin}(\varrho)$ consisting of vectors $\hat\Psi$ such that $\hat\Psi(\eta)\in\cD_2$ for any $\eta\in\mathcal{P}$. The domain $\cD_2\subset\cH$ was defined in Sec.~\ref{sec:domains}. The unbounded operators introduced below are meant to be defined on the domain $\hat\cD_{\rout/\rin}(\varrho)$.

Now we introduce various operators acting in the asymptotic Hilbert spaces. The scattering matrix and the interacting fields are uniquely defined by
\begin{equation}
\begin{gathered}
 \hat S(\varrho):\,\hat{\cH}_{\rin}(\varrho)\to\hat{\cH}_{\rout}(\varrho),
 \quad
 (\hat S(\varrho)\hat\Psi)(\eta):=S_\rmod(\eta)\hat\Psi(\eta),
 \\
 \hat C_{\ret/\adv}(h):\,\hat{\cH}_{\rin/\rout}(\varrho)\to\hat{\cH}_{\rin/\rout}(\varrho),
 \quad
 (\hat C_{\ret/\adv}(h)\hat\Psi)(\eta):=C_{\rmod,\ret/\adv}(\eta;h)\hat\Psi(\eta),
\end{gathered} 
\end{equation}
where $S_\rmod(\eta)$ and $C_{\rmod,\ret/\adv}(\eta;h)$ are given by Eqs.~\eqref{eq:dfn_S_op_mod_scalar}. In similar way we define the representation of the Poincar{\'e} group acting in the spaces $\hat\cH_{\rout/\rin}$,
\begin{equation}
 (\hat U(a)\hat\Psi)(\eta):=U_\rmod(\eta;a)\hat\Psi(\eta),
 \quad 
 \hat U(a)=\exp(\ri \hat P(\varrho)\cdot a),
\end{equation}
where $U_\rmod(\eta;a)$ was introduced in Def.~\ref{def:translation_scalar} and~\ref{def:translation_qed} in the case of the scalar model and QED, respectively. 

As follows from results of Sec.~\ref{sec:fields}, the long-range tail of the massless fields in completely determined by $\varrho$. Moreover, the photon creators and annihilators can be obtained with the use LSZ limits. They are given by
\begin{equation}
 \hat a^\#(f):\,\hat\cH_{\rout/\rin}(\varrho)\to \hat\cH_{\rout/\rin}(\varrho),
 \quad
 (\hat a^\#(f)\hat\Psi)(\eta) := a^\#(\eta;f)\hat\Psi(\eta),
\end{equation}
where $f\in\cS(\R^4)$ and $a(\eta;f)=\int\mHO{k}\,\overline{f(k)}\,a(\eta;k)$ with $a(\eta;k)$ defined by Eqs.~\eqref{eq:asymp_a_scalar} and~\eqref{eq:asymp_a_qed_phys}. In the case of QED the photon creators and annihilators depend on the polarization index $s=1,2$ which was omitted in the above definition.

According to our findings in the Sec.~\ref{sec:fields} and~\ref{sec:energy_momentum} there are no asymptotic states with just one electron (and no photons). Nevertheless, we can define the following operators
\begin{equation}
 \hat b^\#(\eta,\varrho_1;f):\,\hat\cH_{\rout/\rin}(\varrho)\to \hat\cH_{\rout/\rin}(\varrho\mp \varrho_1),
 \quad
 (\hat b^\#(\eta,\varrho_1;f)\hat\Psi)(\eta) := b^\#(p)\hat\Psi(\eta),
\end{equation}
which we call the electron creators and annihilators. In the above formula $f\in L^2(\Hm,\rd\mu_\mass)$, $\eta$ is a profile, $\varrho_1$ is a sector measure and $b^\#(f)$ stand for the standard electron annihilator or creator in the Fock space $\cH$. The operators $\hat b^\#(\eta,\varrho_1;f)$, $\eta\in\mathcal{P}$, are defined uniquely because any $\hat\Psi\in\hat\cH_{\rout/\rin}(\varrho)$ is determined uniquely by $\hat\Psi(\eta')\in\cH$ with arbitrary $\eta'\in\mathcal{P}$. In the case of QED we define in a similar way the positron creators and annihilators $\hat d^\#(\eta,\varrho_1;f)$ and assume that $\varrho_1(\Hm)=1$ for electrons and $\varrho_1(\Hm)=-1$ for positrons (there is no such constraint in the scalar model). The electron and positron annihilators in QED depend on the spin index $\sigma=1,2$ which was omitted. In the case of the scalar model the operators $\hat b^\#(\eta;f):=\hat b^\#(\eta,\varrho_1=0;f)$ act within one sector. Otherwise, they interpolate between sectors with different long-range tails of the massless scalar field. In QED the operators $\hat b^\#(\eta,\varrho_1;f)$ always interpolate between different sectors.

Summing up, we defined the following objects
\begin{equation}
 \hat{\cH}_{\rin/\rout}(\varrho),\quad 
 \hat S,\quad
 \hat C_{\ret/\adv}(h),\quad
 \hat U(a)=\re^{\ri\hat P\cdot a},\quad
 \hat a^\#(f),\quad
 \hat b^\#(\eta,\varrho_1;f).
\end{equation}
All of the operators depend on the sector $\hat{\cH}_{\rin/\rout}(\varrho)$ in which they act. The operators $\hat a^\#(f)$ satisfy the standard commutation relation. Similarly, the operators $\hat b^\#(\eta,\varrho_1;f)$ with fixed $\eta\in\mathcal{P}$ and $\varrho_1\in\mathcal{Q}$ satisfy the standard (anti-)commutation relation. However, as follows from Eqs.~\eqref{eq:asymp_a_scalar} and~\eqref{eq:asymp_a_qed_phys}, the operators $\hat a^\#(f)$ and $\hat b^\#(\eta,\varrho_1;f)$ do not commute. Indeed, we have
\begin{equation}
 [\hat a(k),\hat b(\eta;p)] = \frac{e\,\F{\eta}(k)}{2p\cdot k} \, \hat a(k),
 \quad
 [\hat a(k),\hat b^*(\eta;p)] = -\frac{e\,\F{\eta}(k)}{2p\cdot k} \, \hat a(k)
\end{equation}
in the case of the scalar model with $\varrho_1=0$ and
\begin{equation}
 [\hat a(s,k),\hat b^\#(\eta,\varrho_1;p)] = \pm e\,\F{\eta}(k)\, \hat a(s,k) 
 \int_{\Hm}(\rd\varrho_p(p')-\rd\varrho_1(p'))\,
 \frac{\varepsilon(s,k)\cdot p'}{p'\cdot k}
\end{equation}
in the case of QED, where $\varrho_p$ is the Dirac measure at $p\in\Hm$. Note that asymptotic creation and annihilation operators fulfilling similar commutation relation appear also in the model of QED investigated by Morchio and Strocchi~
\cite{morchio2016infrared}. We recall that the spectrum of the energy-momentum operators $\hat P^\mu$ is analysed in Sec.~\ref{sec:energy_momentum}. One shows that the operators $\hat a^\#(k)$ have energy momentum transfer on the light-cone,
\begin{equation}
 \hat U(a)\,a^*(k)\,\hat U(a)^{-1} = \re^{\ri k\cdot a}\,a^*(k),
 \quad
 \hat U(a)\,a(k)\,\hat U(a)^{-1} = \re^{-\ri k\cdot a}\,a(k).
\end{equation}
The energy-momentum transfer of the operators $\hat b^\#(\eta,p)$ is not sharp. The operators $a^\#(k)$ add or subtract one photon from a state on which they act. The operators $\hat b^\#(\eta,\varrho_1,p)$ add or subtract one electron together with an irremovable cloud of photons accompanying it. The vector $\hat\Omega\in\hat\cH_{\rout/\rin}(\varrho=0)$ is the translationally-invariant vacuum state. The vectors $\hat a^*(k)\hat\Omega\in \hat\cH_{\rout/\rin}(\varrho=0)$ describes a state with one photon of momentum $k$. The vector $\hat b^*(\eta,\varrho_1,p)\hat\Omega\in\hat\cH_{\rout/\rin}(\varrho_1)$ describes a state with one massive electron of momentum $p$ and a coherent cloud of photons depending on $p$, $\eta$, $\varrho_1$. In particular, the state $\hat b^*(\eta,\varrho_1,p)\hat\Omega$ have the energy-momentum content on and above the mass hyperboloid and is an improper eigenvector of the photon annihilation operator. For example, in the case of the scalar model with $\varrho_1=0$ and $\hat\Omega\in\cH_{\rout/\rin}(\varrho=0)$ it holds
\begin{equation}
 \hat a(k)\hat b^*(\eta,p)\hat\Omega = -\frac{e\,\eta(k)}{p\cdot k}\hat b^*(\eta,p)\hat\Omega.
\end{equation}

\subsection{Cross sections}\label{sec:cross_section}

In this section we recall the textbook construction of the standard differential cross sections which is applicable to models without long-range interactions. We follow~\cite{fredenhagenquantenfeldtheorie,peskin1995introduction} and for simplicity consider an interacting QFT containing only scalar particles of mass $\mass$. The definition of the cross section involves several idealization regarding the incoming state. First, since by clustering scattering events that involve more than two incoming particles are typically suppressed one usually restricts attention to states with two incoming particles. Second, one presumes that the target particle and the particle in the incident beam are uncorrelated and have some fixed sharp and distinct momenta $\underline{p}_1$ and $\underline{p}_2$, respectively. Third, one supposes that the distribution of the impact parameters of the incoming particle is uniform in the plane perpendicular to $\underline{p}_1$ and $\underline{p}_2$ (in the region of the scattering). Let $e_\mu$, $\mu=0,1,2,3$, be an orthonormal basis such that $\underline{p}_1\cdot e_1=\underline{p}_1\cdot e_2=0$, $\underline{p}_2\cdot e_1=\underline{p}_2\cdot e_2=0$. In order to construct the incoming state $\omega_{\underline{p}_1,\underline{p}_2}(\cdot)$ that satisfies the above assumptions consider the following vector in the Fock space
\begin{equation}
 \Psi_{\lambda,b,\underline{p}_1,\underline{p}_2} = b^*(f_{\lambda,0,\underline{p}_1})b^*(f_{\lambda,b,\underline{p}_2})\Omega,
\end{equation}
where $\lambda\in(0,1)$, $b\in\R^4$ is the impact parameter, $\Omega$ is the vacuum state, $b^*(f)$, $f\in L^2(\Hm,\rd\mu_\mass)$, is the creation operator and 
\begin{equation}\label{eq:cross_f}
 f_{\lambda,b,\underline{p}}(p) = (2E_\mass(\vec p))^{1/2} \lambda^{-3/2} h((\vec p-\vec{\underline{p}})/\lambda) \re^{\ri p\cdot b}
\end{equation}
for some $h\in C^\infty_{\mathrm{c}}(\R^3)$ such that 
\begin{equation}
 \int\frac{\rd^3 \vec p}{(2\pi)^3}\,|h(\vec p)|^2 = 1.
\end{equation}
Note that $\|\Psi_{\lambda,b,\underline{p}_1,\underline{p}_2}\|=1$. Assume that $K$ is a translationally-invariant observable, $K=U_\fr(a)KU_\fr(-a)$, such that
\begin{equation}\label{eq:K_form}
 (p'_1,p'_2|K|p''_1,p''_2) = (2\pi)^4 \delta(p'_1+p'_2-p''_1-p''_2)\,\mathcal{K}(p'_1,p'_2,p''_1,p''_2)
\end{equation}
with $\mathcal{K}$ continuous in some neighbourhood of $(\underline{p}_1,\underline{p}_2,\underline{p}_1,\underline{p}_2)\in\Hm^{\times 4}$. The incoming state is defined by the following limit
\begin{equation}\label{eq:cross_section_incoming_state}
 \omega_{\underline{p}_1,\underline{p}_2}(K)
 :=
 \lim_{\lambda\searrow0}\lim_{r\to\infty}\,
 \omega_{r,\lambda,\underline{p}_1,\underline{p}_2}(K)
 =\frac{\mathcal{K}(\underline{p}_1,\underline{p}_2,\underline{p}_1,\underline{p}_2)}{4E_\mass(\vec{\underline{p}}_1)E_\mass(\vec{\underline{p}}_2)\left|\frac{\vec{\underline{p}}_1}{E_\mass(\vec{\underline{p}}_1)} - \frac{\vec{\underline{p}}_2}{E_\mass(\vec{\underline{p}}_2)}\right|},
\end{equation}
where
\begin{equation}
 \omega_{r,\lambda,\underline{p}_1,\underline{p}_2}(K) 
 :=
 \int \rd\vartheta_r(b)\, (\Psi_{\lambda,b,\underline{p}_1,\underline{p}_2}|K\Psi_{\lambda,b,\underline{p}_1,\underline{p}_2})
\end{equation}
and the measure
\begin{equation}\label{eq:b_measure}
 \rd\vartheta_r(b) := \rd^4b\,\delta(b\cdot e_0)\delta(b\cdot e_3)\,\theta(r^2+b^2)
\end{equation}
describes the uniform distribution of impact parameters $b$ of the incident particle in the disk of radius $r$ in the plane perpendicular to $\underline{p}_1,\underline{p}_2$. In order to show the last equality in Eq.~\eqref{eq:cross_section_incoming_state} one uses the following identity
\begin{multline}
 \lim_{r\to\infty}\,
 \omega_{r,\lambda,\underline{p}_1,\underline{p}_2}(K) 
 =
 \int\mHm{p'_1}\mHm{p'_2}\mHm{p''_1}\mHm{p''_2}\,
 (2\pi)^4\delta(p'_1+p'_2-p''_1-p''_2)\,
 \\\times 
 (2\pi)\delta((p'_2-p''_2)\cdot e_1)\,
 (2\pi)\delta((p'_2-p''_2)\cdot e_2)\,
 |f_{\lambda,0,\underline{p}_1}(p'_1)|^2|f_{\lambda,0,\underline{p}_2}(p'_2)|^2 \,\mathcal{K}(p'_1,p'_2,p''_1,p''_2).
\end{multline}
Observe that the incoming state $\omega_{\underline{p}_1,\underline{p}_2}$ is not well-defined for $K=\id$ which violates the assumption about $K$ stated above. In particular, $\omega_{\underline{p}_1,\underline{p}_2}(\cdot)$ is not a state in the standard algebraic QFT sense~\cite{haag2012local} and perhaps it would be more appropriate to call it a weight.

The standard differential cross section is defined by the following formula
\begin{equation}
 \rd\sigma_{\underline{p}_1,\underline{p}_2}(p_1,\ldots,p_n) =  \omega_{\underline{p}_1,\underline{p}_2}(S^* \rd K(p_1,\ldots,p_n) S),
\end{equation}
where
\begin{equation}\label{eq:observable_dK}
 \rd K(p_1,\ldots,p_n)=\mHm{p_1}\ldots\mHm{p_n}~b^*(p_1)\ldots b^*(p_n)|\Omega)(\Omega|b(p_1)\ldots b(p_n)
\end{equation}
and $S$ is the standard scattering matrix which is well defined because of the assumption that the model under consideration contains only massive particles. Note that
\begin{equation}\label{eq:inv_matrix_elem}
 (p_1,\ldots,p_n|S|\underline{p}_1,\underline{p}_2)
 =
 (2\pi)^4 \delta(p_1+\ldots+p_n-\underline{p}_1-\underline{p}_2)\,
 \mathcal{S}(p_1,\ldots,p_n;\underline{p}_1,\underline{p}_2),
\end{equation}
where $\mathcal{S}$ is the invariant matrix element. It is typically continuous for non-exceptional momenta (i.e. for non-coinciding momenta and outside the thresholds). Using the above notation we arrive at the textbook formula for the differential cross section
\begin{multline}\label{eq:cross_section_std}
 \rd\sigma_{\underline{p}_1,\underline{p}_2}(p_1,\ldots,p_n)  = \mHm{p_1}\ldots\mHm{p_n}
 \\
 (2\pi)^4 \delta(p_1+\ldots+p_n-\underline{p}_1-\underline{p}_2)\,
 \frac{|\mathcal{S}(p_1,\ldots,p_n;\underline{p}_1,\underline{p}_2)|^2}{4E_\mass(\vec{\underline{p}}_1)E_\mass(\vec{\underline{p}}_2)\left|\frac{\vec{\underline{p}}_1}{E_\mass(\vec{\underline{p}}_1)} - \frac{\vec{\underline{p}}_2}{E_\mass(\vec{\underline{p}}_2)}\right|}.
\end{multline}
For fixed $\underline{p}_1,\underline{p}_2\in\Hm$ the differential cross section $\rd\sigma_{\underline{p}_1,\underline{p}_2}$ is a measure on a subset of $\Hm^{\times n}$. Unfortunately, we are not aware of any reference with a rigorous construction of the differential cross section in perturbation theory even in purely massive models (the main difficulty is to prove that the observable $S^*\rd K(p_1,\ldots,p_n) S$ is in the domain of definition of the incoming state $\omega_{\underline{p}_1,\underline{p}_2}$).

\subsection{Inclusive cross sections}\label{sec:inclusive_cross_section}

In this section we show how to obtain the inclusive differential cross section using the scattering matrix defined in the paper. Our discussion is formal. Nevertheless, we believe it demonstrates the compatibility of our construction with the standard procedure used in practical computations.

Let us first recall the standard method~\cite{yennie1961infrared,weinberg1965infrared,weinberg1995quantum,peskin1995introduction} of computing the inclusive cross sections in QED. For concreteness, we assume that the theory is regularized with the use of a positive photon mass $\epsilon$. The scattering matrix in such a theory is denoted by $S_\epsilon$ and is well-defined as long as $\epsilon>0$. The inclusive cross section is defined in the following way
\begin{equation}\label{eq:inclusive_std}
 \rd\sigma^{\mathrm{incl}}_{\underline{p}_1,\underline{p}_2}(E;p_1,\ldots,p_n):=\lim_{\epsilon\searrow 0}
 \lim_{\substack{r\to\infty\\\lambda\searrow0}}\,
 \omega_{r,\lambda,\underline{p}_1,\underline{p}_2}(S^*_{\epsilon}\rd K_E(p_1,\ldots,p_n)S_{\epsilon}),
\end{equation}
where the state $\omega_{r,\lambda,\underline{p}_1,\underline{p}_2}$ was introduced in Sec.~\ref{sec:cross_section}, $E\in(0,\mass)$ is a fixed threshold ($\mass$ is the mass of the electron) and
\begin{equation}\label{eq:observable_dK_E}
 \rd K_E(p_1,\ldots,p_n):=\mHm{p_1}\ldots\mHm{p_n}~
 b^*(p_1)\ldots b^*(p_n)1_{[0,E]}(P^0_\fr)b(p_1)\ldots b(p_n)
\end{equation} 
is the observable that corresponds to detecting $n$ massive particles with sharp momenta $p_1,\ldots,p_n$ by a detector with the energy sensitivity $E$ (one could easily consider more general observables describing a detection of massless particles with energies larger than the threshold). In the above formula, $P^0_\fr$ is the standard energy operator in the Fock space, $1_{[0,E]}$ is the characteristic function of the interval $[0,E]$ and $b^\#(p)$ are the standard creation and annihilation operators of massive particles. Note that the observable~\eqref{eq:observable_dK_E} is translationally invariant and in the case $E=0$ the observables~\eqref{eq:observable_dK_E} and~\eqref{eq:observable_dK} coincide. It is easy to see that the above formula for the inclusive cross section is equivalent to the formula~\eqref{eq:inclusive_overview}. 
According to the formula~\eqref{eq:inclusive_std} in order to determine the inclusive cross section in QED one first computes the inclusive cross section in QED with massive photons and subsequently takes the massless limit. It is crucial that the limit $\epsilon\searrow 0$ is taken after the limits $r\to\infty$, $\lambda\searrow0$. The scattering matrix $S_\epsilon$ diverges in the limit $\epsilon\searrow 0$ and the existence of the limit $\epsilon\searrow 0$ in~\eqref{eq:inclusive_std} relies on cancellations of IR divergences. It is well-known~\cite{yennie1961infrared,weinberg1965infrared} that such cancellation can occur only if the momenta of all incoming and outgoing particles are sharp. The latter condition is satisfied only in the limit $r\to\infty$, $\lambda\searrow0$. Consequently, the formula~\eqref{eq:inclusive_std} does not seam to have any physical interpretation in QED with no IR cutoff. 

Now we discuss the construction of the inclusive cross section in our framework. As we argued in Sec.~\ref{sec:physical_interpretation} the improper electron creation and annihilation operators $b^\#(\eta,\varrho_1;p)$ add or subtract an electron with the four-momentum $p$ together with a cloud of photons that is correlated with the momentum $p$ and depends on the profile $\eta$ and the sector measure $\varrho_1$. Note that if we choose $\varrho_1=\varrho_{\underline{p}}$ to be a Dirac measure at some reference four-momentum $\underline{p}$, then the cloud disappears if the electron momentum $p$ coincides with the reference momentum $\underline{p}$. Thus, in order to construct the state describing the incoming particles with no radiation we are led to consider the following family of normalized vectors 
\begin{equation}
 \hat \Psi_{\lambda,b,\underline{p}_1,\underline{p}_2} = \hat b^*(\eta,\varrho_{\underline{p}_1};f_{\lambda,0,\underline{p}_1})\hat b^*(\eta,\varrho_{\underline{p}_2};f_{\lambda,b,\underline{p}_2})\hat \Omega
 \in \hat\cH_{\rin}(\varrho_{\underline{p}_1}+\varrho_{\underline{p}_2}),
\end{equation}
where $\hat\Omega\in\cH_\rin(0)$ and $f_{\lambda,b,\underline{p}}$ is given by Eq.~\eqref{eq:cross_f}. The above vector describes two electrons (one could also consider an incoming state with one electron and one photon or an incoming state with two photons). In line with the definitions given in Sec.~\ref{sec:state_space} we have
\begin{equation}
 \hat \Psi_{\lambda,b,\underline{p}_1,\underline{p}_2}(\eta) = b^*(f_{\lambda,0,\underline{p}_1})\hat b^*(f_{\lambda,b,\underline{p}_2})\hat \Omega
 \in \cH,
\end{equation}
where $\cH$ is the standard Fock space. The incoming state is defined in the following way
\begin{equation}\label{eq:inclusive_in}
 \omega_{\underline{p}_1,\underline{p}_2}(\hat K) 
 =
 \lim_{\substack{r\to\infty\\\lambda\searrow0}}
 \int \rd\vartheta_r(b)\, (\hat\Psi_{\lambda,b,\underline{p}_1,\underline{p}_2}|\hat K\hat\Psi_{\lambda,b,\underline{p}_1,\underline{p}_2})
 =
 \frac{\hat{\mathcal{K}}(\underline{p}_1,\underline{p}_2,\underline{p}_1,\underline{p}_2)}{4E_\mass(\vec{\underline{p}}_1)E_\mass(\vec{\underline{p}}_2)\left|\frac{\vec{\underline{p}}_1}{E_\mass(\vec{\underline{p}}_1)} - \frac{\vec{\underline{p}}_2}{E_\mass(\vec{\underline{p}}_2)}\right|},
\end{equation}
where the distribution of the impact parameters $\rd\vartheta_r(b)$ is given by Eq.~\eqref{eq:b_measure}. The above limit exists for sufficiently regular translationally invariant observables~$\hat K$ such that
\begin{multline}
 (2\pi)\delta((p'_2-\underline{p}_2)\cdot e_1)\,
 (2\pi)\delta((p'_2-\underline{p}_2)\cdot e_2)~
 (p'_1,p'_2|\hat K|\underline{p}_1,\underline{p}_2) 
 \\
 =
 (2\pi)\delta((p'_2-\underline{p}_2)\cdot e_1)\,
 (2\pi)\delta((p'_2-\underline{p}_2)\cdot e_2)~
 (2\pi)^4 \delta(p'_1+p'_2-\underline{p}_1-\underline{p}_2)
 ~\hat{\mathcal{K}}(\underline{p}_1,\underline{p}_2,\underline{p}_1,\underline{p}_2).
\end{multline}

Now we define the translationally invariant observable that describes detecting of hard particles with sharp momenta $p_1,\ldots,p_n$ by a detector with the sensitivity $E>0$. First, we note that the improper electron annihilator or creator $b^\#(\eta,\varrho_1;p)$ with fixed four-momentum $p$ can be interpreted as a form on $\hat\cD_{\rout/\rin}(\varrho\mp\varrho_1)^*\times \hat\cD_{\rout/\rin}(\varrho)$. It is easy to see that formally $\hat b^\#(p):=b^\#(\eta,\varrho_p;p)$ creates or annihilates just an electron (without changing the photon content of the state). Observe that for $\hat\Psi\in\hat\cH_{\rout/\rin}(\varrho)$ and $\Psi'\in \hat\cH_{\rout/\rin}(\varrho\mp\varrho_p)$ it holds
\begin{equation}
 (\hat\Psi|\hat b^*(p)\hat\Psi)
 =(\Hat\Psi'(\eta)|b^*(p)\hat\Psi'(\eta)),
 \quad (\hat\Psi|\hat b(p)\hat\Psi)
 =(\Hat\Psi'(\eta)|b(p)\hat\Psi'(\eta)),
\end{equation}
where $b^\#(p)$ are the standard annihilation and creation operators. One verifies that the forms $\hat b^\#(p)$ have sharp energy-momentum transfer on the mass hyperboloid
\begin{equation}
\begin{aligned}
 \hat U(\rho+\rho_p;a)\hat b^*(p) \hat U(\rho;a)^{-1} 
 =\re^{\ri p\cdot a}\, \hat b^*(p),
 \quad 
 \hat U(\rho-\rho_p;a)\hat b(p) \hat U(\rho;a)^{-1} =\re^{-\ri p\cdot a}\,  \hat b(p), 
\end{aligned} 
\end{equation}
where $\hat U(a)\equiv \hat U(\rho;a)$ is the translation operators. In similar manner we define
\begin{equation}\label{eq:observable_dK_E_hat}
 \rd \hat K_E(p_1,\ldots,p_n)
 :=
 \mHm{p_1}\ldots\mHm{p_n}
 \,\hat b^*(p_1) \ldots \hat b^*(p_n) 
 1_{[0,E]}(\hat P^0) \hat b(p_1) \ldots \hat b(p_n),
\end{equation}
where $1_{[0,E]}$ is the characteristic function of the interval $[0,E]$ and $E>0$. We interpret the above expression as a form $\hat\cD_{\rout/\rin}(\varrho)^*\times \hat\cD_{\rout/\rin}(\varrho)$ which takes values in measures on $\Hm^{\times n}$. Observe that it holds
\begin{equation}
\begin{aligned}
 \hat U(a)\,\rd \hat K_E(p_1,\ldots,p_n)\, \hat U(a)^{-1}
 &=\rd \hat K_E(p_1,\ldots,p_n),
 \\
 \hat U(a)\,\hat S^* \rd \hat K_E(p_1,\ldots,p_n) \hat S\, \hat U(a)^{-1}
 &=\hat S^*\rd \hat K_E(p_1,\ldots,p_n) \hat S.
\end{aligned} 
\end{equation}

The inclusive cross section in our framework is defined by the following formula
\begin{equation}\label{eq:inc}
 \rd\sigma_{\underline{p}_1,\underline{p}_2}(p_1,\ldots,p_n) =  \omega_{\underline{p}_1,\underline{p}_2}(\hat S^* \rd \hat K_E(p_1,\ldots,p_n)\hat S).
\end{equation}
Recall that $(\hat S\hat\Psi)(\eta)=S_\rmod(\eta)\hat\Psi(\eta)$ and $S_\rmod(\eta)$ is defined by the adiabatic limit of $S_\rmod(\eta,g_\epsilon) = R(\eta,g_\epsilon)S_\rout^\ras(\eta,g_\epsilon)S(g)S_\rin^\ras(\eta,g_\epsilon)R(\eta,g_\epsilon)^{-1}$, where $S_{\rout/\rin}^\ras(\eta,g)$ and $R(\eta,g)$ are the Dollard modifiers and the sector operator introduced in Sec.~\ref{sec:def_modified} and $S(g)$ is the standard Bogoliubov S-matrix. By the following two lemmas the inclusive cross section~\eqref{eq:inc} formally coincides with the inclusive cross section obtained with the use of the standard scattering matrix.
\begin{lem}
Let $K$ be a translationally invariant observable of the form~\eqref{eq:K_form} with a kernel $\mathcal{K}$ which is continuous in some neighbourhood of $\underline{p}_1,\underline{p}_2,\underline{p}_1,\underline{p}_2$. The limits
\begin{equation}
\begin{gathered}
 \lim_{\lambda\searrow0}\lim_{r\to\infty}\int \rd\vartheta_r(b)\,
 (\hat \Psi_{\lambda,b,\underline{p}_1,\underline{p}_2}(\eta)|K\hat \Psi_{\lambda,b,\underline{p}_1,\underline{p}_2}(\eta)),
 \\
 \lim_{\lambda\searrow0}\lim_{r\to\infty}\int \rd\vartheta_r(b)\,
 (\hat \Psi_{\lambda,b,\underline{p}_1,\underline{p}_2}(\eta)| S_\rin^\ras(\eta,g_\epsilon)R(\eta,g_\epsilon)^{-1} K R(\eta,g_\epsilon)S_\rin^\ras(\eta,g_\epsilon)^{-1} \hat \Psi_{\lambda,b,\underline{p}_1,\underline{p}_2}(\eta))
\end{gathered}
\end{equation}
exist and coincide.
\end{lem}
\begin{lem}
For every threshold $E\in(0,\mass)$, profile $\eta$ and $\hat\Psi,\hat\Psi'\in\hat \cD_{\rout/\rin}(\varrho)$ it holds 
\begin{multline}
 (\hat\Psi|\rd \hat K_E(p_1,\ldots,p_n)\hat\Psi')
 \\
 =
 \lim_{\epsilon\searrow 0}\,
 (\hat\Psi(\eta)|
 R(\eta,g_\epsilon) S_\rout^\ras(\eta,g_\epsilon)
 \rd K_E(p_1,\ldots,p_n)
 S_\rout^\ras(\eta,g_\epsilon)^{-1}
 R(\eta,g_\epsilon)^{-1}
 \hat\Psi'(\eta)),
\end{multline}
where $\rd K_E$ and $\rd\hat K_E$ are defined by Eqs.~\eqref{eq:observable_dK_E} and~\eqref{eq:observable_dK_E_hat}.
\end{lem}
To prove the above lemmas one uses the properties of the Dollard modifiers and the sector operator stated in Sec.~\ref{sec:def_modified}.

\section{Summary and outlook}\label{sec:summary}

We investigated the infrared structure of models of perturbative QFT with long range interactions. We proposed a rigorous perturbative method of constructing the Hilbert space of asymptotic states, the scattering matrix and the retarded and advanced interacting fields in quantum electrodynamics and a model of interacting scalar fields with long-range interactions. We considered a large class of physically-relevant superselection sectors. Our method uses the modified scattering theory~\cite{dollard1964asymptotic} originally developed in quantum mechanics and is inspired by the works of Bogoliubov~\cite{bogoliubov1959introduction}, Kulish and Faddeev~\cite{kulish1970asymptotic} and Morchio and Strocchi~\cite{morchio2016dynamics,morchio2016infrared}. One of the nice features of our construction is a clear separation between the ultraviolet and infrared problem. We do not introduce any ultraviolet regularization and solve the ultraviolet problem using the Epstein-Glaser renormalization technique. Our modified scattering matrix and the modified interacting fields are defined with the use of the adiabatic limit. The existence of this limit is well motivated on physical grounds and was proved in the paper for corrections of low orders. We investigated the physical properties of our construction and the relation to the approach used in practical computations. We proved the translational covariance of the construction and determined the spectrum of the energy-momentum operators. It is likely that in the case of the scalar model the construction is also covariant with respect to Lorentz transformations. In the case of QED the Lorentz covariance is well-known to be broken in sectors with nonzero total electric charge~\cite{frohlich1979infrared,frohlich1979charged}. 

We expect that our method can be generalized to other models of perturbative quantum field theory with interaction vertices containing two massive and one massless field. Examples of such models are the scalar and pseudo-scalar Yukawa theories with a massive Dirac field coupled to a massless scalar or pseudo-scalar field, respectively. In fact, in the case of the later model it is plausible that the adiabatic limit of the standard Bogoliubov scattering matrix and interacting fields exists and no modifications are necessary. Another interesting open problem is the construction of charged interacting fields in QED, such as the interacting Dirac field, which are inevitably nonlocal. It would be also interesting to investigate whether one can express elements of our modified scattering matrix in terms of the Green functions with the use of an appropriately generalized LSZ formula.

\section*{Acknowledgements}

The financial support of the National Science Center, Poland, under the grant UMO-2017/25/N/ST2/01012 is gratefully acknowledged. I wish to thank Wojciech Dybalski and Andrzej Herdegen for helpful discussions and useful comments. I am also grateful to Michael Dütsch, Klaus Fredenhagen, José Gracia-Bondía and Stefan Hollands for helpful remarks and to Markus Fröb for hints on literature.

\appendix

\section{Value of a distribution at a point}\label{sec:value_dist}

\begin{dfn}\label{def:lojasiewicz}
A~distribution $t\in\cS'(\R^N)$ has a~value $t(0)\in\C$ at the origin iff the limit
\begin{equation}
 t(0):=\lim_{\epsilon\searrow0}\int \frac{\rd^N q}{(2\pi)^N}\, t(q) g_\epsilon(q)
\end{equation}
exists for any $g\in\cS(\R^N)$ such that $\int \frac{\rd^N q}{(2\pi)^N} g(q)=1$, $g_\epsilon(q)=\epsilon^{-N} g(q/\epsilon)$. 
\end{dfn}
The above definition of the value of a distribution at a point is usually attributed to {\L}ojasiewicz~\cite{lojasiewicz1957valeur}. The value $t(0)$ is also called the adiabatic limit of the distribution $t$ at $0$. 

\begin{thm}\label{thm:dist_value}
Let $h$ be a smooth function of at most polynomial growth. If a distribution $t\in\cS'(\R^N)$ has a~value $t(0)\in\C$ at the origin, then the distribution $h(q)t(q)$ has a value $h(0)t(0)$ at the origin.
\end{thm}
\begin{proof}
If $t\in\cS'(\R^N)$, then there exist $M\in\N_0$ such that
\begin{equation}
 \left|\int \frac{\rd^N q}{(2\pi)^N}\, t(q) f(q) \right| \leq
 \const \sum_{\substack{\alpha,\beta\\|\alpha|+|\beta|\leq M}}
 \sup_{q\in\R^N}\, \left| q^\alpha \partial^\beta f(q) \right|.
\end{equation}
For $K\in\N_0$ set
\begin{equation}
 h_K(q) = h(q) - \sum_{\substack{\alpha\\|\alpha|\leq K}} \frac{q^\alpha}{\alpha!} (\partial^\alpha h)(0).
\end{equation}
It holds
\begin{multline}
 \sup_{q\in\R^N}\,\left|q^\alpha\partial^\beta (h_K(q) g_\epsilon(q))\right|
 =
 \epsilon^{|\alpha|-|\beta|-N} \sup_{q\in\R^N}\,\left|q^\alpha\partial^\beta (h_K(\epsilon q) g(q))\right|
 \leq \const\,\epsilon^{|\alpha|-|\beta|-N+K},
\end{multline}
where the constant is independent of $\epsilon$. Let $K=M+N+1$. The distribution $h_K(q)t(q)$ has the value zero at the origin. Thus,
\begin{equation}
 \lim_{\epsilon\searrow0}\int \frac{\rd^N q}{(2\pi)^N}\, h(q)t(q) g_\epsilon(q)
 =
 \sum_{\substack{\alpha\\|\alpha|\leq K}}
 \frac{1}{\alpha!} (\partial^\alpha h)(0)~
 \lim_{\epsilon\searrow0}\int \frac{\rd^N q}{(2\pi)^N}\, t(q) q^\alpha g_\epsilon(q)
\end{equation}
It holds
\begin{equation}\label{eq:thm_dist_value}
 \int \frac{\rd^N q}{(2\pi)^N}\, t(q) q^\alpha g_\epsilon(q) =
 \epsilon^{|\alpha|} \int \frac{\rd^N q}{(2\pi)^N}\, t(q) g^\alpha_\epsilon(q), 
\end{equation}
where $g^\alpha(q):=q^\alpha g(q)$ and $g^\alpha_\epsilon(q):=\epsilon^{-N} g^\alpha(q/\epsilon)$. Since by the assumption the distribution~$t$ has a value at zero, the limit $\epsilon\searrow0$ of the expression~\eqref{eq:thm_dist_value} vanishes if $|\alpha|\geq1$ and is equal to $t(0)$ if $\alpha=0$. This finishes the proof.
\end{proof}

\section{BCH formula and Magnus expansion}\label{sec:BCH_Magnus}

Let $A$ and $B$ be operators such that $[A,[A,B]]=[B,[A,B]]=0$. The following identity is known as the BCH formula
\begin{equation}\label{eq:BCH}
 \exp\left( A + B\right) = \exp(A) \exp(B) \exp\left(-\frac{1}{2}[A,B]\right).
\end{equation}
Let $t\mapsto A(t)$ a be one-parameter family of operators such that $[A(t),[A(t'),A(t'')]]=0$. The following identity is a special case of the Magnus expansion~\cite{blanes2009magnus}
\begin{multline}\label{eq:magnus}
 \aTexp\left(-\ri \int \rd t \, A(t)\right)
 \\
 =\exp\left(-\ri \int \rd t \, A(t)\right)
 \exp\left( \frac{1}{2}\int \rd t_1\int \rd t_2 \, \theta(t_1-t_2)  [A(t_1),A(t_2)] \right),
\end{multline}
where $\aTexp$ is the anti-time-ordered exponential.

\section{Long-range tail}\label{sec:long_rang_tail}

\begin{dfn}\cite{buchholz1986gauss}\label{dfn:long_range_tail}
Let $h\in C^\infty_{\mathrm{c}}(\R^4)$ be supported in the spacelike complement of the origin in the Minkowski space. For $R>0$ set $h_R(x):=h(x/R)$. Consider a field $B\in\cS'(\R^4,L(\cD))$ and assume that for some $n\in\N_+$ and all $\Psi,\Psi'\in\cD$ the limit
\begin{equation}
 \lim_{R\to\infty}\,R^n\, (\Psi'|B(h_R)\Psi) =: (\Psi'|B^n_{\mathrm{lr}}(h)\Psi) 
\end{equation}
exists and defines the field $B^n_{\mathrm{lr}}\in\cS'(\R^4,L(\cD))$. The field $B^n_{\mathrm{lr}}$ is homogeneous of degree $-n$ and is called the long-range tail of $B$. The long-range tail is a classical observable, i.e., it commutes with all fields that can be localized in any bounded region of the Minkowski space.
\end{dfn}

\section{Incoming and outgoing LSZ fields}\label{sec:LSZ}

\begin{dfn}
Let $f\in\cS(\R^4)$ be such that $\F{f}$ is supported outside the origin. Set
\begin{equation}
 f_{m,t}(\vec x):=\int\frac{\rd^3 \vec p}{(2\pi)^3\, 2E_m(\vec p)}\,
 \exp(-\ri E_m(\vec p)  t + \ri \vec p\cdot \vec x)
 \,\F{f}(E_m(\vec{p}),\vec p),
\end{equation}
where $E_m(\vec p):=\sqrt{|\vec p|^2+m^2}$. The LSZ limits of the field $B(x)$ are by definition the following limits
\begin{gather}
 \lim_{t\to\pm\infty}\,(-\ri) \int\rd^3\vec x~ f_{m,t}(\vec x)\,
 \overset{\leftrightarrow}{\partial_t}\,(\Psi|B(t,\vec{x})\Psi')
 =:(\Psi|B_m^{(+)}(f)\Psi'), 
 \\
 \lim_{t\to\pm\infty}\,\ri \int\rd^3\vec x~ \overline{f_{m,t}(\vec x)}\,
 \overset{\leftrightarrow}{\partial_t}\,(\Psi|B(t,\vec{x})\Psi')
 =:(\Psi|B_m^{(-)}(f)\Psi'),  
\end{gather}
where $\Psi,\Psi'\in\cD_0$ are arbitrary. If the above limits exist, they can be used to determine the components $B_m^{(\pm)}(f)$ of the field $B(x)$ which have the energy-momentum transfer contained in the forward/backward mass hyperboloid $\pm\Hm$. The operators $B_m^{(\pm)}(f)$ are responsible for the creation or annihilation of particles of mass $m$. 
\end{dfn}
For example, let $\psi(x)$ be the free scalar field of mass $\mass$ (cf. Equation~\eqref{eq:phi_def}). The limits
\begin{equation}
 \lim_{t\to\pm\infty}\,(-\ri) \int\rd^3\vec x~ f_{m,t}(\vec x)\,
 \overset{\leftrightarrow}{\partial_t}\,(\Psi|\psi(t,\vec{x})\Psi') 
\end{equation}
exist for all $\Psi,\Psi'\in\cD_0$ and are equal to $(\Psi|b^*(f)\Psi')$ if $m=\mass$ and vanish otherwise. The free field built with the use of the creation and annihilation operators obtained by taking the LSZ limits is denoted by $\psi_{\textrm{LSZ},\rout/\rin}(x)$ and called the outgoing or incoming LSZ field. In the case of the above example we have $\psi_{\textrm{LSZ},\rout/\rin}(x)=\psi(x)$.

\section{Propagators}\label{sec:propagators}

Let $\phi$ be the free real scalar field, $A_\mu$ be the free vector field in the Feynman gauge and $\psi$ be the free Dirac spinor field. The fields have mass $m\geq0$.
\begin{itemize}
 \item The commutator functions
 \begin{gather}
 [\phi(x),\phi(y)]=-\ri D_m(x-y),
 \quad
 [A_\mu(x),A_\nu(y)]=\ri g_{\mu\nu} D_m(x-y),
 \\
 [\psi_a(x),\overline{\psi}_b(y)] = -\ri S_{m,ab}(x-y):=
 -\ri(\ri \slashed{\partial}_x+m)_{ab} D_m(x-y),
 \end{gather}
\begin{equation}
 D_m(x) 
 := 
 \frac{\ri}{(2\pi)^3} \int \rd^4 k \, \delta(k^2-m^2) \sgn(k^0) \exp(-\ri k\cdot x),
\end{equation}
\begin{equation}
 D_0(x) = \frac{1}{2\pi} \sgn(x^0) \delta(x^2).
\end{equation}
\item The Wightman two point functions
 \begin{gather}
 (\Omega|\phi(x)\phi(y)\Omega)=-\ri D^{(+)}_m(x-y),
 \quad
 (\Omega|A_\mu(x)A_\nu(y)\Omega)=\ri g_{\mu\nu} D^{(+)}_m(x-y),
 \\
 (\Omega|\psi_a(x)\overline{\psi}_b(y)\Omega) = -\ri S^{(+)}_{m,ab}(x-y):=
 -\ri(\ri \slashed{\partial}_x+m)_{ab} D^{(+)}_m(x-y).
 \end{gather}
\begin{equation}
 D_m^{(+)}(x) 
 :=
 \frac{\ri}{(2\pi)^3} \int \rd^4 k \, \delta(k^2-m^2) \theta(k^0) \exp(-\ri k\cdot x) 
\end{equation}
\item The Feynman propagators
 \begin{gather}
 (\Omega|T(\phi(x),\phi(y))\Omega)=-\ri D^{F}_m(x-y),
 \\
 (\Omega|T(A_\mu(x),A_\nu(y))\Omega)=\ri g_{\mu\nu} D^{F}_m(x-y),
 \\
 (\Omega|\T(\psi_a(x),\overline{\psi}_b(y))\Omega) = -\ri S^{F}_{ab}(x-y):=-\ri(\ri \slashed{\partial}_x+m)_{ab} D^{F}_m(x-y),
 \end{gather}
\begin{equation}
 D^F(x):=D^{(+)}(x) + D^\adv(x) =
 \int \frac{\rd^4 k}{(2\pi)^4} \, \frac{\exp(-\ri k\cdot x)}{m^2-k^2-\ri \zerop}.
\end{equation}
\item The retarded and advanced Green functions
\begin{equation}
 D_m^{\ret/\adv}(x)
 := 
\int\frac{\rd^4 k}{(2\pi)^4}\, \frac{\exp(-\ri k\cdot x)}{m^2-k^2\mp\ri \zerop \sgn k^0},
 \quad
 D^{\ret/\adv}_0(x) =  \frac{1}{2\pi} \theta(\pm x^0) \delta(x^2).
\end{equation}
\item The Dirac propagators
\begin{equation}
 D_m^D(x) 
 := \frac{1}{2}\left( D_m^\ret(x) + D_m^\adv(x) \right),
 \quad
 D_0^D(x) = \frac{1}{4\pi} \delta(x^2).
\end{equation}
\end{itemize}

\end{document}